\documentclass[a4paper,12pt]{article}
\usepackage[a4paper,top=3.5cm,bottom=3cm,left=3cm,right=3cm]{geometry}
\usepackage{appendix}
\usepackage{authblk}

\usepackage[numbers,sort&compress]{natbib}
\usepackage{setspace}

\usepackage[utf8]{inputenc} % allow utf-8 input
\usepackage[T1]{fontenc} % use 8-bit T1 fonts
\usepackage[hypertexnames=false]{hyperref} % hyperlinks with unique destinations
\usepackage{url} % simple URL typesetting
\usepackage{booktabs} % professional-quality tables
\usepackage{amsfonts,amsmath,amssymb} % blackboard math symbols
\usepackage{amsthm}
\usepackage{nicefrac} % compact symbols for 1/2, etc.
\usepackage{microtype} % microtypography
\usepackage{algorithm,algorithmic}
\usepackage{graphicx}
\graphicspath{{./}}
\makeatletter\def\input@path{{./}}\makeatother
\usepackage{pdflscape}
\usepackage{longtable}
\usepackage{array}
\usepackage{multirow}

\usepackage{latexsym}

\usepackage{color}

\def\bs{\boldsymbol}

\def\doubleR{\mathbb{R}}

\def\ci{\perp\!\!\!\!\perp}

\def\distiid{\overset{iid}{\sim}}

\newtheorem{definition}{Definition}
\newtheorem{proposition}{Proposition}

\newtheorem{assumption}{Assumption}

\newtheorem{lemma}{Lemma}

\usepackage{lineno}
\newcommand*\patchAmsMathEnvironmentForLineno[1]{
  \expandafter\let\csname old#1\expandafter\endcsname\csname #1\endcsname
  \expandafter\let\csname oldend#1\expandafter\endcsname\csname end#1\endcsname
  \renewenvironment{#1}
     {\linenomath\csname old#1\endcsname}
     {\csname oldend#1\endcsname\endlinenomath}}
\newcommand*\patchBothAmsMathEnvironmentsForLineno[1]{
  \patchAmsMathEnvironmentForLineno{#1}
  \patchAmsMathEnvironmentForLineno{#1*}}
\AtBeginDocument{
\patchBothAmsMathEnvironmentsForLineno{equation}
\patchBothAmsMathEnvironmentsForLineno{align}
\patchBothAmsMathEnvironmentsForLineno{flalign}
\patchBothAmsMathEnvironmentsForLineno{alignat}
\patchBothAmsMathEnvironmentsForLineno{gather}
\patchBothAmsMathEnvironmentsForLineno{multline}
}
\modulolinenumbers[5]
\title{Efficient estimation of weighted treatment effects under two-phase sampling}

\author[1]{Kazuharu Harada\thanks{Corresponding author: Kazuharu Harada.
E-mail: \href{mailto:haradak@tokyo-med.ac.jp}{haradak@tokyo-med.ac.jp}.}}
\author[1]{Masataka Taguri}

\affil[1]{Department of Health Data Science, Tokyo Medical University,
6-1-1 Shinjuku, Shinjuku-ku, Tokyo 160-8402, Japan}
\date{}

\begin{document}
\maketitle
\begin{abstract}
Two-phase sampling offers a practical way to collect costly confounders only in a subsample while retaining inexpensive information for a larger cohort. In observational causal studies, however, phase-2 selection can distort estimation of population causal effects if the sampling mechanism is ignored, and available phase-1 information may also be exploited to improve efficiency. Yet efficiency theory for causal estimands under such designs remains limited, particularly beyond the average treatment effect.

In this paper, we derive the semiparametric efficiency bound for a class of propensity-score-weighted average treatment effects, which includes the average treatment effect, effects among treated and untreated populations, and the overlap effect, under two-phase sampling. In addition to straightforward weighting estimators based on the known sampling probabilities, we propose an enriched doubly robust estimator that attains the efficiency bound when all nuisance functions are consistently estimated. In particular, under outcome-dependent sampling, substantial efficiency gains can arise in some settings by appropriately incorporating phase-1 information. We further conduct extensive simulation studies, varying the choice of phase-1 variables and sampling schemes, to characterize when and to what extent leveraging phase-1 information leads to efficiency gains.
\end{abstract}



\section{Introduction}

Many statistical methods rely on an ignorable sampling mechanism that links the observed sample to a target population. When selection depends on variables related to the analysis and is ignored, standard estimators need not identify their population targets even if the within-sample outcome or treatment model is otherwise appropriate.
In practice, simple random sampling is often infeasible because of technological, logistical, or financial constraints. Researchers therefore
use designs with inclusion probabilities that vary across subpopulations and must account for those probabilities in estimation and inference.

Survey sampling is a well-established field of statistics that systematically investigates how sampling design affects estimation and inference. Foundational works \cite{Cochran1977-bs, Sarndal2003-ol, Lohr2021-dl} have extensively discussed the influence of design features, such as stratification, clustering, and unequal-probability selection, on the estimation of population quantities, and have developed principles for constructing unbiased or efficient estimators under these complex sampling schemes.

Among these designs, two-phase sampling has played a particularly important role in both survey and biomedical research \citep{Neyman1938-dj, Hansen1946-dr, Prentice1986-sc}. In a typical two-phase design, inexpensive or routinely available variables are collected at phase 1, possibly for the entire database, and a subsample is then selected at phase 2 to obtain additional variables that are costly or time-consuming to measure. This flexibility allows researchers to oversample rare or underrepresented strata based on phase-1 variables, thereby improving efficiency and feasibility compared with uniform random sampling.

Such designs are especially valuable in observational causal studies, where adjustment for confounding is essential but measuring all relevant covariates for the entire population is often infeasible. A well-known example is the European Prospective Investigation into Cancer and Nutrition (EPIC) cohort \cite{Riboli2002-ep}, which enrolled over 500{,}000 participants and collected baseline questionnaires and biospecimens at phase 1, while costly assays of circulating hormones, metabolomics panels, or genotypes are performed at phase 2 only on cases and a randomly selected subcohort. Within this cohort, causal analyses have addressed the effects of dietary, hormonal, or genetic exposures on chronic disease risk \cite{Forouhi2014-fa}.

Consequently, increasing attention has been paid to the interplay between sampling designs and causal inference in recent years \cite{Nattino2024-zf, Yukang2025-da}. These studies reveal that the choice of sampling scheme can substantially affect the efficiency and potential bias of causal effect estimators.

This paper focuses on a practically important class of
propensity-score-weighted average treatment effects (WATEs), including the
average treatment effect (ATE), average treatment effect on the
treated/untreated (ATT/ATU), and average treatment effect for the overlap
  population (ATO). We derive the semiparametric efficiency bound for this
class and study the consistency and efficiency of several estimators under
explicit sampling-design and nuisance-model conditions.

The present study is closely related to several key works.
First, Wang et al. \cite{Wang2009-kn} derived the semiparametric efficiency bound for the ATE under two-phase sampling that allows outcome variables to be included in phase 1, assuming a parametric model for the propensity score. They proposed a locally efficient estimator that attains this bound asymptotically, as well as a simplified sub-efficient estimator for practical implementation.
    Zhang et al. \cite{Zhang2019-dx} discussed efficient ATE estimation when sampling depends solely on treatment assignment.
    Song and Yu \cite{Song2022-hc} considered stratified two-phase sampling based on the binary treatment and finite and discrete covariates, establishing efficiency bounds and efficient estimators for both ATE and ATT.

Our contributions can be summarized as follows:{
\begin{enumerate}
    \item We derive the efficient influence function and semiparametric
    efficiency bound for the WATE class under known conditionally independent
    Poisson phase-2 sampling, without imposing parametric restrictions on
    the full-data nuisance functions.
    \item We propose two estimators, an inverse probability of sampling weighting (IPSW) estimator and an enriched doubly robust (EDR) estimator, characterize conditions for their consistency, show that the latter attains the efficiency bound when all nuisance functions are consistently estimated, and obtain a closed-form decomposition of the enrichment gain into residual-variation and effect-heterogeneity components.
    \item We provide a double-machine-learning formulation with cross-fitting
    that preserves $\sqrt{n}$-asymptotic normality under data-adaptive nuisance
    estimation satisfying the stated positivity, moment, and rate conditions, and attains the efficiency bound.
    \item Combining the theoretical decomposition with simulations, we examine which phase-1 variables yield larger efficiency gains, notably the outcome under outcome-dependent sampling and effect-modifying covariates, providing practical design guidance for two-phase observational causal studies.
\end{enumerate}
}

\section{Preliminaries}
\subsection{Causal Inference under Two-Phase Sampling} \label{sec:prelim1}
  We introduce the identification and basic estimators of causal effects under
two-phase sampling within the potential outcomes framework
\citep{Rubin1974-zp, Holland1986-ut}.

Throughout this paper, we assume that low-cost covariates $V$ and the binary treatment variable $A\in\{0,1\}$ are always observed in phase 1.
The outcome variable $Y$, which may be binary or continuous, may or may not be collected in phase 1; we consider both scenarios, theoretically in a unified manner and separately in the simulation study.
Let $S$ denote the collection of all variables observed in phase 1.
Additional high-cost covariates collected in phase 2 are denoted by $W$, and the full set of covariates is expressed as $X=(V,W)$.
Let $\delta\in\{0,1\}$ be the indicator denoting whether a unit is selected into phase 2, and let the known (designed) sampling probability be $q(S=s) = P(\delta=1\mid S=s)$. We assume that $q(S)\ge q_{\min}>0$ almost surely for some constant $q_{\min}$.

When $S$ is a discrete variable with finite support, a common sampling design is simple random sampling without replacement (SRSWOR) within strata defined by $S$.
In this case, the phase-1 sample is partitioned into strata, and a predetermined number of individuals are chosen for each stratum.
The sampling probability is then given by $q(S=k) = m_{k}/n_{k}$,
where $n_{k}$ and $m_{k}$ denote the phase-1 and phase-2 sample sizes within stratum $S=k$, respectively.
Alternatively, one may predefine $q(S=k)$ and perform Bernoulli sampling independently for each $i$ with probability $q(S_i)$; this design is referred to as Poisson sampling.

Unless stated otherwise, the efficiency calculations, asymptotic results, variance formulas, and Wald inference are established under Poisson sampling. Stratified SRSWOR is used only for the finite-sample identity in Proposition~\ref{prop:srswor} and for numerical illustration.

Let $Y^1$ and $Y^0$ denote the potential outcomes that would be observed under treatment and control.

We assume consistency, $Y=AY^1+(1-A)Y^0$, conditional unconfoundedness,
\begin{gather*}
Y^a\ci A\mid X,\qquad a=0,1,
\end{gather*}
and positivity: there exists a constant $c\in(0,1/2)$ such that
$c<\pi(x)<1-c$ for all $x$ in the support of $X$, where
$\pi(x)=P(A=1\mid X=x)$ denotes the propensity score
\cite{Li2018-cq, Matsouaka2024-ht, Matsouaka2024-vl}.
When useful, we write $\pi_a(x)=P(A=a\mid X=x)$ for $a\in{0,1}$.

The target parameter of this paper is WATE, defined as the weighted average of the difference between the potential outcomes under treatment and control in the population. Specifically, we consider the case in which the weight is a nonnegative differentiable function $h:[0,1]\to\mathbb{R}_+$ of the propensity score. Then, WATE is given by
\begin{gather*}
\tau_{h} = C_h^{-1}E[h\{\pi_1(X)\}(Y^1 - Y^0)],
\end{gather*}
where $C_h = E[h\{\pi_1(X)\}]$.
We assume $0<C_h<\infty$ and
$E[h\{\pi_1(X)\}|Y^1-Y^0|]<\infty$.

WATE is a general class of treatment effects that includes many practically useful parameters, such as
{
\begin{itemize}
    \item ATE with $h\{\pi_1(x)\}=1$,
    \item ATT with $h\{\pi_1(x)\}=\pi_1(x)$,
    \item ATU with $h\{\pi_1(x)\}=1-\pi_1(x)$, and
    \item ATO with $h\{\pi_1(x)\}=\pi_1(x)(1-\pi_1(x))$.
\end{itemize}
}
For further examples belonging to this class, see
    Refs.~\cite{Matsouaka2024-ht, Matsouaka2024-vl, Yiming2025-do}.

Let $\mu_a(x) = E[Y^a\mid X=x]$ $(a=0,1)$ denote the conditional expectation of the potential outcome given covariates $X=x$, and define the conditional average treatment effect (CATE) as
\begin{gather*}
\tau(x) = \mu_1(x) - \mu_0(x) = E[Y^1 - Y^0\mid X=x].
\end{gather*}
Then, WATE can be equivalently expressed as
\begin{gather*}
\tau_{h}
= C_h^{-1}E[h\{\pi_1(X)\}\{\mu_1(X) - \mu_0(X)\}]
= C_h^{-1}E[h\{\pi_1(X)\}\tau(X)].
\end{gather*}
For notational simplicity, we let $q_i:= q(S_i)$, $\pi_{ai}:= \pi_a(X_i)$ for $a=0,1$, $\mathbf{1}^A_{ai}:= \mathbf{1}\{A_i=a\}$, $\mu_{ai}:= \mu_a(X_i)$, $\tau_i:= \tau(X_i)$, and $h_i:= h(\pi_{1i})$.
A hat ($\hat{\ }$) on these symbols denotes the corresponding estimator.

Under two-phase sampling, we consider a basic inverse probability of treatment weighting (IPTW) estimator of H\'{a}jek type \citep{Hajek1964-km} for ATE $\tau_{h=1}$, for example:
\begin{gather*}
    \hat\tau^{\text{iptw}}_{h=1} = \left[\sum_{i=1}^n\frac{\delta_i\mathbf{1}^A_{1i}}{q_i\hat \pi_{1i}}\right]^{-1}
    \left[\sum_{i=1}^n\frac{\delta_i\mathbf{1}^A_{1i}}{q_i\hat \pi_{1i}}Y_i\right] -
    \left[\sum_{i=1}^n\frac{\delta_i\mathbf{1}^A_{0i}}{q_i\hat\pi_{0i}}\right]^{-1}
    \left[\sum_{i=1}^n\frac{\delta_i\mathbf{1}^A_{0i}}{q_i\hat\pi_{0i}}Y_i\right].
\end{gather*}

Since the propensity score is a function of $X$, which includes the phase-2 variables $W$, it is necessary to account for the two-phase sampling design when estimating the propensity score.
For example, suppose that a parametric logistic regression model $\pi(x;\alpha)$ is specified for the propensity score, with parameter $\alpha$, and let $u_\pi(a,x;\alpha)$ denote the corresponding score function.
Then, under correct model specification and known sampling probabilities, the solution $\hat\alpha$ to the following weighted estimating equation yields a consistent estimator of the true value of $\alpha$:
\begin{gather*}
    \sum_{i=1}^n \frac{\delta_i}{q_i} u_\pi(A_i,X_i;\hat\alpha) = 0.
\end{gather*}
If the propensity score is consistently estimated, then $\hat\tau^{\mathrm{iptw}}_{h=1}$ is a consistent estimator of the ATE.

\subsection{Semiparametric Theory} \label{sec:prelim2}
  We work in a semiparametric model that leaves functions such as the propensity
score $\pi(x)$ and CATE $\tau(x)$ unrestricted while imposing the causal
identification assumptions. This subsection reviews the notation needed below.

Let $\mathcal{M}$ denote a semiparametric model, and suppose that we observe independent samples $Z_1,\ldots,Z_n$ drawn from a distribution $P\in\mathcal{M}$.
Let $p$ denote a density of $P$ with respect to a dominating measure.
Consider a regular one-dimensional parametric submodel $P_\varepsilon\in\mathcal{M}$ satisfying $P_\varepsilon = P$ when $\varepsilon = 0$, and let $u = \partial_\varepsilon \log{p_\varepsilon}|_{\varepsilon=0}$ denote the score function with respect to $\varepsilon$.
The tangent space $\mathcal{T}$ of $\mathcal{M}$ is then defined as the closure of the linear span of the score functions of all regular parametric submodels.
Under the regularity of the model, any element of $\mathcal{T}$ has mean zero and is square-integrable with respect to $P$.

For a pathwise differentiable target functional $\theta(P)$, there exists a
unique $\phi\in\mathcal{T}$ satisfying
\begin{gather*}
    \dot\theta(P_\varepsilon) = \left.\frac{d}{d\varepsilon}\theta(P_\varepsilon)\right|_{\varepsilon=0}
        = E[\phi(Z)u(Z)]
\end{gather*}
for any regular submodel.
This $\phi$ is called the efficient influence function (EIF), and the semiparametric efficiency bound for estimating the parameter $\theta(P)$ is given by $\text{Var}(\phi(Z))$.

Let $\hat\theta_n$ be an estimator of $\theta(P)$.
If it admits the representation
\begin{gather*}
    \sqrt{n}(\hat\theta_n - \theta(P)) = n^{-1/2}\sum_{i=1}^{n} \tilde\phi(Z_i) + o_p(1),
\end{gather*}
then $\hat\theta_n$ is said to be regular asymptotically linear (RAL). Here, $\tilde\phi(Z_i)$ is a mean-zero, square-integrable function with respect to $P$, called the influence function (IF).
By the central limit theorem, a RAL estimator satisfies
\begin{gather*}
    \sqrt{n}(\hat\theta_n - \theta(P)) \overset{d}{\rightarrow} N(0,\text{Var}(\tilde\phi(Z))),
\end{gather*}
and thus its efficiency can be evaluated through the variance of its IF.
  The EIF is the projection of any IF onto the tangent space and has the smallest
variance; an RAL estimator with this IF is semiparametrically efficient.

Under fully observed data, WATE estimators can be constructed from the EIF
 \cite{Li2018-cq,Tao2019-db,Yiming2025-do}, leaving
the nuisance functions $\pi$, $\mu_1$, and $\mu_0$ unrestricted.

Let all nuisance parameters be collectively denoted by $\eta$.
The EIF for $\tau_h$ is given by
\begin{align*}
    \phi(O_i;\tau_h,\eta)
    =&~ C_h^{-1}h_i\left\{\frac{\mathbf{1}^A_{1i}}{\pi_{1i}}(Y_i - \mu_{1i}) - \frac{\mathbf{1}^A_{0i}}{\pi_{0i}}(Y_i - \mu_{0i})\right\} \\
    &~ + C_h^{-1}\{h_i + h'(\pi_{1i})(A_i-\pi_{1i})\}(\tau_i - \tau_h),
\end{align*}
where $O_i = (A_i, X_i, Y_i)$ for $i\in\{1,\ldots,n\}$, and $h'$ denotes the derivative of $h$ with respect to its scalar argument.

The RAL estimator $\hat\tau_h$ with this EIF can be constructed as the solution to the following estimating equation, in which the nuisance components $\eta$ are replaced by their estimators:
\begin{align*}
    \sum_{i=1}^n \left[
        \hat h_i\left\{\frac{\mathbf{1}^A_{1i}}{\hat\pi_{1i}}(Y_i - \hat\mu_{1i}) - \frac{\mathbf{1}^A_{0i}}{\hat\pi_{0i}}(Y_i - \hat\mu_{0i})\right\}
        + \{\hat h_i + h'(\hat\pi_{1i})(A_i-\hat\pi_{1i})\}(\hat\tau_i - \hat\tau_h)
    \right] = 0.
\end{align*}

Typically, the nuisance estimators are obtained using parametric working models.
If $\pi$ is consistently estimated, this estimator is consistent for $\tau_h$.
Moreover, when $h$ is linear in the propensity score, $h(t)=a+bt$ (as in
the ATE, ATT, or ATU), the estimator enjoys double robustness: it is consistent if either
the propensity score model $\pi$ or the outcome regression models
$(\mu_1,\mu_0)$ are correctly specified \citep{Bang2005-sc,Tao2019-db}.
When all working models are correctly specified, this estimator is semiparametrically efficient.

  For flexible nuisance estimation, cross-fitted double machine learning
   (DML) \cite{Chernozhukov2017-yn, Chernozhukov2018-cc,
Yiming2025-do} separates nuisance fitting from evaluation of the
target equation; Section~\ref{sec:DML} adapts it to two-phase sampling.

For further details on semiparametric theory, see
    Refs.~\cite{Bickel1998-el,Tsiatis2006-sz}; for semiparametric theory in
  statistical causal inference, see
    Refs.~\cite{Kennedy2017-pc,Fisher2021-nf,Hines2022-nh}.

\section{Semiparametric Efficiency Bound}

We denote by $Z^\star=(A,V,W,Y^1,Y^0)$ the complete data.
Under consistency, the corresponding fully observed data are $Z=(A,V,W,Y)$.
Under the actual two-phase sampling design, the observed data take one of the following forms:
{
\begin{gather*}
    O_{\mathrm{ODS}}=(A,V,Y,\delta,\delta W), \qquad
    O_{\mathrm{IS}}=(A,V,\delta,\delta W,\delta Y).
\end{gather*}
}
Here, $O_{\mathrm{ODS}}$ corresponds to outcome-dependent sampling (ODS), in which the outcome $Y$ is available for all individuals at phase 1 and the phase-2 subsampling determines whether $W$ is additionally observed. In contrast, $O_{\mathrm{IS}}$ corresponds to outcome-independent sampling, in which neither $W$ nor $Y$ is fully observed at phase 1 and both are observed only for sampled individuals.
When no confusion arises, we write $O$ to denote either $O_{\mathrm{ODS}}$ or $O_{\mathrm{IS}}$.
The phase-1 variable set $S\subset O$ is given by $S=(A,V,Y)$ in the former case and $S=(A,V)$ in the latter.

Under the assumptions required for identifying the treatment effect introduced in Section~\ref{sec:prelim1}, and when the phase-2 sampling probability depends only on $S$, the tangent space of the observed-data model can be characterized as follows.

\begin{proposition}\label{prop:tangent}
    The observed-data scores spanning the tangent space are of the form
     \begin{align*}
        u(o) =&~ \delta \left\{u_{F}(z) - E[u_{F}(z)\mid S=s]\right\} + u_{S}(s), \\
        u_{F}(z) =&~ a u_1(y\mid x) + (1-a)u_0(y\mid x) + b(x)(a-\pi(x)) + u_X(x),
    \end{align*}
    where $\int u_a(y\mid x)p_a(y\mid x)dy = 0$ for any $x$ with the conditional density of $Y^a$ given $X$ ($a\in\{0,1\}$), $u_X, u_S$ are arbitrary mean-zero, square-integrable functions, and $b$ is an arbitrary square-integrable function.
\end{proposition}

Furthermore, by considering an appropriate one-dimensional parametric submodel and differentiating the WATE with respect to its parameter, we obtain the pathwise derivative. The EIF is then defined as the element of the tangent space whose inner product with any score function equals this pathwise derivative.

\begin{proposition}\label{prop:EIF}
   In an observational study, the complete data distribution satisfies the assumptions of conditional unconfoundedness, positivity, and consistency.
   When two-phase sampling is conducted based on the phase-1 variable $S$, the EIF of the WATE is given by
      \begin{align*}
        \phi^{\mathrm{eff}}(o; \tau_h, \eta)
            =&~ E[\phi_{F}(Z; \tau_h, \eta)\mid S=s]
            + \delta\left\{
                \frac{\phi_{F}(z; \tau_h, \eta)
                -E[\phi_{F}(Z; \tau_h, \eta)\mid S=s]}{q(s)}
            \right\},\\
        \phi_{F}(z; \tau_h, \eta)
            =&~ C_h^{-1}h\{\pi_1(x)\}\left\{
               \frac{\mathbf{1}\{a=1\}}{\pi_1(x)}(y - \mu_{1}(x))
               - \frac{\mathbf{1}\{a=0\}}{\pi_0(x)}(y - \mu_{0}(x))\right.\\
            &\hspace{43mm}\left.{}+ \tau(x) - \tau_h\right\} \\
        &~~~ + C_h^{-1}h'(\pi_1(x))(\tau(x) - \tau_h)(a-\pi_1(x)).
    \end{align*}
    Assume that
    $E[\phi_F(Z;\tau_h,\eta)^2/q(S)]<\infty$.
     The semiparametric efficiency bound $\mathcal{V}_h$ is equal to the variance of $\phi^{\mathrm{eff}}(o; \tau_h, \eta)$, expressed in two ways as
    \begin{align*}
        \mathcal{V}_h
            =&~ \text{Var}(\phi^{\mathrm{eff}}(O; \tau_h, \eta)) \\
            =&~ \text{Var}(\phi_F(Z;\tau_h,\eta)) + E\left[\left(\frac{1}{q(S)} - 1\right)\text{Var}(\phi_F(Z;\tau_h,\eta)\mid S)\right]\\
            =&~ \text{Var}(E[\phi_F(Z;\tau_h,\eta)\mid S]) + E\left[
                \frac{\text{Var}(\phi_F(Z;\tau_h,\eta)\mid S)}{q(S)}
            \right].
    \end{align*}
    The variance of the full-data EIF is given by
    \begin{align*}
        \text{Var}(\phi_F(Z;\tau_h,\eta))
            =&~ C_h^{-2}E\left[
                h^2\{\pi_1(X)\}\left\{\frac{\sigma_1^2(X)}{\pi_1(X)} + \frac{\sigma_0^2(X)}{\pi_0(X)}\right\}\right] \\
                & + C_h^{-2}E[\{h^2\{\pi_1(X)\} + h'(\pi_1(X))^2\pi_1(X)\pi_0(X)\}\{\tau(X)-\tau_h\}^2],
    \end{align*}
    where $\sigma_a^2(X)=\operatorname{Var}(Y^a\mid X)
    =\operatorname{Var}(Y\mid A=a,X)$ for $a=0,1$.
\end{proposition}
The function $\phi_F(Z;\tau_h,\eta)$ is the EIF for the WATE under the fully observed-data model \citep{Yiming2025-do}.
The observed-data augmentation in Proposition~\ref{prop:EIF} is the
standard missing-at-random transformation of a full-data influence
function \citep{Robins1994-rrz,Tsiatis2006-sz}.
The second expression shows that $\mathcal{V}_h$ equals the fully-observed-data
efficiency bound $\mathrm{Var}(\phi_F)$ plus an additional term from phase-2 subsampling.
Proofs of Propositions \ref{prop:tangent} and \ref{prop:EIF} are provided in Appendix \ref{app:proofs}.

\section{Estimation}\label{sec:estim}
We propose several estimators under two-phase sampling and study their target
consistency, efficiency, and behavior under working-model misspecification.
We use parametric working models to estimate nuisance parameters.
Let $\pi(x;\alpha)$ and $\mu_a(x;\beta^a)$, with $\alpha,\beta^a\in\doubleR^p$ and $a\in\{0,1\}$, denote parametric working models for the nuisance functions $\pi(x)$ and $\mu_a(x)$. The estimators of $\alpha$ and $\beta^a$ are obtained via weighted maximum likelihood using the inverse sampling probability.
Denote the population roots of their weighted score equations by
$\alpha^*$ and $\beta^{a*}$. When the corresponding working models are
correctly specified,
\begin{gather*}
    \pi(x;\alpha^*) = \pi(x),\qquad
    \mu_a(x;\beta^{a*}) = \mu_a(x).
\end{gather*}
For simplicity, the dimension of each nuisance parameter is set to be
$p$. We use $\theta_h\in\mathbb{R}^r$ as a unified notation for the
target block, as detailed below.

Let $\eta=(\alpha^\top,\beta^{1\top},\beta^{0\top})^\top$ in the same order as the nuisance estimating equations below. For the parametric estimators, write
\[
    \vartheta_{\mathrm{ipsw}}
        =(\theta_h^\top,\eta^\top)^\top
        \in\mathbb{R}^{r+3p},
    \qquad
    \vartheta_{\mathrm{enr}}
        =(\theta_h^\top,\eta^\top,\gamma^\top)^\top
        \in\mathbb{R}^{r+3p+d_\gamma},
\]
where $\gamma\in\mathbb{R}^{d_\gamma}$ indexes the enrichment model
introduced below. For $e\in\{\mathrm{ipsw},\mathrm{enr}\}$, let
$\Psi_e(O;\vartheta_e)$ denote the corresponding full, stacked
estimating function and let $\vartheta_e^*$ denote its population root.
The first $r$ components of $\vartheta_e^*$ are denoted by
$\theta_{h,e}^*$.

The regularity conditions for the parametric estimators and the asymptotic
results below are stated in Appendix~\ref{app:parametric-regularity}.

\subsection{Inverse Probability of Sampling Weighting}\label{sec:ipsw}
A simple approach to eliminating sampling bias is the inverse probability of sampling weighting (IPSW) estimator, which weights the estimating equation for the fully observed data by $\delta / q(S)$ \citep{Horvitz1952-ht}.

The DR-type estimator described later requires only a single main estimating equation for $\tau_h$. In contrast, for the IPTW-type estimator, it is more convenient to introduce separate estimating equations for the components $\mu^a_{h}:= C_h^{-1}E[h\{\pi(X)\}Y^a]$ for $a \in \{0, 1\}$, rather than working with $\tau_h$ alone; accordingly, we consider a system of three stacked estimating equations.
Thus, $\theta_h=(\mu_h^1,\mu_h^0,\tau_h)^\top$ and $r=3$ for the
IPTW-type equation, whereas $\theta_h=\tau_h$ and $r=1$ for the DR-type
equation.

\begin{definition}[IPSW estimator]\label{def:ipsw}
    Let $m_F(z;\theta_h,\eta)$ denote the main estimating function for $\theta_h$ under the fully observed setting.
    The IPSW estimator $\hat\theta^{\text{ipsw}}_h$ solves the following estimating equations:
    \begin{gather*}
        E_n\left(\begin{array}{c}
            \frac{\delta}{q(S)}m_F(Z;\hat\theta^{\text{ipsw}}_h,\hat\eta) \\
            \frac{\delta}{q(S)}u_\pi(Z;\hat\alpha) \\
            \frac{\delta}{q(S)}u_1(Z;\hat\beta^1) \\
            \frac{\delta}{q(S)}u_0(Z;\hat\beta^0)
        \end{array}\right) = 0,
    \end{gather*}
    where $E_n$ denotes the averaging operator over the whole sample $i\in\{1,\ldots,n\}$.
\end{definition}

As the estimating function for the fully observed data $m_F$, we consider in this section the following IPTW-type and DR-type estimating functions.
\begin{align*}
    m_F^{\text{iptw}}(Z;\theta_h,\eta)
        =&~ \left(\begin{array}{c}
            \mathbf{1}^A_1h\{\pi_1(X;\alpha)\}(Y-\mu_h^1)/\pi_1(X;\alpha) \\
            \mathbf{1}^A_0h\{\pi_1(X;\alpha)\}(Y-\mu_h^0)/\pi_0(X;\alpha) \\
            \mu_h^1 - \mu_h^0 - \tau_h
        \end{array}\right), \\
    m_F^{\text{dr}}(Z;\theta_h,\eta)
        =&~ h\{\pi_1(X;\alpha)\}
            \frac{\mathbf{1}^A_1}{\pi_1(X;\alpha)}
            \{Y - \mu_1(X;\beta^1)\}\\
            &~ -h\{\pi_1(X;\alpha)\}
            \frac{\mathbf{1}^A_0}{\pi_0(X;\alpha)}
            \{Y - \mu_0(X;\beta^0)\}\\
            &~ + \left[ h\{\pi_1(X;\alpha)\}
            + h'\{\pi_1(X;\alpha)\}\{A-\pi_1(X;\alpha)\}\right]\\
            &\hspace{25mm}\times\{\tau(X;\beta^1,\beta^0) - \tau_h\}
\end{align*}
We refer to the IPSW-IPTW estimator as the SIW estimator,
and the IPSW-DR estimator as the SDR estimator.
When the target parameter is the ATE,
these estimators coincide with the ``Simple Inverse Weighting (SIW)''
and ``Simple Doubly Robust (SDR)'' estimators
proposed by Wang et al. \cite{Wang2009-kn}.

The asymptotic distribution of the IPSW estimator follows from standard
 M-estimation theory \cite{Van_der_Vaart2000-bl,
Stefanski2002-wx, Tsiatis2006-sz}.

By the phase-2 sampling design, $E\{\delta/q(S)\mid Z\}=1$, and hence
\begin{gather*}
    E\left[\frac{\delta}{q(S)}m_F(Z;\tilde\theta_h,\tilde\eta)\right] = E\left[m_F(Z;\tilde\theta_h,\tilde\eta)\right]
\end{gather*}
for any values of $(\tilde\theta_h,\tilde\eta)$,
which implies that assessing the unbiasedness of the IPSW estimating equation reduces to verifying the unbiasedness of the fully observed-data estimating function.

For the SIW estimator, the nuisance model involves only the propensity score, which must be correctly specified for target consistency.
On the other hand, for the SDR estimator, if the weight function $h$ is linear in the propensity score, the estimating equation remains unbiased as long as either the propensity score model or both outcome regression models are correctly specified, which is the property of double robustness \citep{Bang2005-sc, Tao2019-db}.
Important cases with linear weights include the ATE, ATT, and ATU, for which double robustness holds.
When the weight function $h$ is nonlinear in its argument, as in the case of the ATO, the propensity score model must be correctly specified.
Under propensity-score misspecification, the SDR estimator instead converges to the pseudo-true root of its population estimating equation, which generally is not the WATE obtained by simply replacing $\pi_1$ with the probability limit of the misspecified propensity-score estimator.

Let $J_{\mathrm{ipsw}}$ be the Jacobian of the stacked IPSW estimating
function, evaluated at
$\vartheta_{\mathrm{ipsw}}^*
=(\theta_{h,\mathrm{ipsw}}^{*\top},\eta^{*\top})^\top$.
Its block form is
\begin{equation*}
    J_{\mathrm{ipsw}}=
    \left(\begin{array}{cccc}
        J_{11} & J_{12} & J_{13} & J_{14} \\
        {\bs 0}_{p\times r} & J_{22} & {\bs 0}_{p\times p} & {\bs 0}_{p\times p} \\
        {\bs 0}_{p\times r} & {\bs 0}_{p\times p} & J_{33} & {\bs 0}_{p\times p}\\
        {\bs 0}_{p\times r} & {\bs 0}_{p\times p} & {\bs 0}_{p\times p} & J_{44}
    \end{array}\right).
\end{equation*}
The nonzero blocks are
\begin{align*}
    J_{11}&=E[\partial_{\theta_h^\top}m_F],
    &J_{12}&=E[\partial_{\alpha^\top}m_F],\\
    J_{13}&=E[\partial_{\beta^{1\top}}m_F],
    &J_{14}&=E[\partial_{\beta^{0\top}}m_F],\\
    J_{22}&=E[\partial_{\alpha^\top}u_\pi],
    &J_{33}&=E[\partial_{\beta^{1\top}}u_1],\\
    J_{44}&=E[\partial_{\beta^{0\top}}u_0],
\end{align*}
where every derivative is evaluated at the corresponding population
root. We used $E\{\delta/q(S)\mid Z\}=1$ to express these blocks.

Block inversion gives the following target-block influence function:
\begin{align*}
    \phi^{\text{ipsw}}_h
    (O;\theta_{h,\mathrm{ipsw}}^*,\eta^*)
        =&~ -\frac{\delta}{q(S)}J_{11}^{-1}\left\{
            m_F(Z;\theta_{h,\mathrm{ipsw}}^*,\eta^*)
             - J_{12}J^{-1}_{22}u_\pi(Z;\alpha^*)\right.\\
             &\qquad\qquad\qquad \left.- J_{13}J^{-1}_{33}u_1(Z;\beta^{1*})
             - J_{14}J^{-1}_{44}u_0(Z;\beta^{0*})
        \right\}.
\end{align*}
  \begin{proposition}\label{prop:asympIPSW}
    Under the regularity conditions in Appendix~\ref{app:parametric-regularity},
\[
    \sqrt{n}\{\hat\theta^{\text{ipsw}}_h
                  -\theta_{h,\mathrm{ipsw}}^*\}
    =\frac{1}{\sqrt n}\sum_{i=1}^n
        \phi^{\text{ipsw}}_h
        (O_i;\theta_{h,\mathrm{ipsw}}^*,\eta^*)+o_p(1).
\]
Consequently,
\[
    \sqrt{n}\{\hat\theta^{\text{ipsw}}_h
                  -\theta_{h,\mathrm{ipsw}}^*\}
    \overset{d}{\longrightarrow}
    N_r\!\left(
        0,\,
        \operatorname{Var}\{
        \phi^{\text{ipsw}}_h
        (O;\theta_{h,\mathrm{ipsw}}^*,\eta^*)\}
    \right).
\]
  Under the model-specification conditions stated above,
$\theta_{h,\mathrm{ipsw}}^*=\theta_h$ and
$\hat\theta_h^{\mathrm{ipsw}}\overset{p}{\longrightarrow}\theta_h$.
  \end{proposition}
Proofs for this proposition and Proposition \ref{prop:asympENR} in the next subsection are provided in Appendix \ref{app:proofs}.

When the DR-type estimating function is used for $m_F$ and all working models are correctly specified, we have $J_{12} = J_{13} = J_{14} = 0$, and $-J_{11}^{-1}m_F^{\text{dr}} = \phi_F$. Hence, the asymptotic variance is given by
\begin{gather*}
    \text{Var}(\phi^{\text{sdr}}_h(O;\theta_h,\eta))
        = \text{Var}\left(
            \frac{\delta}{q(S)}\phi_F(Z;\tau_h,\eta)
        \right)
        = E\left[
            \frac{1}{q(S)}\phi_F(Z;\tau_h,\eta)^2
        \right].
    \end{gather*}
Comparison with Proposition~\ref{prop:EIF} shows that this variance
generally exceeds $\mathcal V_h$; Section~\ref{sec:benefit} gives the exact
difference.

\subsection{Enriched Estimators}\label{sec:enrich}
Following Wang et al. \cite{Wang2009-kn}, we call estimators that
augment the IPSW equation with phase-1 information in $S$ ``enriched.''

Let $g(s;\theta_h,\gamma)\in\mathbb{R}^r$, where
$\gamma\in\mathbb{R}^{d_\gamma}$ collects the parameters of separate working
models for the conditional means of the components used to construct $m_F$. Let
$U_g(O;\eta,\gamma)\in\mathbb{R}^{d_\gamma}$ be the corresponding stacked
estimating function. For example, write
$m_F^{\mathrm{dr}}=N_F-\tau_hD_F$, model $E(N_F\mid S)$ and $E(D_F\mid S)$
by $g_N(S;\gamma_N)$ and $g_D(S;\gamma_D)$, respectively, define
$\gamma=(\gamma_N^\top,\gamma_D^\top)^\top$, and set

\begin{align*}
    g(S;\tau_h,\gamma)
       &=g_N(S;\gamma_N)-\tau_hg_D(S;\gamma_D),\\
    U_g(O;\eta,\gamma)
       &=\frac{\delta}{q(S)}
       \begin{pmatrix}
       \partial_{\gamma_N}g_N(S;\gamma_N)
          \{N_F(Z;\eta)-g_N(S;\gamma_N)\}\\
       \partial_{\gamma_D}g_D(S;\gamma_D)
          \{D_F(Z;\eta)-g_D(S;\gamma_D)\}
       \end{pmatrix}.
\end{align*}
For EIW, $U_g$ analogously stacks separate equations for the conditional means defined in Appendix~\ref{sec:closed}.
Saturated models may be used when $S$ has finite support.

 \begin{definition}[Enriched estimator]\label{def:enr}
    The enriched estimator $\hat\theta^{\text{enr}}_h$ solves the following estimating equations:
        \begin{gather*}
        E_n\left(\begin{array}{c}
            \frac{\delta}{q(S)}m_F(Z;\hat\theta^{\text{enr}}_h,\hat\eta)
                + \left(1 - \frac{\delta}{q(S)}\right)g(S;\hat\theta^{\text{enr}}_h, \hat\gamma) \\
            \frac{\delta}{q(S)}u_\pi(Z;\hat\alpha) \\
            \frac{\delta}{q(S)}u_1(Z;\hat\beta^1) \\
            \frac{\delta}{q(S)}u_0(Z;\hat\beta^0) \\
            U_g(O;\hat\eta,\hat\gamma)
        \end{array}\right) = 0
    \end{gather*}
\end{definition}

We use the same IPTW- and DR-type estimating functions, referring to
their enriched versions as EIW and EDR.
By the phase-2 sampling design, for any values of
$(\tilde\theta_h,\tilde\eta,\tilde\gamma)$,
\begin{align*}
    &E\left[
        \frac{\delta}{q(S)}m_F(Z;\tilde\theta_h,\tilde\eta)
        + \left\{1 - \frac{\delta}{q(S)}\right\}
          g(S;\tilde\theta_h,\tilde\gamma)
    \right]\\
    &\hspace{35mm}
    = E\left[m_F(Z;\tilde\theta_h,\tilde\eta)\right].
\end{align*}
Thus, EIW and EDR inherit the model-specification conditions stated above
for SIW and SDR regardless of whether the enrichment working models are
correctly specified. The conditional-mean equality stated below is imposed
for the efficiency result, but not for target consistency.

The Jacobian block linking the target equation to $\gamma$ satisfies
\begin{align*}
    J_{15}
    &=E\left[\left\{1-\frac{\delta}{q(S)}\right\}
       \partial_{\gamma^\top}
       g(S;\theta_{h,\mathrm{enr}}^*,\gamma^*)\right]
       =0_{r\times d_\gamma},
\end{align*}
because $E\{1-\delta/q(S)\mid S\}=0$.

Then, block inversion gives the following target-block influence function:
\begin{align*}
    &~ \phi^{\text{enr}}_h
    (O;\theta_{h,\mathrm{enr}}^*,\eta^*,\gamma^*) \\
    =&~-J_{11}^{-1}\left[
        \begin{aligned}
        &\frac{\delta}{q(S)}
        \left\{
          m_F(Z;\theta_{h,\mathrm{enr}}^*,\eta^*)
          -J_{12}J_{22}^{-1}u_\pi(Z;\alpha^*)\right.\\
        &\hspace{27mm}\left.
          -J_{13}J_{33}^{-1}u_1(Z;\beta^{1*})
          -J_{14}J_{44}^{-1}u_0(Z;\beta^{0*})
        \right\}\\
        &+\left(1-\frac{\delta}{q(S)}\right)
          g(S;\theta_{h,\mathrm{enr}}^*,\gamma^*)
        \end{aligned}
        \right].
\end{align*}
Note that the components of the Jacobian, $J_{11}, J_{12}, J_{13}, J_{14}, J_{22}, J_{33}, J_{44}$, are identical between the IPSW estimator and the enriched estimator for any choice of $m_F$.

For $m_F^{\text{dr}}$, suppose that all full-data nuisance models are correctly specified and that, at the population root, $g(S;\theta_{h,\mathrm{enr}}^*,\gamma^*)
=E\{m_F(Z;\theta_{h,\mathrm{enr}}^*,\eta^*)\mid S\}$.
The influence function then becomes
\begin{align*}
    \phi^{\text{enr}}_h(O;\theta_h,\eta)
        =&~ \frac{\delta}{q(S)}\phi_F(Z;\tau_h,\eta)
            + \left(1-\frac{\delta}{q(S)}\right)E[\phi_F(Z;\tau_h,\eta)\mid S],
\end{align*}
which equals the EIF in Proposition~\ref{prop:EIF}, denoted by
  $\phi^{\mathrm{eff}}(O;\tau_h,\eta)$.
 \begin{proposition}\label{prop:asympENR}
    Under the regularity conditions in Appendix~\ref{app:parametric-regularity},
    \begin{align*}
        \sqrt{n}\{\hat\theta^{\text{enr}}_h
        -\theta_{h,\mathrm{enr}}^*\}
        =&~\frac{1}{\sqrt n}\sum_{i=1}^n
        \phi^{\text{enr}}_h
        (O_i;\theta_{h,\mathrm{enr}}^*,\eta^*,\gamma^*)+o_p(1).
    \end{align*}
    Consequently,
    \begin{align*}
        \sqrt{n}\{\hat\theta^{\text{enr}}_h
        -\theta_{h,\mathrm{enr}}^*\}
        \overset{d}{\longrightarrow}&~
        N_r\!\left(0,
        \operatorname{Var}\{
        \phi^{\text{enr}}_h
        (O;\theta_{h,\mathrm{enr}}^*,\eta^*,\gamma^*)\}\right).
    \end{align*}
    Under the model-specification conditions stated above,
    $\theta_{h,\mathrm{enr}}^*=\theta_h$ and
    $\hat\theta_h^{\mathrm{enr}}\overset{p}{\longrightarrow}\theta_h$.
    If $m_F=m_F^{\text{dr}}$, all full-data nuisance working models are correctly
    specified, and $g(S;\theta_h,\gamma^*)=
    E\{m_F(Z;\theta_h,\eta^*)\mid S\}$, then
    \begin{align*}
        \sqrt{n}(\hat\tau^{\text{enr}}_h-\tau_h)
        \overset{d}{\longrightarrow}&~N(0,\mathcal{V}_h),\\
        \operatorname{AVar}(\hat\tau^{\text{enr}}_h)
            =&~ \frac{1}{n}\operatorname{Var}\{\phi^{\mathrm{eff}}(O;\tau_h,\eta)\}
            = \frac{\mathcal V_h}{n}.
    \end{align*}
\end{proposition}

\paragraph{Remark 1}
For ATE, the EDR estimator above is equal to the enriched doubly robust (EDR) estimator proposed by Wang et al. \cite{Wang2009-kn}.
They assumed that the propensity score follows a parametric model and, under this structure, derived the locally efficient (LE) estimator.
However, since the LE estimator requires iterative projections and is difficult to implement, the EDR estimator was proposed as a more practical alternative.
In contrast, we derive the semiparametric efficiency bound under a nonparametric model for the propensity score and demonstrate that the EDR estimator achieves this bound.

\subsection{Variance estimation}
The asymptotic variances of the IPSW and enriched estimators can be estimated by computing the sample variance of their corresponding influence functions and dividing it by $n$.
This variance estimator is equivalent to the so-called sandwich estimator in standard M-estimation \citep{White1980-sr,Stefanski2002-wx}, and it remains robust even when the nuisance models are misspecified.

\subsection{Correction of finite-sample bias}\label{sec:biascorr}
The IPSW and enriched estimators considered here are ratios (SDR and EDR)
or differences of two ratios (SIW and EIW). Ratio estimators are generally
not exactly unbiased in finite samples, even when their numerator and
denominator estimate their respective targets without bias. Under standard
regularity conditions, their leading bias is typically of order $n^{-1}$
and can be appreciable when the denominator is unstable or unequal inverse
weights yield a small effective sample size
\citep{Cochran1977-bs, Hartley1954-bn, Lohr2021-dl, Sarndal2003-ol}.

The ordinary delete-one jackknife targets this leading term by
recomputing the estimator after omitting each phase-1 unit
\citep{Quenouille1956-pn, Shao1995-xd}. If $\hat\theta_{(-i)}$ denotes the
estimator without unit $i$ and
$\bar\theta_{n-1}=n^{-1}\sum_{i=1}^n\hat\theta_{(-i)}$, its bias-corrected
form is $n\hat\theta_n-(n-1)\bar\theta_{n-1}$. Because this requires $n$
additional fits, it can be computationally expensive when the phase-1
sample is large.

We therefore use a computationally cheaper equal-group version of the
delete-$d$ jackknife. Randomly partition the phase-1 sample into
$D$ equal groups of size $d=n/D$, let $\hat\theta_{(-j)}$ denote the
estimator after deleting group $j$, and define
\[
 \bar\theta_{n-d}=\frac{1}{D}\sum_{j=1}^D\hat\theta_{(-j)}.
\]
Each replicate repeats the complete estimation procedure, including
refitting the nuisance and phase-1 conditional-mean models on the retained
units. If $E(\hat\theta_n)=\theta^*+b/n+O(n^{-2})$ and
$E(\hat\theta_{(-j)})=\theta^*+b/(n-d)+O\{(n-d)^{-2}\}$ with the same $b$,
the jackknife-corrected estimator is
\begin{equation}\label{eq:jackknife}
 \hat\theta_n^{\mathrm{JK}}
 =\frac{n}{d}\hat\theta_n-\frac{n-d}{d}\bar\theta_{n-d}
 =D\hat\theta_n-(D-1)\bar\theta_{n-d}.
\end{equation}
For fixed $D$, this correction cancels the leading $O(n^{-1})$ term, leaving
bias $O(n^{-2})$ about $\theta^*$ without changing the probability limit
$\theta^*$.

\section{Further Methodological Considerations}

\subsection{Optimal Sampling Probabilities}
The phase-2 sampling probability $q(s)$ can be determined by the researcher.
From the viewpoint of estimation efficiency, a reasonable design principle is
to choose $q(s)$ to minimize the semiparametric efficiency bound under a given
cost constraint on data collection.

In this section, as is the case in many practical applications, we assume that the support of $S$ consists of $K$ discrete values, and that $q(S)$ takes the values $q_1,\ldots,q_K$.
Let $P(S=k)=p_k~(k=1,\ldots,K)$ denote the marginal probabilities in the phase 1 sample, satisfying $\sum_{k=1}^K p_k = 1$, and let $\bar{q}$ denote the overall sampling probability for the phase 2 sample.
If the estimator admits an efficient influence function (EIF), the problem of determining the optimal sampling probabilities reduces to the following constrained optimization problem:
\begin{gather}\label{eq:optimA}
 \min_{(q_1,\ldots,q_K)\in(0,1]^K}
 \sum_{k=1}^Kp_k\frac{\sigma_k^2}{q_k}
 \quad\text{subject to}\quad
 \sum_{k=1}^Kp_kq_k\leq\bar q.
\end{gather}
where $\sigma_k^2 = \text{Var}(\phi_F(Z;\tau_h,\eta)\mid S=k)$.
The objective function corresponds to the component of $\mathcal{V}_h$ that depends on $q$. Since this is a convex optimization problem, it can be solved by the method of Lagrange multipliers.

\begin{proposition}\label{prop:optq}
Assume $p_k>0$ and $\sigma_k>0$ for every $k$, and that the overall phase-2
sampling fraction is sufficiently small that no upper bound $q_k\leq1$ binds.
Then the optimal sampling probabilities solving \eqref{eq:optimA} are
\begin{equation}\label{eq:optq-interior}
 q_k^*=\frac{\bar q\,\sigma_k}{\sum_{j=1}^Kp_j\sigma_j}.
\end{equation}
The corresponding minimum of the objective is
\[
 \bar q^{-1}\left(\sum_{k=1}^Kp_k\sigma_k\right)^2.
\]
\end{proposition}

This corresponds to Neyman allocation in the present setting; analogous influence-function-based allocations have been developed for design-based regression estimators under two-phase sampling \citep{Neyman1934-gc,Chen2022-os}.
If an unconstrained probability exceeds one, the allocation is capped and recalibrated; the Appendix gives the full solution. Under the conditions of Proposition~\ref{prop:optq}, proportional allocation relative to the phase-1 stratum sizes ($q_k$ constant) is optimal if and only if $\sigma_k$ is constant. Equal expected allocation ($p_kq_k$ constant) is optimal if and only if $p_k\sigma_k$ is constant.
Because $K$ is finite and every $\sigma_k$ is positive, the optimal design has $\min_k q_k^*>0$ and therefore satisfies the sampling-positivity condition in Section~\ref{sec:prelim1}.

At the planning stage, the stratum-specific conditional variances may be estimated from external or pilot data, possibly using a working model, and then used to set the phase-2 allocation \citep{Reilly1996-os,Gilbert2014-ot}.

The following special case gives analytic guidance for outcome-independent sampling while retaining treatment in the sampling strata. Write $v_a(x)=\operatorname{Var}(Y^a\mid X=x)$.

\begin{proposition}[Treatment-stratified allocation]\label{prop:simpledesign}
Suppose $S=(A,V)$, where $V$ has support $\{1,\ldots,K_V\}$, and write
\begin{align*}
 p_{a,k}&=P(A=a,V=k),\qquad q_{a,k}=q(A=a,V=k),\\
 \sigma_{a,k}^2&=\operatorname{Var}(\phi_F\mid A=a,V=k).
\end{align*}
Assume $\tau(X)=\tau_h$, $v_a(X)=\sigma^2>0$ for $a=0,1$,
and $\sigma_{a,k}>0$ for every $a\in\{0,1\}$ and $k\in\{1,\ldots,K_V\}$.
Then
\begin{align}
 \sigma_{1,k}^2
 &=C_h^{-2}\sigma^2
 E\left[\frac{h^2\{\pi(X)\}}{\pi^2(X)}\,\middle|\,A=1,V=k\right],
 \label{eq:sigma-treated-cell}\\
 \sigma_{0,k}^2
 &=C_h^{-2}\sigma^2
 E\left[\frac{h^2\{\pi(X)\}}{\{1-\pi(X)\}^2}\,\middle|\,A=0,V=k\right].
 \label{eq:sigma-control-cell}
\end{align}
Consequently, for an interior solution,
$q_{a,k}^*\propto\sigma_{a,k}$. If, in addition,
$\pi(X)=\pi(V)$ and $\pi_k=P(A=1\mid V=k)$, then
\begin{equation}\label{eq:optq-treatment-cells}
 q_{1,k}^*\propto\frac{h(\pi_k)}{\pi_k},
 \qquad
 q_{0,k}^*\propto\frac{h(\pi_k)}{1-\pi_k}.
\end{equation}
\end{proposition}

\begin{proof}
Under the constant-effect assumption, the full-data EIF reduces to
\[
 \phi_F=C_h^{-1}h\{\pi(X)\}
 \left\{
 \frac{A}{\pi(X)}(Y-\mu_1(X))
 -\frac{1-A}{1-\pi(X)}(Y-\mu_0(X))
 \right\}.
\]
Conditional unconfoundedness and homoskedasticity imply that its conditional
mean given $(A,X)$ is zero and that its conditional variance is
$C_h^{-2}\sigma^2h^2\{\pi(X)\}/\pi^2(X)$ for $A=1$ and
$C_h^{-2}\sigma^2h^2\{\pi(X)\}/\{1-\pi(X)\}^2$ for $A=0$.
Taking expectations over $X$ conditional on $(A,V)$ gives
\eqref{eq:sigma-treated-cell} and \eqref{eq:sigma-control-cell}; the allocation
statements then follow from Proposition~\ref{prop:optq} and $h\geq0$.
\end{proof}

For example, if $\pi(X)=\pi(V)$, then for the ATT, $h(\pi)=\pi$, and
\eqref{eq:optq-treatment-cells} gives $q_{1,k}^*\propto1$ and
$q_{0,k}^*\propto\pi_k/(1-\pi_k)$. Thus, when treatment is used in the
sampling strata, the optimal design need not assign the same probability to
treated and untreated units having the same $V$.

\subsection{When is the Enriched Estimator Beneficial?}\label{sec:benefit}

Let $n_k$ and $m_k$ denote the phase-1 and phase-2 sample sizes, respectively, for stratum $k\in\{1,\ldots,K\}$. Suppose phase-2 sampling is stratified SRSWOR, i.e., $m_k$ units are selected without replacement from each stratum $k$. Then the (design) sampling probability in stratum $k$ is $q_k=m_k/n_k$. In this case, the IPSW estimator and the enriched estimator are exactly equal in finite samples, so the simpler IPSW estimator suffices.

\begin{proposition} \label{prop:srswor}
    Under stratified SRSWOR in phase 2, the augmentation term vanishes:
    \begin{align*}
        \sum_{i=1}^n\left(\frac{\delta_i}{q(S_i)} - 1\right)\hat{g}(S_i)
        =&~ \sum_{k=1}^K\hat{g}(k)\sum_{i: S_i = k}\left(\frac{\delta_i}{q_k} - 1\right) \\
        =&~ \sum_{k=1}^K\hat{g}(k)\left(m_k\cdot\frac{n_k}{m_k} - n_k\right)=0,
    \end{align*}
    where $\hat{g}$ is an arbitrary function of $S$.
\end{proposition}
Here the inverse sampling weights exactly calibrate to the phase-1
stratum indicators, so the identity is the saturated-stratum case of
the calibration/regression-estimation connection
\citep{Deville1992-ds,Lumley2011-lsd}.

Under phase-2 Poisson sampling with known $q(S)$, suppose that the nuisance and projection models are correctly specified.
The asymptotic variance gain of EDR over SDR is
  \begin{equation}\label{eq:enrichment-gain}
    G(\phi_F) = E\left[\left(\frac{1}{q(S)} - 1\right)E[\phi_F\mid S]^2\right].
\end{equation}
This gain is nonnegative, increases as $q(S)$ decreases or
$E[\phi_F\mid S]^2$ increases, and is zero precisely when
$E[\phi_F\mid S]=0$ almost surely on the set where $q(S)<1$.
EIW is not guaranteed to improve on SIW, although the simulations show substantial gains
for some ODS settings.

Conditioning the full-data EIF in Proposition~\ref{prop:EIF} on $S=s$
gives two interpretable components. Suppressing the argument $X$ of $\pi_a$
and $\mu_a$,
\begin{align*}
E[\phi_F\mid S=s]
    &= \underbrace{
        C_h^{-1}E\!\left[
            h(\pi_1)\left\{
                \frac{A(Y-\mu_1)}{\pi_1}
                - \frac{(1-A)(Y-\mu_0)}{\pi_0}
            \right\}\,\middle|\, S=s\right]
    }_{\text{(i)~IPTW residual mean}}\\
    &\quad+ \underbrace{
        C_h^{-1}E\!\left[
            \{h(\pi_1)+h'(\pi_1)(A-\pi_1)\}
            \{\tau(X)-\tau_h\}\mid S=s\right]
    }_{\text{(ii)~effect-heterogeneity mean}}.
\end{align*}
Enrichment is therefore most valuable when $S$ captures residual variation in
potential outcomes or treatment-effect heterogeneity, or both.
In particular, observing $Y$ at phase~1 can make this projection more
informative and yield a substantial reduction in asymptotic variance.

\subsection{Double Machine Learning Formulation of the Proposed Estimators}\label{sec:DML}

As shown in Section~\ref{sec:enrich}, under correct specification the EDR
estimator has the efficient influence function for the WATE under two-phase
sampling. The corresponding estimating function is Neyman orthogonal
\citep{Neyman1959,Neyman1979-ek}. In this section, we consider settings in
which nuisance functions are estimated using data-adaptive approaches such as
random forests and Lasso regression, and reformulate the EDR estimator within
the framework of cross-fitted double machine learning
\citep{Chernozhukov2018-cc}.

Write $\eta_F=(\pi,\mu_0,\mu_1)$ for the full-data nuisance functions and
let $\eta=(\eta_F,g_N,g_D)$, where $g_N$ and $g_D$ denote phase-1
conditional-mean models for $N_F$ and $D_F$, respectively. The estimating
function of the EDR estimator is linear in the target parameter. Using the
EDR notation $m_F^{\mathrm{dr}}=N_F-\tau_hD_F$ introduced above, we rewrite
the observed-data estimating function as
\begin{align*}
    m(O;\tau_h,\eta)
        =&~ N(O;\eta)-\tau_hD(O;\eta), \\
    N(O;\eta)
        =&~ \frac{\delta}{q(S)}N_F(Z;\eta_F)
        + \left(1-\frac{\delta}{q(S)}\right)g_N(S), \\
    D(O;\eta)
        =&~ \frac{\delta}{q(S)}D_F(Z;\eta_F)
        + \left(1-\frac{\delta}{q(S)}\right)g_D(S),\\
    N_F(Z;\eta_F)
        =&~ h\{\pi_1(X)\}\left\{
            \frac{\mathbf{1}^A_1}{\pi_1(X)}(Y-\mu_1(X))
            - \frac{\mathbf{1}^A_0}{\pi_0(X)}(Y-\mu_0(X))
            + \tau(X)
        \right\} \\
        &\qquad + h'\{\pi_1(X)\}\tau(X)(A-\pi_1(X)),\\
    D_F(Z;\eta_F)
        =&~ h\{\pi_1(X)\} + h'\{\pi_1(X)\}(A-\pi_1(X)).
\end{align*}

Let $\eta_{F0}$ denote the true full-data nuisance functions and define
\begin{align*}
    g_{N0}(s)&=E\{N_F(Z;\eta_{F0})\mid S=s\},\\
    g_{D0}(s)&=E\{D_F(Z;\eta_{F0})\mid S=s\},\\
    \eta_0&=(\eta_{F0},g_{N0},g_{D0}).
\end{align*}
Throughout the DML theory, the projection models are assumed to be correctly
specified and their estimators to converge to $g_{N0}$ and $g_{D0}$.
When $S$ has fixed finite support and every cell has positive probability,
this condition can be satisfied by a saturated cell model.

Since the sampling probability $q(S)$ is known and
$E(\delta\mid Z)=q(S)$, the EDR estimating function preserves the
unbiasedness of the full-data estimating function at the truth:
\begin{gather*}
    E\{m(O;\tau_h,\eta_0)\}
    =E\{m_F^{\mathrm{dr}}(Z;\tau_h,\eta_{F0})\}=0.
\end{gather*}

\begin{definition}[Double Machine Learning for WATE under Two-phase Sampling]\label{def:DML}
    The algorithm proceeds as follows for a fixed integer $K\ge2$.
    \begin{enumerate}
        \item Randomly partition the phase-1 indices into $K$ folds, with each phase-2 unit inheriting its phase-1 fold. When $S$ is discrete with few levels, partition within its strata. Let $I_k$ denote fold $k$ and
         $I_{-k}$ its complement.
        \item For $k = 1,\ldots,K$, perform the following steps in order:
        \begin{enumerate}
            \item Using the phase-2 observations in $I_{-k}$, estimate the full-data nuisance functions with inverse-sampling weights, yielding $\hat\eta_F^{-k}$, and compute $N_F(Z_i;\hat\eta_F^{-k})$ and
            $D_F(Z_i;\hat\eta_F^{-k})$ for those observations.
            \item Fit the phase-1 projection models to these training-fold pseudo-outcomes by inverse-sampling-weighted least
            squares. Denote the resulting functions by $\hat g_N^{-k}$ and $\hat g_D^{-k}$, and set $\hat\eta^{-k}=(\hat\eta_F^{-k},\hat g_N^{-k},\hat g_D^{-k})$.
            \item For $i\in I_k$, compute $N(O_i;\hat\eta^{-k})$ and $D(O_i;\hat\eta^{-k})$.
        \end{enumerate}
        \item Solve the following estimating equation and obtain the estimator of the WATE:
        \begin{gather*}
            \hat\tau^{\text{dml}}_h
            = \frac{
            \sum_{k=1}^K\sum_{i\in I_k}N(O_{i},\hat{\eta}^{-k})
            }{
            \sum_{k=1}^K\sum_{i\in I_k}D(O_{i},\hat{\eta}^{-k})
            }.
        \end{gather*}
    \end{enumerate}
\end{definition}

As in the original DML framework, we impose the nuisance-rate and regularity
conditions given in the Appendix. These conditions cover both the full-data
nuisance functions and the separately fitted projection functions.

\begin{proposition}[Asymptotic properties of DML-EDR] \label{prop:DML}
    Under Assumptions~\ref{assum:Cherno31} and
\ref{assum:Cherno32}, the estimator
$\hat{\tau}_h^{\text{dml}}$ has the asymptotic linear representation
    \begin{gather*}
    \sqrt{n}(\hat{\tau}_h^{\text{dml}} - \tau_h)
    =C_h^{-1}\frac{1}{\sqrt n}\sum_{i=1}^n
      m(O_i;\tau_h,\eta_0)+o_p(1).
    \end{gather*}
Consequently,
    \begin{align*}
    \sqrt{n}(\hat{\tau}_h^{\text{dml}} - \tau_h)
    \overset{d}{\rightarrow}&~ N(0,\mathcal{V}_h), \\
\mathcal{V}_h =&~ C_h^{-2}
    \operatorname{Var}\{m(O;\tau_h,\eta_0)\}, \\
C_h =&~ E\{D(O;\eta_0)\}.
\end{align*}
The DML estimator of the asymptotic variance
\begin{gather*}
    \hat{\mathcal{V}}_h = \hat{C}_{h}^{-2}\frac{1}{n}\sum_{k=1}^K\sum_{i\in I_k}m(O_{i};\hat{\tau}_h^{\text{dml}},\hat{\eta}^{-k})^2,
    \quad \text{with}~
    \hat{C}_{h} = \frac{1}{n}\sum_{k=1}^K\sum_{i\in I_k} D(O_{i};\hat{\eta}^{-k})
\end{gather*}
is consistent for $\mathcal V_h$.
Consequently, for any fixed $\alpha\in(0,1)$, the Wald-type confidence interval
\begin{gather*}
    \left[
    \hat{\tau}_h^{\mathrm{dml}}
    \pm
    z_{1-\alpha/2}
    \sqrt{\hat{\mathcal V}_h/n}
    \right]
\end{gather*}
has asymptotic coverage probability $1-\alpha$.
Moreover,
$C_h^{-1}m(O;\tau_h,\eta_0)=
\phi^{\mathrm{eff}}(O;\tau_h,\eta_{F0})$.
\end{proposition}

\section{Simulation}\label{sec:sim}
\subsection{Performance of IPSW and Enriched Estimators}
\subsubsection{Settings}
We use DGP~1 (Appendix~\ref{app:dgp01}), in which eight independent binary covariates $V_1,\ldots,V_8$ and one continuous covariate $W_c$ determine a binary treatment $A$ and binary potential outcomes $(Y^1,Y^0)$.

In phase 1, under ODS, the stratum variable is $S=(A,V_{\text{obs}},Y)$; under non-ODS, $S=(A,V_{\text{obs}})$.
A single $V_j$ is designated as $V_{\text{obs}}$ in each scenario, yielding eight strata under ODS and four under non-ODS.
The other seven binary covariates and $W_c$ are observed only at phase~2.

Phase-2 sampling was conducted via two procedures: stratified SRSWOR and Poisson sampling (details in Appendix~\ref{app:dgp01}).

The target parameters were the population-level ATE, ATT, ATU, and ATO. The true values computed from a large reference sample were $0.165$, $0.212$, $0.146$, and $0.198$, respectively.
The estimators evaluated in this section were SIW, EIW, SDR, and EDR. The propensity-score model and both treatment-specific outcome models were correctly specified logistic regressions with covariates $V_1, \ldots, V_8, W_c$.
For each estimator we also report a jackknife-corrected version from Section~\ref{sec:biascorr}, using a fixed $D=10$ and an unstratified random partition into equal-size deletion groups.

Let $\hat\tau_h$ denote the estimated value and $\tau_h$ the true value obtained from the large sample, and let $E^*$ denote the expectation over the Monte Carlo replications.
The estimators were evaluated using the following metrics: $\text{Bias} = E^*[\hat\tau_h] - \tau_h$, $\text{RMSE} = \sqrt{E^*[(\hat\tau_h - \tau_h)^2]}$, and $\text{EmpSE} = \sqrt{E^*[(\hat\tau_h - E^*[\hat\tau_h])^2]}$.
We also report EmpSE and RMSE relative to an oracle that uses the true $P(A=1\mid X)$, $Y^1$, and $Y^0$ for every phase-1 unit. To quantify enrichment gains for the IPTW- and DR-type estimators, we computed
\begin{gather*}
    \%\text{Gain} = 100\times \frac{\text{EmpSE(IPSW)} - \text{EmpSE(Enriched)}}{\text{EmpSE(IPSW)}}.
\end{gather*}

To assess the asymptotic behavior, the phase-2 sample size was set to $m\in\{200, 500, 1{,}000, 2{,}000\}$, and the phase-1 sample size was set to $n\in\{4m, 10m\}$.
In total, 256 simulation scenarios were conducted, combining eight sample-size configurations, four sampling schemes, and eight choices of $V_{\text{obs}}$.
Each simulation scenario was replicated $1{,}000$ times.

\subsubsection{Results}
Figure~\ref{fig:PShist} shows limited propensity-score overlap; the true treatment prevalence was approximately 28--29\%.
\begin{figure}[htbp]
    \centering
    \includegraphics[width=0.75\linewidth]{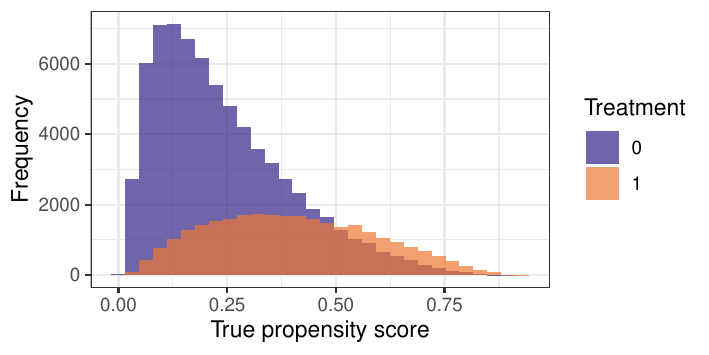}
    \caption{Histogram of the true propensity score by treatment group.}
    \label{fig:PShist}
\end{figure}

Consistent with Proposition~\ref{prop:srswor}, the uncorrected IPSW and enriched estimators agreed exactly under stratified SRSWOR.
We report the Poisson-sampling results here and the SRSWOR results in Appendix~\ref{sec:addsim}.

Under ODS, bias decreased with sample size, was largest for ATE and ATU and for SIW and EIW, and was reduced by enrichment and the delete-$d$ jackknife (Figure~\ref{fig:Bias_p010_poisson}).

\begin{figure}[htbp]
\centering
\includegraphics[width=0.95\linewidth]{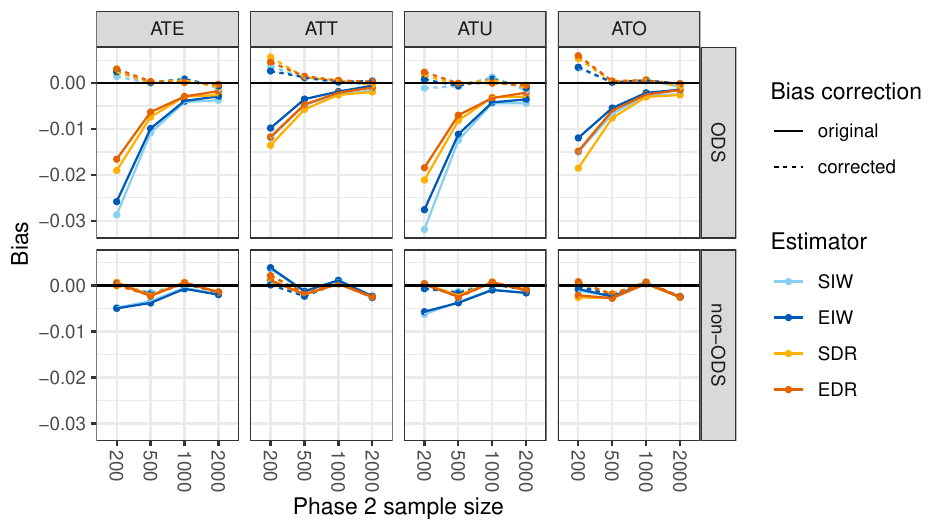}
\caption{Bias of each estimator by sample size under Poisson sampling with $V_{\text{obs}}=V_1$ and $n=10m$.}
\label{fig:Bias_p010_poisson}
\end{figure}

Empirical SE and RMSE also decreased with sample size.
At $n=10{,}000$ and $m=1{,}000$, EDR had the smallest ATE estimation error under ODS for every choice of $V_{\text{obs}}$
  (Table~\ref{tab:tab:ATE_n10000_m1000_poisson_ODS}).
Coverage was slightly below nominal overall, less so under non-ODS; details are in Appendix~\ref{sec:addsim}.

\begin{table}[tbp]
\centering
\caption{\label{tab:tab:ATE_n10000_m1000_poisson_ODS}Performance of each estimator when targeting the ATE under $n = 10{,}000$, $m = 1{,}000$, outcome-dependent Poisson sampling. For each cell of $V_{\text{obs}}$ and estimator, the three numbers shown (from top to bottom) represent the relative RMSE, relative empirical SE, and bias ($\times 10^2$).}
\centering
\begin{tabular}[t]{ccccccccc}
\toprule
\multicolumn{1}{c}{ } & \multicolumn{4}{c}{Original} & \multicolumn{4}{c}{Bias-corrected} \\
\cmidrule(l{3pt}r{3pt}){2-5} \cmidrule(l{3pt}r{3pt}){6-9}
$V_{\text{obs}}$ & SIW & EIW & SDR & EDR & SIW & EIW & SDR & EDR\\
\midrule
$V_1$ & 5.35 & 4.10 & 5.01 & 3.53 & 5.34 & 4.07 & 5.02 & 3.54\\
 & 5.31 & 4.05 & 4.99 & 3.50 & 5.33 & 4.07 & 5.02 & 3.54\\
 & -0.41 & -0.39 & -0.29 & -0.29 & 0.10 & 0.08 & 0.06 & 0.01\vspace{2mm}\\
$V_2$ & 5.65 & 4.46 & 5.24 & 3.87 & 5.61 & 4.40 & 5.26 & 3.89\\
 & 5.59 & 4.40 & 5.21 & 3.83 & 5.61 & 4.40 & 5.26 & 3.88\\
 & -0.53 & -0.48 & -0.39 & -0.33 & 0.03 & 0.04 & 0.07 & 0.07\vspace{2mm}\\
$V_3$ & 5.33 & 4.08 & 5.02 & 3.51 & 5.26 & 4.01 & 4.99 & 3.48\\
 & 5.23 & 4.01 & 4.94 & 3.45 & 5.25 & 4.01 & 4.98 & 3.48\\
 & -0.69 & -0.53 & -0.59 & -0.42 & -0.13 & -0.02 & -0.13 & -0.02\vspace{2mm}\\
$V_4$ & 5.24 & 4.06 & 4.98 & 3.60 & 5.17 & 4.01 & 4.94 & 3.58\\
 & 5.17 & 4.01 & 4.93 & 3.56 & 5.17 & 4.01 & 4.94 & 3.59\\
 & -0.58 & -0.45 & -0.50 & -0.36 & -0.05 & 0.03 & -0.12 & -0.04\vspace{2mm}\\
$V_5$ & 5.67 & 4.50 & 5.23 & 3.90 & 5.55 & 4.38 & 5.18 & 3.85\\
 & 5.55 & 4.38 & 5.14 & 3.80 & 5.55 & 4.37 & 5.17 & 3.85\\
 & -0.75 & -0.67 & -0.63 & -0.54 & -0.16 & -0.13 & -0.15 & -0.12\vspace{2mm}\\
$V_6$ & 5.60 & 4.34 & 5.27 & 3.88 & 5.56 & 4.31 & 5.26 & 3.88\\
 & 5.57 & 4.31 & 5.24 & 3.84 & 5.56 & 4.31 & 5.27 & 3.89\\
 & -0.41 & -0.39 & -0.38 & -0.35 & 0.17 & 0.14 & 0.09 & 0.06\vspace{2mm}\\
$V_7$ & 5.26 & 4.14 & 5.03 & 3.62 & 5.20 & 4.04 & 5.02 & 3.59\\
 & 5.19 & 4.05 & 4.98 & 3.57 & 5.20 & 4.04 & 5.02 & 3.59\\
 & -0.60 & -0.58 & -0.50 & -0.44 & -0.03 & -0.06 & -0.04 & -0.04\vspace{2mm}\\
$V_8$ & 5.51 & 4.43 & 5.24 & 4.06 & 5.46 & 4.37 & 5.26 & 4.09\\
 & 5.45 & 4.36 & 5.21 & 4.03 & 5.46 & 4.37 & 5.26 & 4.08\\
 & -0.51 & -0.47 & -0.37 & -0.30 & 0.08 & 0.07 & 0.12 & 0.12\vspace{2mm}\\
\bottomrule
\end{tabular}
\end{table}

At $m=1{,}000$, enrichment gains were substantially larger under ODS, increased as the phase-2 sampling fraction decreased, and were somewhat larger for the outcome-related $V_1$ and $V_3$. Under non-ODS, enrichment provided little
benefit, and EIW occasionally had higher empirical SE than SIW
(Figure~\ref{fig:gain_cor_1000_poisson}).

\begin{figure}[htbp]
\centering
\includegraphics[width=0.95\linewidth]{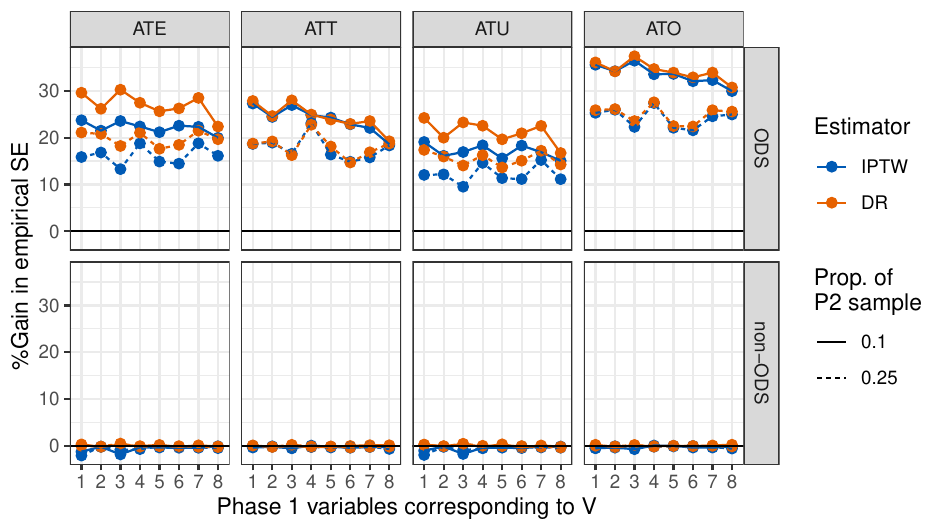}
\caption{Efficiency gains from enrichment by the choice of $V_{\text{obs}}$ under Poisson sampling. The estimator is bias-corrected, and the phase-2 sample size is 1{,}000.}
\label{fig:gain_cor_1000_poisson}
\end{figure}

\clearpage
\subsection{\texorpdfstring{Validation of ML-based EDR estimation}{Validation of ML-based EDR estimation}}
\subsubsection{Settings}
We evaluate the EDR estimator when all nuisance functions are estimated nonparametrically.
Specifically, we use SuperLearner \citep{van-der-Laan2007-uv} with the library $\{\texttt{SL.glm}, \texttt{SL.ranger}\}$, consisting of logistic regression and random forests.
This simulation study assesses whether the estimator exhibits the expected $\sqrt{n}$-consistent and asymptotically normal behavior when nuisance functions are learned from data, and whether cross-fitting is needed to obtain valid confidence interval coverage.

We use DGP~2 (Appendix~\ref{app:dgp02}), with ten binary latent covariates, one binary phase-1 covariate, and no treatment effect; hence all four target parameters are zero. Phase-2 sampling follows ODS with Poisson sampling at $\bar q=0.25$.

Two covariate scenarios are considered:
\begin{itemize}
  \item \textbf{Linear:} $W_j=Z_j$ for $j=1,\ldots,10$, giving additive generating predictors in the observed covariates.
  \item \textbf{Nonlinear:} invertible pair encodings and higher-order products of the latent covariates give a more complex observed-covariate representation for the nuisance learners.
\end{itemize}
In both scenarios, the full covariate vector supplied to SuperLearner is
    \[
    X=(W_1,\ldots,W_{10},V).
\]
In the nonlinear scenario, $W_1,\ldots,W_5$ are additionally expanded into dummy variables.

Two configurations are compared. We refer to the  non-cross-fitted version
as \textbf{ML-EDR (w/o CF)} and the five-fold cross-fitted version (Definition~\ref{def:DML}) as \textbf{DML-EDR ($K=5$)}.
The phase-1 sample size ranges over $n\in\{1{,}000,2{,}000,4{,}000,8{,}000\}$.
All 1{,}000 replications were used without outlier removal.
Performance is evaluated via Bias, EmpSE, MeanSE, and the empirical coverage of the 95\% Wald confidence interval (nominal level 0.95).

\subsubsection{Results}\label{sec:sim2_res}

Table~\ref{tab:dml_coverage} reports the empirical coverage of 95\% Wald confidence intervals and, in the rightmost column of each panel, the ratio MeanSE/EmpSE for the ATE estimand.
A ratio below 1 indicates that the standard error estimator underestimates the true sampling variability and can contribute to undercoverage.

For ML-EDR (w/o CF), the SE ratio falls below 1 at small $n$,
reaching $0.78$ at $n=1{,}000$ in the linear scenario
and $0.65$ in the nonlinear scenario, reflecting overfitting when nuisance functions are fitted and evaluated on the same sample. The ratio approaches 1 as $n$ grows, and coverage at $n=8{,}000$ ranges from 93.3\% to 95.3\%.

At $n=1{,}000$, the coverage probabilities for ATE and ATU are 84.9\% and 84.0\%, respectively, in the linear scenario, and decrease to 77.8\% and 76.6\% in the nonlinear scenario.
ATO coverage was higher at 92.8\% and 91.4\%, consistent with overlap weights $h\{\pi(X)\}=\pi(X)\{1-\pi(X)\}$ down-weighting units with extreme propensity scores.

Five-fold cross-fitting in DML-EDR ($K=5$) mitigates the standard-error underestimation observed for ML-EDR (w/o CF).
The ATE SE ratio exceeds 1 in all configurations, ranging from $1.01$ to $1.03$ at $n=1{,}000$, consistent with improved calibration and coverage near the nominal level.
For DML-EDR ($K=5$), the coverage probabilities range from 91.7\% to 94.4\% at $n=1{,}000$ and from 94.5\% to 96.8\% at $n=8{,}000$.
In the nonlinear scenario at $n=1{,}000$, DML-EDR ($K=5$) has small upward bias across estimands (maximum: $0.010$), with about
$\bar q n(K-1)/K=200$ phase-2 observations per nuisance-training set.
The bias decreases rapidly with $n$ and is negligible at $n\geq4{,}000$.

Relative to $K=1$, the $K=5$ estimator had higher EmpSE in small samples, with increases of up to 81\% at $n=1{,}000$; by $n=8{,}000$, however, the increase was below 6\% for every estimand in both scenarios, consistent with the asymptotically vanishing cost of sample splitting
\citep{Chernozhukov2018-cc}.

Full results are given in Appendix~\ref{sec:addsim_dgp2}.

\begin{table}[htbp]
  \centering
  \caption{Empirical coverage (\%) of the 95\% Wald confidence interval and MeanSE/EmpSE ratio (for ATE) for the ML-based EDR estimators under DGP~2.
  Here $K=1$ denotes ML-EDR (w/o CF), and $K=5$ denotes DML-EDR ($K=5$). True parameter values are zero for all estimands. MeanSE/EmpSE~$< 1$ indicates standard error underestimation.}
  \label{tab:dml_coverage}
  \vspace{2mm}
  \small
  \setlength{\tabcolsep}{4.5pt}
  \begin{tabular}{rr ccccc c ccccc}\toprule
    & & \multicolumn{5}{c}{Linear ($W_j = Z_j$)}
    & \phantom{x}
    & \multicolumn{5}{c}{Nonlinear (transformed $W$)} \\
    \cmidrule{3-7}\cmidrule{9-13}
    $n$ & $K$ & ATE & ATT & ATU & ATO & \multicolumn{1}{c}{SE ratio}
             && ATE & ATT & ATU & ATO & \multicolumn{1}{c}{SE ratio} \\\midrule
    \multirow[t]{2}{*}{1{,}000}
      & 1 & 84.9 & 84.8 & 84.0 & 92.8 & 0.78
          && 77.8 & 77.9 & 76.6 & 91.4 & 0.65 \\
      & 5 & 94.0 & 93.5 & 93.6 & 93.9 & 1.03
          && 94.4 & 93.9 & 93.5 & 91.7 & 1.01 \\\addlinespace
    \multirow[t]{2}{*}{2{,}000}
      & 1 & 91.5 & 91.0 & 91.3 & 95.4 & 0.92
          && 88.8 & 87.6 & 87.6 & 95.1 & 0.83 \\
      & 5 & 96.0 & 95.2 & 95.1 & 94.7 & 1.06
          && 96.0 & 94.8 & 95.4 & 93.6 & 1.06 \\\addlinespace
    \multirow[t]{2}{*}{4{,}000}
      & 1 & 93.3 & 92.7 & 91.5 & 95.7 & 0.95
          && 92.0 & 90.3 & 90.2 & 95.1 & 0.91 \\
      & 5 & 96.1 & 95.5 & 94.4 & 95.5 & 1.03
          && 95.6 & 94.6 & 95.6 & 94.8 & 1.04 \\\addlinespace
    \multirow[t]{2}{*}{8{,}000}
      & 1 & 94.2 & 93.4 & 94.3 & 95.3 & 0.97
          && 94.1 & 93.3 & 93.7 & 94.8 & 0.96 \\
      & 5 & 95.2 & 94.5 & 95.5 & 95.4 & 1.01
          && 96.4 & 94.6 & 96.8 & 95.2 & 1.03 \\
    \bottomrule
  \end{tabular}
\end{table}

\section{Real Data Analysis}\label{sec:realdata}

\subsection{Settings}

We use data from the observational study of Connors et al.
\cite{Connors1996-hp}, who investigated the effect of right heart catheterization (RHC) in critically ill patients, as a semi-synthetic illustration. The original study did not use a two-phase design; here the complete dataset is treated as a phase-1 cohort and a two-phase structure is imposed computationally.
The publicly available dataset contains $n = 5{,}735$ patients enrolled across five US academic medical centers between 1989 and 1994.
The treatment variable $A$ is a binary indicator of whether RHC was performed during the first 24 hours after study entry
($A = 1$; $n = 2{,}184$, $38.1\%$), versus no RHC ($A = 0$; $n = 3{,}551$, $61.9\%$).
The outcome $Y$ is 30-day mortality, which occurred in 1{,}918 patients ($33.4\%$, crude rates $30.6\%$ vs.\ $38.0\%$ in the no-RHC and RHC groups, respectively).
Recorded patient characteristics are summarized descriptively in Table~\ref{tab:table1} (Appendix~\ref{sec:add_real}).

Phase-1 variables $V$ were chosen to be inexpensive, routinely available categorical covariates ascertainable at admission without laboratory tests:
age group ($<55$ / $55$--$70$ / $\geq 70$),
sex (Female / Male),
and a three-level primary-disease cluster (ARF/Sepsis/MOSF,
Cardiac/Resp, Other) obtained by collapsing the original nine-category primary-diagnosis variable.
This yields $3 \times 2 \times 3 = 18$ unique $V$ cells.
Because the outcome is treated as available at phase~1, the phase-1 information is $S=(A,Y,V)$.
Phase-2 subsampling was carried out using Poisson sampling stratified by $A \times Y \times$ age group ($2 \times 2 \times 3 = 12$ strata), at four nominal sampling fractions
$\bar{q} \in \{0.10,\, 0.25,\, 0.50,\, 1.00\}$.
The corresponding phase-2 sample sizes were $m=584$, $1{,}438$, $2{,}847$, and $5{,}735$, respectively.

The same prespecified nuisance-feature set was used in the
parametric and DML analyses. It comprised 27 complete variables: exact age; sex; the three-level primary-disease cluster; cancer status; race, education, income, and insurance; 11 comorbidity or transfer-history indicators; and eight admission-diagnosis indicators.

Six estimators were compared.
\begin{itemize}
  \item \textbf{Parametric-SIW, EIW, SDR, and EDR}: the propensity score and, for SDR and EDR, both outcome regressions were fitted by standard unpenalized inverse-sampling-weighted logistic regression with main effects. EIW and EDR additionally fitted the numerator and denominator enrichment projections on $S$. Standard errors were obtained from the full stacked M-estimating equation, including all fitted nuisance and projection parameters.
  \item \textbf{ML-EDR (w/o CF)}: nuisance functions were estimated by SuperLearner \citep{van-der-Laan2007-uv} with the library
    \[
    \begin{aligned}
      \mathcal L=\{&
      \texttt{SL.glm},\,\texttt{SL.ranger},\\
      &\texttt{SL.ranger.conservative},\,\texttt{SL.ranger.wide}\}.
    \end{aligned}
    \]
    This combines logistic regression with three random-forest variants: the default \texttt{SL.ranger};
    \texttt{SL.ranger.conservative}, with reduced \texttt{mtry}
    ($\lfloor\sqrt{p}/2\rfloor$) and a larger minimum-node size; and \texttt{SL.ranger.wide}, with \texttt{mtry} $=p$.
    Nuisance models were trained and evaluated on the same phase-2 sample ($K = 1$, no cross-fitting).
  \item \textbf{DML-EDR ($K = 5$)}: same SuperLearner library as ML-EDR (w/o CF), with five-fold cross-fitting.
\end{itemize}
For the two parametric enriched estimators and both ML-based EDR estimators, the numerator and denominator projections used a 15-column model comprising a saturated factor for the 12 $A\times Y\times$ age-group sampling cells and additive main effects for sex and the three-level primary-disease cluster.
All estimators were applied to the four target parameters ATE, ATT, ATU, and ATO.
Point estimates and 95\% confidence intervals as a function of $\bar{q}$ are displayed in Figure~\ref{fig:rda_estimates}.
The ratio of each estimated standard error to its full-data benchmark ($\bar{q} = 1$) is shown in Figure~\ref{fig:rda_se_inflation}.

\subsection{Results}
At $\bar q=1$, the six estimators gave similar positive estimates ($0.095$--$0.110$ across estimands and methods), with all intervals excluding zero. Estimates varied more across methods and estimands under subsampling; over all $\bar q<1$ cells they ranged from $-4.9\times10^{-5}$ to $0.181$.

Relative to full data, standard-error ratios exceeded one
at every subsampling fraction (Figure~\ref{fig:rda_se_inflation}), ranging from $1.09$ to $1.37$ at
$\bar q=0.50$, from $1.14$ to $1.95$ at $\bar q=0.25$, and from $1.09$ to $3.54$ at $\bar q=0.10$. Within the two parametric simple/enriched pairs, the enriched-to-simple standard-error ratios were $0.84$--$0.91$, $0.70$--$0.83$, and $0.61$--$0.80$, respectively. Thus, all paired ratios
were below one, while absolute standard errors increased with more aggressive subsampling.

\begin{figure}[htbp]
\centering
\includegraphics[width=\linewidth]{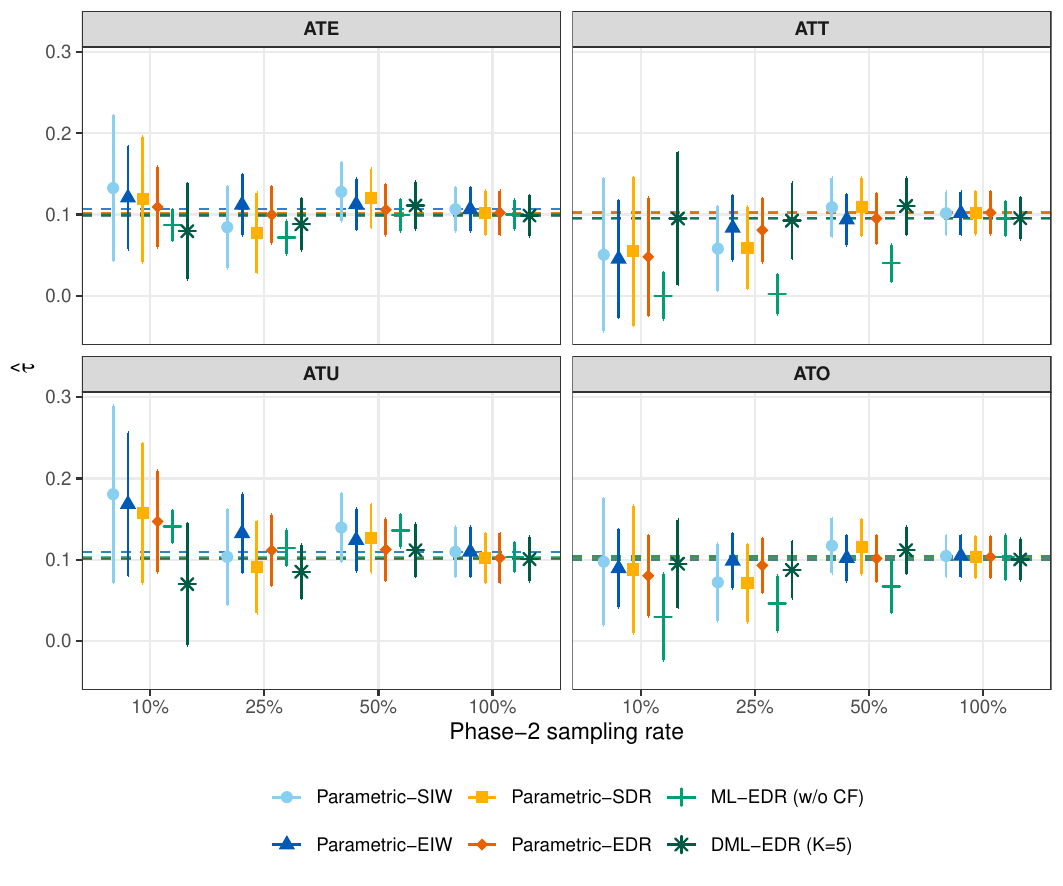}
\caption{Point estimates and 95\% confidence intervals for ATE, ATT, ATU, and ATO as a function of the phase-2 sampling fraction $\bar{q}$.
Parametric intervals use the full stacked M-estimation sandwich; ML-based intervals use the target estimating function. The dashed horizontal lines indicate the method-specific full-data ($\bar{q}=1$) estimates.}
\label{fig:rda_estimates}
\end{figure}

\begin{figure}[htbp]
  \centering
  \includegraphics[width=\linewidth]{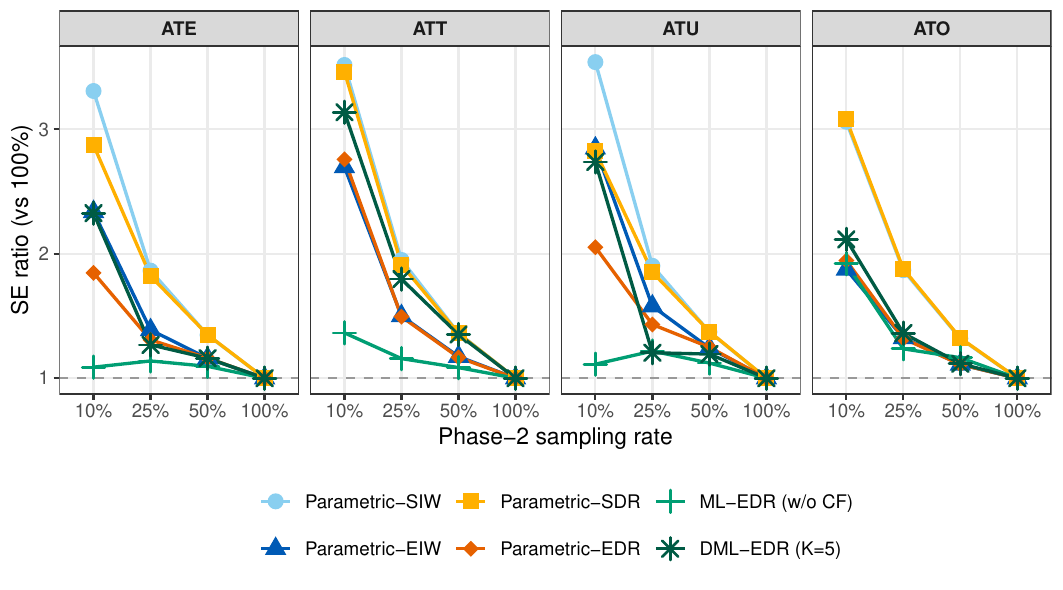}
  \caption{Ratio of the estimated standard error at phase-2 sampling fraction $\bar{q}$ to the corresponding full-data ($\bar{q} = 1$) standard error, for each estimand and method. The dashed horizontal line at 1 indicates no inflation.}
  \label{fig:rda_se_inflation}
\end{figure}

  \clearpage

\section{Discussion}

This paper derived the observed-data efficient influence function and the
semiparametric efficiency bound for a differentiable class of
propensity-score-weighted WATEs under known conditionally independent Poisson
phase-2 sampling, imposing no parametric restrictions on the full-data
distribution, and thereby extended the ATE results of Wang et al.
\cite{Wang2009-kn} to a broad family of estimands that includes the ATE, ATT,
ATU, and ATO. Building on this bound, we proposed IPSW and EDR estimators,
characterized the model-specification conditions under which they target the
WATE of interest, and showed that the EDR estimator attains the bound when the
required full-data nuisance and phase-1 conditional-mean models are consistently
estimated. We also obtained a closed-form expression for its efficiency gain
together with a decomposition of $E(\phi_F\mid S)$ into an IPTW-residual
component and an effect-heterogeneity component
(Section~\ref{sec:benefit}). A double-machine-learning formulation with
cross-fitting was further shown to preserve $\sqrt{n}$-asymptotic normality and
efficiency under data-adaptive nuisance estimation satisfying the stated
positivity, moment, and rate conditions.

These results carry direct implications for the design of two-phase
observational studies. The integrand of the enrichment-gain formula is
$\{q(S)^{-1}-1\}E(\phi_F\mid S)^2$, so phase-1 variables are most valuable when
they predict residual variation in the potential outcomes or capture
treatment-effect heterogeneity. For a fixed conditional projection, the gain
attributable to enrichment increases as $q(S)^{-1}$; however, smaller inclusion
probabilities also reduce phase-2 information, so a larger gain relative to the
corresponding unenriched estimator need not imply better absolute precision.
A particularly important special case arises when the outcome is inexpensive
and can be recorded at phase~1: under such outcome-dependent sampling,
including $Y$ in $S$ can make the conditional projection more informative and
yield substantial variance reductions in some settings, but this improvement
is not guaranteed. Consistent with this distinction, our simulations showed
substantial gains in some outcome-dependent settings and little gain in others,
whereas the RHC illustration showed both the precision cost of subsampling and,
within the parametric pairs, the gains attainable through enrichment. These findings suggest that
outcome-predictive and effect-modifying variables may be promising phase-1
measurements, subject to measurement costs and sampling constraints.

We emphasize that the efficiency theory, which includes the efficient influence
function, the bound, and the gain formula, holds for general $S$, whether
discrete or continuous, under known conditionally independent Poisson phase-2
sampling. The restriction to discrete $S$ with finite support and adequate
stratum sizes enters only at the design and saturated-implementation stage: it
underlies the stratified-SRSWOR identity of
Proposition~\ref{prop:srswor}, under which the augmentation term vanishes
exactly; the optimization of the stratum-specific sampling probabilities
(Propositions~\ref{prop:optq} and \ref{prop:simpledesign}); and saturated
stratum-wise estimation of the phase-1 conditional means. This delineation
clarifies that our main Poisson-sampling theory is not limited to discrete
phase-1 information, whereas the SRSWOR identity and the discrete
optimal-design and implementation results are.

Relaxing this discreteness at the design and implementation stage is a natural
next step. Estimation can accommodate continuous $S$, since the phase-1
conditional-mean functions can be fitted by data-adaptive regression rather
than by saturated stratum means, subject to the stated DML conditions.
Extending and operationalizing the allocation results when the finite
allocations $q(S=k)$ become an inclusion-probability function $q(\cdot)$ is a
topic for future work. Extensions in which $q(S)$ is estimated rather than
known, or phase-2 selections are dependent or adaptive, require separate theory
that accounts for the additional uncertainty or dependence induced by the
sampling mechanism. This conditional-projection perspective provides a
principled basis for considering phase-1 measurement, phase-2 sampling, and
analysis together in two-phase observational studies.

\clearpage
\section*{\texorpdfstring{Acknowledgments}{Acknowledgments}}

The authors used ChatGPT and Codex (OpenAI) and Claude Code (Anthropic) for language editing, consistency checks, and assistance with manuscript and code revision.
All AI-assisted output was reviewed and verified by the authors, who take full responsibility for the manuscript.

\section*{\texorpdfstring{Author's Statements}{Author's Statements}}

\noindent\textbf{Funding information:}
Kazuharu Harada was supported in part by JSPS KAKENHI 25K21165, and Masataka Taguri by JSPS KAKENHI 24K14862.

\noindent\textbf{Author contributions:}
KH: conceptualization, methodology, software, formal analysis, visualization, and original draft. MT: conceptualization, methodology, supervision, and review
and editing. Both authors reviewed the results, approved the final manuscript, consented to submission, and accept responsibility for its content.

\noindent\textbf{Conflict of interest:}
The authors declare no conflicts of interest.

\noindent\textbf{Informed consent:}
Not applicable to the present secondary analysis of publicly available data.

\noindent\textbf{Ethical approval:}
Not applicable to the present work. The empirical illustration used a publicly available dataset; no participants were recruited and no new data were collected.

\noindent\textbf{Data availability statement:}
The right-heart-catheterization data are publicly available from the Vanderbilt University Department of Biostatistics at \url{https://hbiostat.org/data/repo/rhc.csv}, with documentation at \url{https://hbiostat.org/data/repo/rhc.html}. Code for the estimators, simulations, and empirical illustration will be made publicly available upon publication; a review copy is available from the corresponding author on reasonable request.
\clearpage
\bibliographystyle{vancouver}
\bibliography{refs}

\clearpage

\begin{appendices}

\section*{\texorpdfstring{Appendices for ``Efficient estimation of weighted treatment effects under two-phase sampling''}{Appendices for ``Efficient estimation of weighted treatment effects under two-phase sampling''}}
\noindent Kazuharu Harada and Masataka Taguri

\section{Proofs and Derivations} \label{app:proofs}
\subsection{\texorpdfstring{Proof of Propositions~\ref{prop:tangent} and~\ref{prop:EIF}}{Proof of Propositions 1 and 2}}
The observed-data likelihood consists of four components determined by the combination of $(A,\delta)$. That is, for $O_{\mathrm{ODS}}=o$,
\begin{align*}
    (A,\delta) = (1,1):&~L_{11}(o) = q(1,v,y)\pi(x)p_X(x)p_1(y\mid x)\\
    (A,\delta) = (1,0):&~L_{10}(o) = \{1-q(1,v,y)\}p_S(A=1,v,y)\\
    (A,\delta) = (0,1):&~L_{01}(o) = q(0,v,y)\{1-\pi(x)\}p_X(x)p_0(y\mid x)\\
    (A,\delta) = (0,0):&~L_{00}(o) = \{1-q(0,v,y)\}p_S(A=0,v,y),
\end{align*}
where $p_X$ is the marginal density of $X$, and $p_a(y\mid x)$ is the conditional density of $Y^a$ given $X$. The outcome-independent case follows by taking $S=(A,V)$ and treating
$(W,Y)$ as phase-2 variables. Hence, the score with respect to a regular one-dimensional parametric submodel is given by
   \begin{align*}
     U_\varepsilon(O)
    =&~ A\delta \partial_\varepsilon\log L_{11}(Z;\varepsilon)
    + A(1-\delta) \partial_\varepsilon\log L_{10}(O;\varepsilon) \\
     &~~~~~ + (1-A)\delta \partial_\varepsilon\log L_{01}(Z;\varepsilon)
    + (1-A)(1-\delta) \partial_\varepsilon\log L_{00}(O;\varepsilon).
 \end{align*}
 Since $q$ is known, each component can be rewritten as
 \begin{align*}
    \partial_\varepsilon \log L_{11}(o) =&~ \partial_\varepsilon \log p(A=1,v,w,y;\varepsilon)\\
    \partial_\varepsilon \log L_{10}(o) =&~ \partial_\varepsilon \log p_S(A=1,v,y;\varepsilon)\\
    \partial_\varepsilon \log L_{01}(o) =&~ \partial_\varepsilon \log p(A=0,v,w,y;\varepsilon)\\
    \partial_\varepsilon \log L_{00}(o) =&~ \partial_\varepsilon \log p_S(A=0,v,y;\varepsilon).
\end{align*}
The joint density satisfies $p(A=a,v,w,y;\varepsilon)=
\pi_a(x;\varepsilon)p_X(x;\varepsilon)p_a(y\mid x;\varepsilon)$.

Here, $\partial_\varepsilon\log p_S(A,V,Y;\varepsilon)$ can be rewritten as
\begin{lemma}\label{lem:score}
    \begin{gather*}
        \partial_\varepsilon\log p_S(A,V,Y;\varepsilon) = E[\partial_\varepsilon\log p(A,V,W,Y;\varepsilon)\mid A,V,Y].
    \end{gather*}
\end{lemma}
\begin{proof}
    \begin{align*}
        \partial_\varepsilon\log p_S(A,V,Y;\varepsilon)
            =&~ \frac{\int\partial_\varepsilon p(A,V,w,Y;\varepsilon)dw}{p_S(A,V,Y;\varepsilon=0)} \\
            =&~ \int\frac{\partial_\varepsilon p(A,V,w,Y;\varepsilon)}{p(A,V,w,Y;\varepsilon=0)}\cdot\frac{p(A,V,w,Y;\varepsilon=0)}{p_S(A,V,Y;\varepsilon=0)}dw \\
            =&~ \int\partial_\varepsilon\log p(A,V,w,Y;\varepsilon)p(w\mid A,V,Y;\varepsilon=0)dw \\
            =&~ E[\partial_\varepsilon \log p(A,V,W,Y;\varepsilon)\mid A,V,Y].
    \end{align*}
\end{proof}

Then, the observed-data score can be expressed as
\begin{align*}
    U_\varepsilon(O)
    =&~ \partial_\varepsilon \log p_S(A,V,Y;\varepsilon) \\
    &~~~+ \delta \left\{ \partial_\varepsilon \log p(A,V,W,Y;\varepsilon) - E[\partial_\varepsilon \log p(A,V,W,Y;\varepsilon)\mid A,V,Y]\right\}.
 \end{align*}

Since $\partial_\varepsilon \log p(A,V,W,Y;\varepsilon)$ is the score of a one-dimensional parametric submodel for the fully observed data, it follows from Ref.~\cite{Hahn1998-lj} that
\begin{align*}
    \partial_\varepsilon \log p(a,v,w,y;\varepsilon)
        =&~ a u_1(y\mid x;\varepsilon) + (1-a)u_0(y\mid x;\varepsilon) \\
        &~~~ + \frac{\dot \pi(x;\varepsilon)}{\pi(x;\varepsilon)\{1-\pi(x;\varepsilon)\}}(a-\pi(x;\varepsilon)) + u_X(x;\varepsilon),
\end{align*}
where
\begin{align*}
    u_1(y\mid x;\varepsilon) =&~ \partial_\varepsilon \log p_1(y\mid x;\varepsilon) \\
    u_0(y\mid x;\varepsilon) =&~ \partial_\varepsilon \log p_0(y\mid x;\varepsilon) \\
    \dot \pi(x;\varepsilon) =&~ \partial_\varepsilon \pi(x;\varepsilon) \\
    u_X(x;\varepsilon) =&~ \partial_\varepsilon \log p_X(x;\varepsilon).
\end{align*}

Consequently, the closure of the linear span of these scores gives the tangent
space in Proposition~\ref{prop:tangent}.

Under the same one-dimensional parametric submodel, the WATE is given by
\begin{align*}
    \tau_h(\varepsilon)
        =&~ C_{h,\varepsilon}^{-1}E_\varepsilon[h\{\pi(X;\varepsilon)\}(Y^1 - Y^0)]\\
        =&~ \left\{\int h\{\pi(x;\varepsilon)\} ~ p_X(x;\varepsilon) dx\right\}^{-1}\iint h\{\pi(x;\varepsilon)\}y ~ p_1(y\mid x; \varepsilon)p_X(x;\varepsilon) dydx \\
        &~ - \left\{\int h\{\pi(x;\varepsilon)\} ~ p_X(x;\varepsilon) dx\right\}^{-1}\iint h\{\pi(x;\varepsilon)\}y ~ p_0(y\mid x; \varepsilon)p_X(x;\varepsilon) dydx.
\end{align*}
Differentiating this expression at $\varepsilon=0$ yields the pathwise derivative below.

 \begin{align*}
    \partial_\varepsilon C_{h,\varepsilon}
        =&~ \partial_\varepsilon\int h\{\pi(x;\varepsilon)\} ~ p_X(x;\varepsilon) dx \\
        =&~ \int h'(\pi(x))\dot \pi(x)p_X(x) dx + \int h\{\pi(x)\} ~ u_X(x)p_X(x) dx
\end{align*}

  \begin{align*}
    \partial_\varepsilon E_\varepsilon[h\{\pi(X;\varepsilon)\}Y^a]
    =&~ \partial_\varepsilon\iint h\{\pi(x;\varepsilon)\}y ~ p_a(y\mid x; \varepsilon)p_X(x;\varepsilon) dydx \\
    =&~ \iint h'(\pi(x))\dot \pi(x) y ~ p_a(y\mid x)p_X(x) dydx \\
        &~ + \iint h\{\pi(x)\}y ~ u_a(y\mid x)p_a(y\mid x)p_X(x) dydx \\
        &~ + \iint h\{\pi(x)\}y ~ u_X(x)p_a(y\mid x)p_X(x) dydx \\
    =&~ \int h'(\pi(x))\dot \pi(x) \mu_a(x) p_X(x) dx \\
        &~ + \iint h\{\pi(x)\}\{y - \mu_a(x)\}u_a(y\mid x)p_a(y\mid x)p_X(x) dydx \\
        &~ + \int h\{\pi(x)\}\mu_a(x)u_X(x)p_X(x) dx.
\end{align*}

Therefore, the pathwise differential of WATE at $\varepsilon=0$ is given by
  \begin{align*}
    \partial_\varepsilon\tau_h(\varepsilon)
        =&~ C_h^{-1}\left\{\partial_\varepsilon E_\varepsilon[h\{\pi(X;\varepsilon)\}(Y^1 - Y^0)] - \partial_\varepsilon C_{h,\varepsilon}\tau_h\right\} \\
    =&~ C_h^{-1}\int h'(\pi(x))\dot \pi(x) \{\tau(x) - \tau_h\} p_X(x) dx \\
        &~ + C_h^{-1}\iint h\{\pi(x)\}\{y - \mu_1(x)\}u_1(y\mid x)p_1(y\mid x)p_X(x) dydx \\
        &~ - C_h^{-1}\iint h\{\pi(x)\}\{y - \mu_0(x)\}u_0(y\mid x)p_0(y\mid x)p_X(x) dydx \\
        &~ + C_h^{-1}\int h\{\pi(x)\}\{\tau(x) - \tau_h\}u_X(x)p_X(x) dx.
\end{align*}

Let
  \begin{align*}
    \phi^{\mathrm{eff}}(O; \tau_h, \eta)
        =&~ E[\phi_{F}(Z; \tau_h, \eta)\mid S]
        + \delta\left\{
            \frac{\phi_{F}(Z; \tau_h, \eta)
            -E[\phi_{F}(Z; \tau_h, \eta)\mid S]}{q(S)}
        \right\},\\
    \phi_{F}(Z; \tau_h, \eta)
        =&~ C_h^{-1}h\{\pi_1(X)\}\left\{
            \frac{\mathbf{1}^A_{1}}{\pi_1(X)}(Y - \mu_{1}(X))
            - \frac{\mathbf{1}^A_{0}}{\pi_0(X)}(Y - \mu_{0}(X))\right.\\
        &\hspace{43mm}\left.{}+ \tau(X) - \tau_h\right\} \\
    &~~~ + C_h^{-1}h'(\pi_1(X))(\tau(X) - \tau_h)(A-\pi_1(X)).
\end{align*}
  The function $\phi^{\mathrm{eff}}(O;\tau_h,\eta)$ belongs to the tangent space in
Proposition~\ref{prop:tangent}. It remains to verify that
 \begin{gather*}
    \partial_\varepsilon\tau_h(\varepsilon) = E[\phi^{\mathrm{eff}}(O; \tau_h, \eta)u(O)].
\end{gather*}

\begin{align*}
    &~ E[\phi^{\mathrm{eff}}(O; \tau_h, \eta)u(O)] \\
        =&~ E\left[\left\{
            E[\phi_{F}\mid S]
        + \delta\left(
            \frac{\phi_{F}}{q(S)} - \frac{E[\phi_{F}\mid S]}{q(S)}
        \right)\right\}\left\{
            u_{S} + \delta \left(u_{F} - E[u_{F}\mid S]\right)
        \right\}
        \right] \\
        =&~ E[u_{S}E[\phi_{F}\mid S]]
            + \underbrace{E[E[\phi_{F}\mid S]\left(u_{F} - E[u_{F}\mid S]\right)]}_{=0} \\
            &~~ + \underbrace{E\left[\frac{\delta}{q(S)}\left(
            \phi_{F} - E[\phi_{F}\mid S]
        \right)u_{S}\right]}_{=0} + E\left[\frac{\delta}{q(S)}\left(
            \phi_{F} - E[\phi_{F}\mid S]
        \right)\left(u_{F} - E[u_{F}\mid S]\right)\right] \\
        =&~ E[\phi_{F} u_{F}] + E[(u_{S} - E[u_{F}\mid S])E[\phi_{F}\mid S]].
\end{align*}

Since $E[\phi_{F} u_{F}]$ is the inner product between the score function and the EIF under the fully observed data model, it coincides with the pathwise derivative of the WATE.
Moreover, recalling that $u_F,u_S$ are constructed as score functions of regular parametric submodels, Lemma \ref{lem:score} implies $u_S = E[u_F\mid S]$, and therefore the second term $E[(u_{S} - E[u_{F}\mid S])E[\phi_{F}\mid S]]$ is equal to zero.
Combining these results shows that $\partial_\varepsilon\tau_h(\varepsilon) = E[\phi^{\mathrm{eff}}(O; \tau_h, \eta)u(O)]$.

  Writing $\mu_\phi(S)=E\{\phi_F(Z;\tau_h,\eta)\mid S\}$ and conditioning first
on $Z$ gives the efficiency bound directly:
\begin{align*}
\text{Var}\{\phi^{\mathrm{eff}}(O;\tau_h,\eta)\}
&=\text{Var}\{\phi_F(Z;\tau_h,\eta)\}
  +E\left[\left\{\frac{1}{q(S)}-1\right\}
    \{\phi_F(Z;\tau_h,\eta)-\mu_\phi(S)\}^2\right]\\
&=\text{Var}\{\phi_F(Z;\tau_h,\eta)\}
  +E\left[\left\{\frac{1}{q(S)}-1\right\}
    \text{Var}\{\phi_F(Z;\tau_h,\eta)\mid S\}\right],
\end{align*}
where
\begin{align*}
    \text{Var}(\phi_{F}(Z;\tau_h,\eta))
        =&~ C_h^{-2}E\left[
            h^2\{\pi(X)\}\left\{\frac{\sigma_1^2(X)}{\pi(X)} + \frac{\sigma_0^2(X)}{1-\pi(X)}\right\}\right] \\
            &~~~ + C_h^{-2}E[h^2\{\pi(X)\}\{\tau(X)-\tau_h\}^2] \\
            &~~~+ C_h^{-2}E[h'(\pi(X))^2\pi(X)\{1-\pi(X)\}\{\tau(X)-\tau_h\}^2],
\end{align*}
and $\text{Var}(\phi_{F}(Z;\tau_h,\eta) \mid S)$ denotes the variance with respect to $(W,Y)$ for $S = (A,V)$, and to $W$ for $S = (A,V,Y)$.

\subsection{\texorpdfstring{Regularity conditions for the parametric estimators}{Regularity conditions for the parametric estimators}}
\label{app:parametric-regularity}

For each $e\in\{\mathrm{ipsw},\mathrm{enr}\}$, let
$M_e(\vartheta)=E\{\Psi_e(O;\vartheta)\}$ and assume that it has a
unique interior zero $\vartheta_e^*$, interpreted as the pseudo-true value
under working-model misspecification. We invoke the standard regularity
conditions used to establish consistency and smooth finite-dimensional
estimating-equation expansions; see Theorems~5.9 and 5.21 of
Ref.~\cite{Van_der_Vaart2000-bl} and
Refs.~\cite{Stefanski2002-wx,Tsiatis2006-sz}. In particular, under the
Poisson-sampling asymptotics, $O_1,\ldots,O_n$ are independent and
identically distributed, the selected sample root is measurable and
consistent for $\vartheta_e^*$, $\Psi_e(O;\vartheta)$ is continuously
differentiable in a neighborhood of $\vartheta_e^*$, and the empirical
Jacobian converges locally uniformly to its population counterpart. The
population Jacobian
\[
    J_e=E\left\{
      \frac{\partial\Psi_e(O;\vartheta_e^*)}{\partial\vartheta^\top}
    \right\}
\]
is nonsingular, and $E\|\Psi_e(O;\vartheta_e^*)\|^2<\infty$. We also
maintain the sampling and target moment conditions in
Section~\ref{sec:prelim1} and assume that the working treatment
probabilities are bounded away from zero and one near $\alpha^*$.

\subsection{Proof of Propositions \ref{prop:asympIPSW} and
\ref{prop:asympENR}}

For $e\in\{\mathrm{ipsw},\mathrm{enr}\}$, write
\[
    M_{n,e}(\vartheta)
    =\frac{1}{n}\sum_{i=1}^n\Psi_e(O_i;\vartheta).
\]
We first relate the pseudo-true target block to the scientific target. Let
$\Psi_{e,1}$ denote the first $r$ components of $\Psi_e$. The sampling
design gives $E\{\delta/q(S)\mid Z\}=1$ and
$E\{1-\delta/q(S)\mid S\}=0$. Therefore, for arbitrary
$(\tilde\theta_h,\tilde\eta,\tilde\gamma)$,
\begin{align*}
    E\{\Psi_{\mathrm{ipsw},1}
      (O;\tilde\theta_h,\tilde\eta)\}
      &=E\{m_F(Z;\tilde\theta_h,\tilde\eta)\},\\
    E\{\Psi_{\mathrm{enr},1}
      (O;\tilde\theta_h,\tilde\eta,\tilde\gamma)\}
      &=E\{m_F(Z;\tilde\theta_h,\tilde\eta)\}.
\end{align*}
Thus, correctness of the enrichment working models is not required for
target consistency. Under the model-specification conditions in
Sections~\ref{sec:ipsw} and \ref{sec:enrich}, the fully observed-data
equation satisfies
$E\{m_F(Z;\theta_h,\eta^*)\}=0$: the IPTW-type equation requires a
correctly specified propensity-score model, whereas the DR-type equation
is unbiased if that model is correctly specified or if $h$ is linear and
both outcome-regression models are correctly specified. For nonlinear
$h$, the stated guarantee requires a correctly specified propensity-score
model. The nuisance equations vanish at $\eta^*$ by definition, and the
primitive projection equations vanish at $\gamma^*$ for the enriched
estimator. Uniqueness of the stacked population root therefore gives
\[
    \theta_{h,e}^*=\theta_h.
\]
Together with the consistency result cited above, this establishes target
consistency under the stated model-specification conditions. Under
misspecification outside these conditions, the estimator instead converges
to the pseudo-true target block $\theta_{h,e}^*$.

Under the standard smoothness conditions cited above, a first-order
expansion of the sample estimating equation around $\vartheta_e^*$ gives
\[
    0=M_{n,e}(\vartheta_e^*)
      +J_e(\hat\vartheta_e-\vartheta_e^*)+o_p(n^{-1/2}).
\]
Rearranging yields the asymptotic linear representation
\[
    \sqrt{n}(\hat\vartheta_e-\vartheta_e^*)
    =-J_e^{-1}\frac{1}{\sqrt{n}}\sum_{i=1}^n
       \Psi_e(O_i;\vartheta_e^*)+o_p(1).
\]
Thus, with
\[
    K_e=E\{\Psi_e(O;\vartheta_e^*)
          \Psi_e(O;\vartheta_e^*)^\top\},
\]
the multivariate central limit theorem and Slutsky's theorem yield
\[
    \sqrt{n}(\hat\vartheta_e-\vartheta_e^*)
    \overset{d}{\longrightarrow}
    N\!\left(0,J_e^{-1}K_eJ_e^{-\top}\right),
    \qquad J_e^{-\top}:=(J_e^{-1})^\top.
\]

Taking the first $r$ coordinates shows that the target-block influence
function is given by the first $r$ components of
$-J_e^{-1}\Psi_e(O;\vartheta_e^*)$, with covariance equal to the
leading $r\times r$ block of $J_e^{-1}K_eJ_e^{-\top}$.

For IPSW, inversion of the block upper-triangular Jacobian gives the
target-block influence function stated in Section~\ref{sec:ipsw}.
For the enriched estimator,
\[
    E\left[
      \left\{1-\frac{\delta}{q(S)}\right\}
      \partial_{\gamma^\top}
      g(S;\theta_{h,\mathrm{enr}}^*,\gamma^*)
    \right]=0.
\]
The primitive projection equations do not depend on $\theta_h$, and the
nuisance-score rows have zero target and gamma blocks. These block zeros
imply that estimation of $\gamma$ does not contribute to the first-order
target block. Block inversion therefore gives the influence function
stated in Section~\ref{sec:enrich}.

Finally, suppose that $m_F=m_F^{\mathrm{dr}}$, all full-data nuisance
models are correctly specified, and
\[
    g(S;\theta_h,\gamma^*)
    =E\{m_F^{\mathrm{dr}}(Z;\theta_h,\eta^*)\mid S\}.
\]
Then $J_{12}=J_{13}=J_{14}=0$,
$-J_{11}^{-1}m_F^{\mathrm{dr}}=\phi_F$, and
$-J_{11}^{-1}g=E(\phi_F\mid S)$. The enriched target-block influence
function therefore reduces to $\phi^{\mathrm{eff}}$, with variance
$\mathcal V_h$.
This completes the proof.

 \subsection{\texorpdfstring{Proof of Proposition~\ref{prop:optq}}{Proof of the optimal-allocation proposition}}
We first derive the solution to a general optimization problem, of which the
optimization problem in this paper is a special case.
\begin{align*}
    \min_{\bs\theta = (\theta_1,\ldots,\theta_K)} &~ \sum_{k=1}^K{w_k}\cdot\frac{c_k}{\theta_k} \\
    \text{subject to} &~ \sum_{k=1}^K w_k\theta_k \le d, 0 < \theta_k \le 1~~\text{for all}~k. \\
\end{align*}
 Here, $w_k>0$ are weights summing to one, $c_k>0$, and $d\in(0,1)$.

The objective is strictly convex and decreases in every $\theta_k$, so the
budget constraint binds and the optimizer is unique. Let $\lambda>0$ be the
multiplier for the budget constraint and let $\nu_k\geq0$ be the multiplier
for $\theta_k\leq1$. The KKT conditions are
\begin{align*}
 -\frac{w_kc_k}{\theta_k^2}+\lambda w_k+\nu_k&=0,\\
 \nu_k(\theta_k-1)&=0.
\end{align*}
If the upper bound is inactive, then
$\theta_k=\sqrt{c_k/\lambda}$. If that unconstrained value is at least one,
the upper bound is active. Therefore, for some $t=\lambda^{-1/2}>0$,
\begin{equation}\label{eq:proof-capped-allocation}
 \theta_k^*=\min(1,t\sqrt{c_k}).
\end{equation}
For $d\in(0,1)$, the map
$t\mapsto\sum_kw_k\min(1,t\sqrt{c_k})$ is continuous and strictly increasing
at its crossing of $d$, so a unique $t$ satisfies the binding budget
constraint.

Taking $w_k=p_k$, $c_k=\sigma_k^2$, $d=\bar q$, and $t=c$ in
\eqref{eq:proof-capped-allocation} gives the general solution
\begin{equation}\label{eq:optq-capped}
 q_k^*=\min(1,c\sigma_k),\qquad k=1,\ldots,K,
\end{equation}
where $c$ is calibrated so that
$\sum_kp_k\min(1,c\sigma_k)=\bar q$. If no truncation is required, the budget
constraint gives
\[
 t=\frac{d}{\sum_jw_j\sqrt{c_j}},
 \qquad
 \theta_k^*=\frac{d\sqrt{c_k}}{\sum_jw_j\sqrt{c_j}}.
\]
This candidate respects every upper bound if and only if
\[
 d\leq\frac{\sum_jw_j\sqrt{c_j}}{\max_j\sqrt{c_j}}.
\]
Finally, substitution into the objective gives
\[
 \sum_{k=1}^Kw_k\frac{c_k}{\theta_k^*}
 =d^{-1}\left(\sum_{k=1}^Kw_k\sqrt{c_k}\right)^2.
\]
These expressions prove Proposition~\ref{prop:optq};
\eqref{eq:optq-capped} gives the capped extension referred to in the main text.

\subsection{Proof of Proposition \ref{prop:DML}}
In this section, following the linear-score DML framework of
Ref.~\cite{Chernozhukov2018-cc}, we prove
Proposition~\ref{prop:DML} by invoking the full-data results of
Ref.~\cite{Yiming2025-do}.

\subsubsection{Key Assumptions}
We first restate Assumptions 3.1 and 3.2 of
    Ref.~\cite{Chernozhukov2018-cc} to fit the problem setting considered in this
study. These can also be regarded as observed-data versions of Assumptions
B.1.3 and B.1.4 of Ref.~\cite{Yiming2025-do}.

As assumed in Section~\ref{sec:DML}, the projection models are
correctly specified: their target functions are $g_{N0}$ and $g_{D0}$, and
the fold-specific projection estimators converge to these functions.
Throughout, $K$ is fixed, each phase-2 record inherits the fold of its
phase-1 unit, and each $\hat\eta^{-k}$ is constructed exclusively from
observations in $I_{-k}$.
When stratified folds are used and $S$ has fixed finite support, the same
cross-fitting argument is applied within each positive-probability cell.

Let $C_1\ge C_0 > 0$ denote finite constants, and let $\{\delta_n\}_{n\ge 1}$ and $\{\Delta_n\}_{n\ge 1}$ be sequences of positive constants converging to zero with $\delta_n \ge n^{-1/2}$.

    \begin{assumption}[Adapted from Assumption 3.1 of Ref.~\cite{Chernozhukov2018-cc};
Assumption B.1.3 of Ref.~\cite{Yiming2025-do}] \label{assum:Cherno31}
    For $n\ge 3$, the following assumptions hold.
    \begin{enumerate}
        \item[(a)] The parameter $\tau_h$ is identified by $E[m(O;\tau_h,\eta_0)] = 0$.
        \item[(b)] The score function is linear in the form of $m(O;\tau_h,\eta) = N(O;\eta)-\tau_hD(O;\eta)$.
        \item[(c)] The mapping $\eta \mapsto E[m(O;\tau_h,\eta)]$ is twice continuously Gateaux-differentiable.
        \item[(d)] The estimating function obeys (near) Neyman orthogonality at $(\tau_h,\eta_0)$ with respect to the nuisance realization set $\mathcal{T}_n\subset\mathcal{T}$, where $\mathcal{T}$ is a convex set, for
            \begin{gather*}
            \lambda_n:= \sup_{\tilde\eta\in\mathcal{T}_n} \left|
                \left.\partial_\eta E\{m(O;\tau_h,\eta)\}\right|_{\eta=\eta_0}
                [\tilde\eta-\eta_0]
            \right|\le\delta_nn^{-1/2}.
        \end{gather*}
        \item [(e)] The identification condition holds in the sense that $C_0 \le E\{D(O;\eta_0)\}=C_h \le C_1$.
    \end{enumerate}
\end{assumption}

    \begin{assumption}[Adapted from Assumption 3.2 of Ref.~\cite{Chernozhukov2018-cc};
Assumption B.1.4 of Ref.~\cite{Yiming2025-do}] \label{assum:Cherno32}
    For $n\ge 3$ and $\nu\ge4$, the following conditions hold.
    \begin{enumerate}
        \item[(a)] Given a random subset $I$ of size $n/K$, the estimator of the nuisance parameter $\hat\eta:= \hat\eta((O_{i})_{i\in I^c})$ belongs to the realization set $\mathcal{T}_n$ with probability at least $1 - \Delta_n$, where $\mathcal{T}_n$ contains $\eta_0$ and is constrained by the next conditions.
        \item[(b)] The moment conditions hold:
            \begin{align*}
            m_n:=&~ \sup_{\tilde\eta\in\mathcal{T}_n}\left(
                E[|m(O;\tau_h,\tilde\eta)|^\nu]
            \right)^{1/\nu} \le C_1, \\
            m'_n:=&~ \sup_{\tilde\eta\in\mathcal{T}_n}\left(
                E[|D(O;\tilde\eta)|^\nu]
            \right)^{1/\nu} \le C_1.
        \end{align*}
        \item[(c)] The following conditions on the statistical rates $r_n,r'_n$, and $\lambda'_n$ hold:
            \begin{align*}
            r_n:=&~ \sup_{\tilde\eta\in\mathcal{T}_n} \left|
                E[D(O;\tilde\eta)] - E[D(O;\eta_0)]
            \right| \le \delta_n, \\
            r'_n:=&~ \sup_{\tilde\eta\in\mathcal{T}_n} \left(E\left|
                m(O;\tau_h,\tilde\eta) - m(O;\tau_h,\eta_0)
            \right|^2\right)^{1/2} \le \delta_n,\\
            \lambda'_n:=&~ \sup_{r\in(0,1),\tilde\eta\in\mathcal{T}_n} \left|
                \partial_r^2E[m(O;\tau_h,\eta_0 + r(\tilde\eta - \eta_0))]
            \right| \le \delta_n/\sqrt{n}.
        \end{align*}
        \item[(d)] The estimating function is non-degenerate:
          $E\{m(O;\tau_h,\eta_0)^2\}\ge C_0$.
    \end{enumerate}
\end{assumption}

\subsubsection{Proof of Proposition \ref{prop:DML}}
The problem considered in this study extends the full-data setting of
Ref.~\cite{Yiming2025-do} to two-phase sampling. We show how the
observed-data conditions follow from the corresponding full-data conditions,
together with conditions on the separately fitted projection functions.

First, for any
$\tilde\eta=(\tilde\eta_F,\tilde g_N,\tilde g_D)$, we have
\begin{gather*}
 m(O;\tau_h,\tilde\eta)
 =\frac{\delta}{q(S)}m_F^{\mathrm{dr}}(Z;\tau_h,\tilde\eta_F)
  +\left\{1-\frac{\delta}{q(S)}\right\}
  \{\tilde g_N(S)-\tau_h\tilde g_D(S)\}.
\end{gather*}
Because $E(\delta\mid Z)=q(S)$,
\begin{align*}
 E\{m(O;\tau_h,\tilde\eta)\}
   &=E\{m_F^{\mathrm{dr}}(Z;\tau_h,\tilde\eta_F)\},\\
 E\{D(O;\tilde\eta)\}
   &=E\{D_F(Z;\tilde\eta_F)\}.
\end{align*}
These identities hold for arbitrary $\tilde g_N$ and $\tilde g_D$.
Consequently, the population moment and its first and second derivatives in
the full-data nuisance directions reduce to the corresponding full-data
quantities, whereas the derivatives in the projection directions vanish.
Thus, the near-orthogonality and second-order bias conditions reduce to those
of Ref.~\cite{Yiming2025-do}. In addition, the Jacobian rate in
Assumption~\ref{assum:Cherno32}(c) satisfies
\begin{gather*}
 r_n
 =\sup_{\tilde\eta\in\mathcal T_n}
 \left|E\{D_F(Z;\tilde\eta_F)\}
       -E\{D_F(Z;\eta_{F0})\}\right|,
\end{gather*}
and therefore also depends only on the full-data nuisance functions.

Next, we describe the realization-set, moment, and $L_2$ conditions for the
separately fitted projection functions. Let $\mathcal T_n^F$ denote the full-data
realization set considered in Ref.~\cite{Yiming2025-do}, and let
$\mathcal T_n^g$ denote a realization set for the two projection learners.
We take $\mathcal T_n\subset\mathcal T_n^F\times\mathcal T_n^g$ and impose
Assumption~\ref{assum:Cherno32}(a) on the complete fold-specific estimator
$\hat\eta^{-k}=(\hat\eta_F^{-k},\hat g_N^{-k},\hat g_D^{-k})$.
In particular, membership of $\hat\eta_F^{-k}$ in $\mathcal T_n^F$ alone
does not imply membership of the learned projection functions in
$\mathcal T_n^g$.

Write $\tilde g_m(S)=\tilde g_N(S)-\tau_h\tilde g_D(S)$. Under the
positivity condition $q(S)\ge q_{\min}>0$, Minkowski's inequality gives
\begin{align*}
 \|m(O;\tau_h,\tilde\eta)\|_\nu
 &\le q_{\min}^{-1}
   \left\{\|m_F^{\mathrm{dr}}(Z;\tau_h,\tilde\eta_F)\|_\nu
      +\|\tilde g_m(S)\|_\nu\right\},\\
 \|D(O;\tilde\eta)\|_\nu
 &\le q_{\min}^{-1}
   \left\{\|D_F(Z;\tilde\eta_F)\|_\nu
      +\|\tilde g_D(S)\|_\nu\right\}.
\end{align*}
We constrain the component realization sets so that the norms on the
right-hand sides are uniformly bounded. The displayed inequalities then give
the uniform moment bounds in Assumption~\ref{assum:Cherno32}(b).

To verify the $r_n'$ condition in
Assumption~\ref{assum:Cherno32}(c), define
\begin{align*}
 \Delta_F(\tilde\eta_F)
   &=m_F^{\mathrm{dr}}(Z;\tau_h,\tilde\eta_F)
     -m_F^{\mathrm{dr}}(Z;\tau_h,\eta_{F0}),\\
 \Delta_g(\tilde g_N,\tilde g_D)
   &=\{\tilde g_N(S)-g_{N0}(S)\}
     -\tau_h\{\tilde g_D(S)-g_{D0}(S)\}.
\end{align*}
Then
\begin{gather*}
 m(O;\tau_h,\tilde\eta)-m(O;\tau_h,\eta_0)
 =\frac{\delta}{q(S)}\Delta_F(\tilde\eta_F)
  +\left\{1-\frac{\delta}{q(S)}\right\}
    \Delta_g(\tilde g_N,\tilde g_D),
\end{gather*}
and therefore
\begin{gather*}
 \|m(O;\tau_h,\tilde\eta)-m(O;\tau_h,\eta_0)\|_2
 \le q_{\min}^{-1}
 \left\{\|\Delta_F(\tilde\eta_F)\|_2
       +\|\Delta_g(\tilde g_N,\tilde g_D)\|_2\right\}.
\end{gather*}
Thus, the $r_n'$ condition follows from uniform bounds of the required order
on the full-data score error and the combined projection error. No
product-rate condition involving the projection error is required because
the projection term disappears from the population moment.

Finally, we verify Assumption~\ref{assum:Cherno32}(d). At $\eta_0$, let
\begin{gather*}
 g_{m0}(S)=g_{N0}(S)-\tau_hg_{D0}(S)
 =E\{m_F^{\mathrm{dr}}(Z;\tau_h,\eta_{F0})\mid S\}.
\end{gather*}
The observed-data score can then be written as
\begin{gather*}
 m(O;\tau_h,\eta_0)
 =g_{m0}(S)+\frac{\delta}{q(S)}
   \{m_F^{\mathrm{dr}}(Z;\tau_h,\eta_{F0})-g_{m0}(S)\}.
\end{gather*}
Hence,
\begin{align*}
 E\{m(O;\tau_h,\eta_0)\mid S\}&=g_{m0}(S),\\
 \operatorname{Var}\{m(O;\tau_h,\eta_0)\mid S\}
 &=\frac{1}{q(S)}
   \operatorname{Var}\{m_F^{\mathrm{dr}}(Z;\tau_h,\eta_{F0})\mid S\}.
\end{align*}
By the law of total variance and $q(S)\le1$,
\begin{align*}
 \operatorname{Var}\{m(O;\tau_h,\eta_0)\}
 &=\operatorname{Var}\{g_{m0}(S)\}
   +E\left[\frac{1}{q(S)}
     \operatorname{Var}\{m_F^{\mathrm{dr}}(Z;\tau_h,\eta_{F0})\mid S\}\right]\\
 &\ge \operatorname{Var}\{m_F^{\mathrm{dr}}(Z;\tau_h,\eta_{F0})\}.
\end{align*}
Thus, the full-data non-degeneracy condition of
Ref.~\cite{Yiming2025-do} implies
Assumption~\ref{assum:Cherno32}(d).

The score Jacobian is
$-E\{D(O;\eta_0)\}=-C_h$.
Theorems~3.1 and~3.2 of Ref.~\cite{Chernozhukov2018-cc} now give
\begin{gather*}
 \sqrt n(\hat\tau_h^{\mathrm{dml}}-\tau_h)
 =C_h^{-1}\frac{1}{\sqrt n}\sum_{i=1}^n
   m(O_i;\tau_h,\eta_0)+o_p(1),
\end{gather*}
as well as consistency of the variance estimator. The condition $\nu\ge4$,
together with $\delta_n\ge n^{-1/2}$, meets the moment-rate requirement for
the latter result.

Finally,
\begin{gather*}
 m_F^{\mathrm{dr}}(Z;\tau_h,\eta_{F0})
   =C_h\phi_F(Z;\tau_h,\eta_{F0}),
 \qquad
 g_{m0}(S)=C_hE\{\phi_F(Z;\tau_h,\eta_{F0})\mid S\}.
\end{gather*}
It follows that
\begin{align*}
 C_h^{-1}m(O;\tau_h,\eta_0)
 &=\frac{\delta}{q(S)}\phi_F(Z;\tau_h,\eta_{F0})
   +\left\{1-\frac{\delta}{q(S)}\right\}
     E\{\phi_F(Z;\tau_h,\eta_{F0})\mid S\}\\
 &=\phi^{\mathrm{eff}}(O;\tau_h,\eta_{F0}).
\end{align*}
Therefore, the limiting variance is $\mathcal V_h$, which completes the
proof.

 \subsection{Optimal sampling probability for the IPSW estimator}
The objective function for the IPSW estimator is
\begin{gather*}
    \sum_{k=1}^K p_kq_k^{-1}\left\{
        \text{Var}(\tilde\phi_{F}(Z;\tau_h,\eta)\mid S=k)
        + E[\tilde\phi_{F}(Z;\tau_h,\eta)\mid S=k]^2
    \right\},
\end{gather*}
where $\tilde\phi_{F}$ is an influence function for the WATE under fully observed data.
The optimal solution therefore depends not only on the conditional variance of the influence function within each stratum, but also on the variability of its conditional mean.

\begin{proposition}
      Let $c_k=\sigma_k^2+\xi_k^2>0$ and assume
    $\bar q\le\{\sum_{j=1}^Kp_j\sqrt{c_j}\}/\max_j\sqrt{c_j}$.
Under this no-cap condition, the IPSW objective above, subject to the same
budget and box constraints as in \eqref{eq:optimA}, is minimized by the
allocation shown below.
     \begin{gather*}
        q_k = \frac{{\bar q}\sqrt{\sigma^2_k + \xi^2_k}}{\sum_{j=1}^K p_j\sqrt{\sigma^2_j + \xi^2_j}},
    \end{gather*}
    where $\sigma_k^2 = \text{Var}(\tilde\phi_F(Z;\tau_h,\eta)\mid S=k)$ and $\xi_k = E[\tilde\phi_F(Z;\tau_h,\eta)\mid S=k]$.
The corresponding minimum objective value is
    \begin{gather*}
        {\bar q}^{-1}\left\{\sum_{k=1}^Kp_k\sqrt{\sigma^2_k + \xi^2_k}\right\}^2.
    \end{gather*}
\end{proposition}
The conditional-mean term may reflect effect heterogeneity across
sampling strata.

\section{Closed-form Solution to the Primary Estimating Equations}\label{sec:closed}

    For compactness, define
\begin{align*}
    \hat r_{ai}
        &=\frac{\mathbf{1}^A_{ai}}{\hat\pi_{ai}}\hat h_i,
    &\hat s_{ai}&=\hat r_{ai}Y_i,\\
    \hat d_i
        &=\hat h_i+h'(\hat\pi_{1i})(A_i-\hat\pi_{1i}),\\
    \hat n_i
        &=\hat h_i\frac{\mathbf{1}^A_{1i}}{\hat\pi_{1i}}
            (Y_i-\hat\mu_{1i})
          -\hat h_i\frac{\mathbf{1}^A_{0i}}{\hat\pi_{0i}}
            (Y_i-\hat\mu_{0i})
          +\hat d_i\hat\tau_i.
\end{align*}
Solving the primary IPSW estimating equations gives
\begin{align*}
    \hat\mu_h^{a,\text{siw}}
        &=\left\{\sum_{i=1}^n\frac{\delta_i}{q_i}\hat r_{ai}\right\}^{-1}
          \sum_{i=1}^n\frac{\delta_i}{q_i}\hat s_{ai},\\
    \hat\tau_h^{\text{siw}}
        &=\hat\mu_h^{1,\text{siw}}-\hat\mu_h^{0,\text{siw}},\\
    \hat\tau_h^{\text{sdr}}
        &=\left\{\sum_{i=1}^n\frac{\delta_i}{q_i}\hat d_i\right\}^{-1}
          \sum_{i=1}^n\frac{\delta_i}{q_i}\hat n_i.
\end{align*}

    Solving the corresponding enriched estimating equations gives
\begin{align*}
    \hat\mu_h^{a,\text{eiw}}
        &=\left[
          \sum_{i=1}^n\left\{
          \frac{\delta_i}{q_i}\hat r_{ai}
          +\left(1-\frac{\delta_i}{q_i}\right)
             \hat E(\hat r_{ai}\mid S_i)
          \right\}\right]^{-1}\\
        &\quad\times
          \sum_{i=1}^n\left\{
          \frac{\delta_i}{q_i}\hat s_{ai}
          +\left(1-\frac{\delta_i}{q_i}\right)
             \hat E(\hat s_{ai}\mid S_i)
          \right\},\\
    \hat\tau_h^{\text{eiw}}
        &=\hat\mu_h^{1,\text{eiw}}-\hat\mu_h^{0,\text{eiw}},\\
    \hat\tau_h^{\text{edr}}
        &=\left[
          \sum_{i=1}^n\left\{
          \frac{\delta_i}{q_i}\hat d_i
          +\left(1-\frac{\delta_i}{q_i}\right)
             \hat E(\hat d_i\mid S_i)
          \right\}\right]^{-1}\\
        &\quad\times
          \sum_{i=1}^n\left\{
          \frac{\delta_i}{q_i}\hat n_i
          +\left(1-\frac{\delta_i}{q_i}\right)
             \hat E(\hat n_i\mid S_i)
          \right\}.
\end{align*}

  \clearpage
\section{Detailed Simulation Settings}\label{sec:detailsim}
\subsection{Data-Generating Process for DGP~1 (Comparative Study)}\label{app:dgp01}

In each replication, $n$ independent samples
    $Z^\star=(A,V_1,\ldots,V_8,W_c,Y^1,Y^0)$ are generated as follows, where
$A \in \{0,1\}$ is the treatment indicator,
    $X=(V_1,\ldots,V_8,W_c)$ comprises the baseline covariates, and
$(Y^1,Y^0)$ are the potential outcomes.

\paragraph{Covariates.}
The binary covariates $V_1,\ldots,V_8\distiid\mathrm{Bernoulli}(0.5)$ are drawn independently.
A continuous covariate is generated as
    \[
    W_c = -1 + 0.5\times(V_1 + V_3 + V_4 + V_7) + \varepsilon_c,\qquad \varepsilon_c \sim N(0,\,0.5^2).
\]

\paragraph{Treatment.}
    \begin{align*}
    A\mid X &\sim \mathrm{Bernoulli}\{\pi(X)\},\\
    \pi(X)
      &=\mathrm{expit}\{-2.10
        +0.5(V_1+V_2+V_4+V_6)+W_c\}.
\end{align*}

\paragraph{Potential outcomes.}\leavevmode
    \begin{align*}
    Y^1\mid X&\sim\mathrm{Bernoulli}\{p_1(X)\},\\
    p_1(X)&=\mathrm{expit}\{-0.59
      +0.5(V_1+V_2+V_3+V_5)+1.5W_c\},\\
    Y^0\mid X&\sim\mathrm{Bernoulli}\{p_0(X)\},\\
    p_0(X)&=\mathrm{expit}\{\mathrm{logit}(0.4)+W_c\}.
\end{align*}
The observed outcome is $Y = AY^1 + (1-A)Y^0$.

Table~\ref{tab:V1toV8} summarizes the dependence structure between $V_1$--$V_8$ and the variables $(A,W_c,Y^1,Y^0)$.
Since none of $V_1$--$V_8$ is directly associated with $Y^0$, any $V_j$ related to $Y^1$ or $W_c$ contributes to increased effect heterogeneity.

\begin{table}[htbp]
    \centering
    \caption{Presence or absence of associations between the binary covariate and other variables. A check mark (\checkmark) indicates that the covariate is included in the generating mechanism of the corresponding column variable.}
    \vspace{2mm}
    \begin{tabular}{ccccc}\hline
        \rule{0pt}{2.5ex}
              & $Y^1$ & $Y^0$ & $A$ & $W_c$ \\\hline
        $V_1$ & \checkmark & & \checkmark & \checkmark \\
        $V_2$ & \checkmark & & \checkmark & \\
        $V_3$ & \checkmark & & & \checkmark \\
        $V_4$ & & & \checkmark & \checkmark \\
        $V_5$ & \checkmark & & & \\
        $V_6$ & & & \checkmark & \\
        $V_7$ & & & & \checkmark \\
        $V_8$ & & & & \\\hline
    \end{tabular}
    \label{tab:V1toV8}
\end{table}

\paragraph{Phase-1 sample and sampling design.}
In phase 1, under ODS, the stratum variable is $S=(A,V_{\text{obs}},Y)$; under non-ODS, $S=(A,V_{\text{obs}})$.
Among $V_1$--$V_8$, only a single variable is designated as $V_{\text{obs}}$ in each scenario, yielding eight strata under ODS and four under non-ODS.
   For each scenario, this corresponds to the theoretical partition
$V=V_{\text{obs}}$ and $W=(V_j:j\ne\text{obs},W_c)$.

Phase-2 sampling is conducted via two procedures.
The first is a stratified SRSWOR procedure applied to the phase-1 sample, which draws an equal number of observations from each stratum defined by $S$; the sampling probability $q(S)$ is computed after phase-2 selection.
The second uses Poisson sampling with stratum-specific probabilities precomputed from a large reference dataset of $n = 10^8$ observations, calibrated to yield equal expected counts per stratum.
    The two procedures have comparable limiting first-order
inclusion probabilities.
 The large reference dataset was also used to compute the true values of the target parameters:
$\tau_{\mathrm{ATE}}=0.165$, $\tau_{\mathrm{ATT}}=0.212$, $\tau_{\mathrm{ATU}}=0.146$, $\tau_{\mathrm{ATO}}=0.198$.

\subsection{\texorpdfstring{Data-Generating Process for DGP~2 (ML-Based EDR Validation)}{Data-Generating Process for DGP 2 (ML-Based EDR Validation)}}\label{app:dgp02}

\paragraph{Latent variables.}
Let $Z_1,\ldots,Z_{10}\distiid\mathrm{Bernoulli}(0.5)$ be ten independent
binary latent variables, and let $V\sim\mathrm{Bernoulli}(0.5)$ be a binary
phase-1 covariate, all mutually independent.

\paragraph{Treatment.}
Define $\alpha=(1,-1,0.5,-0.5,0.25,-0.25,0,0,0,0)^{\top}$. The treatment
indicator is drawn as
\begin{gather*}
  A\mid Z,V\;\sim\;\mathrm{Bernoulli}\!\left(\pi(Z,V)\right),\qquad
  \pi(Z,V)=\mathrm{expit}\!\bigl(\alpha^{\top}(Z-0.5\cdot\mathbf{1})+(V-0.5)\bigr).
\end{gather*}
Only $Z_1,\ldots,Z_6$ affect treatment selection; $Z_7,\ldots,Z_{10}$ do not.

\paragraph{Outcome.}
Set $\mu_Y=\sum_{j=1}^{10}(Z_j-0.5)+(V-0.5)$, with
$\mathrm{Var}(\mu_Y)=10\times\tfrac{1}{4}+\tfrac{1}{4}=\tfrac{11}{4}$.
With latent-outcome $R^2$ fixed at $r^2=0.75$, define
$\sigma^2=\mathrm{Var}(\mu_Y)\,(1-r^2)/r^2=\tfrac{11}{12}$.
Generate $\widetilde{Y}\sim N(\mu_Y,\sigma^2)$ and set
\begin{gather*}
  Y\mid Z,V\;\sim\;\mathrm{Bernoulli}\!\left(\mathrm{expit}(\widetilde{Y}/\sigma)\right).
\end{gather*}
Because $\mu_Y$ does not depend on $A$, the treatment has no causal effect:
$\tau_{\mathrm{ATE}}=\tau_{\mathrm{ATT}}=\tau_{\mathrm{ATU}}=\tau_{\mathrm{ATO}}=0$.

\paragraph{Phase-1 sample and sampling design.}
Every unit in the phase-1 sample of size $n$ has $S=(A,V,Y)$ observed (ODS).
Because all three variables are binary, the stratum space has $2^3=8$ cells.
Phase-2 sampling is Poisson with stratum-specific probabilities
$q(S=s)$ calibrated from a reference sample of size $10^6$ so that an equal
expected number of observations is allocated to each stratum, subject to the
overall sampling fraction $\bar{q}=0.25$.

\paragraph{Observed phase-2 covariates.}
Two scenarios are considered (see main text, Section~\ref{sec:sim}).
In the \emph{linear} scenario $W_j=Z_j$ for $j=1,\ldots,10$.
In the \emph{nonlinear} scenario:
\begin{align*}
  W_k &= Z_{2k-1}+2\,Z_{2k}\;\in\{0,1,2,3\}, \quad k=1,\ldots,5
    & &\text{(pair encodings, invertible),}\\
  W_6 &= Z_1Z_3,\quad W_7=Z_2Z_4,\quad W_8=Z_5Z_7
    & &\text{(second-order products),}\\
  W_9 &= Z_1Z_5Z_9,\quad W_{10}=Z_2Z_6Z_{10}
    & &\text{(third-order products).}
\end{align*}
The pair encodings together allow complete reconstruction of $Z_1,\ldots,Z_{10}$
by a sufficiently flexible model; the interaction terms enrich the feature space
without adding identifiable information beyond the pair encodings.

\clearpage
\section{Additional Simulation Results}\label{sec:addsim}

\subsection{Asymptotic Behavior (DGP~1)}
Figure~\ref{fig:coverage} shows coverage of the 95\% EDR intervals based on the sample variance of the EIF. Coverage was below nominal in small samples, especially under ODS with Poisson sampling, and improved with sample size; ATO had the most favorable coverage. The phase-2 fraction and point-estimator bias correction had little effect on coverage.

\begin{figure}[htbp]
\centering
\includegraphics[width=0.9\linewidth]{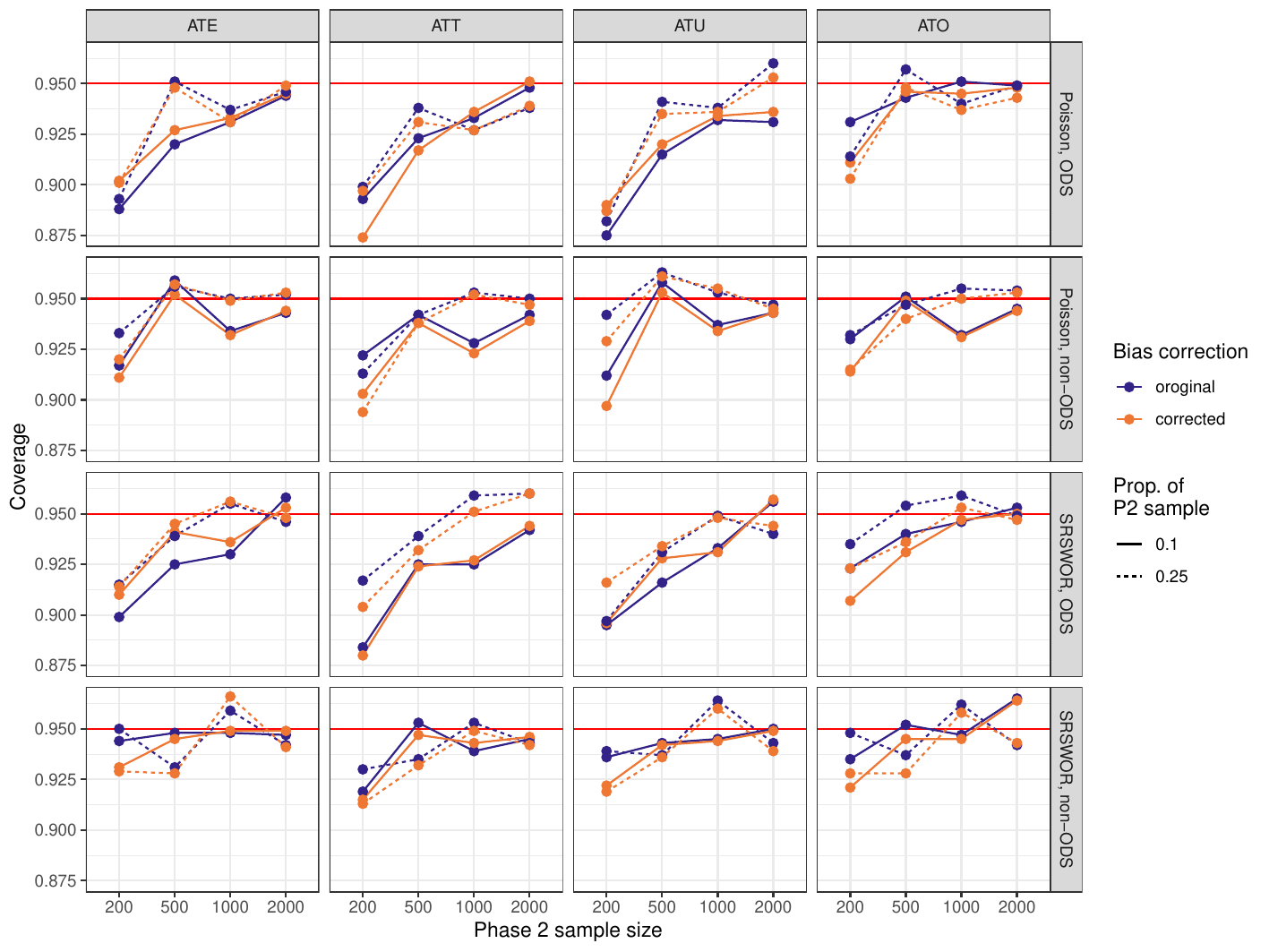}
\caption{Coverage proportion of 95\% confidence intervals of the EDR estimator based on the corresponding EIF when $V_{\text{obs}}=V_1$.}
\label{fig:coverage}
\end{figure}

  Empirical SE and RMSE decreased with sample size
(Figures~\ref{fig:asymp_empSE_pois} and \ref{fig:asymp_RMSE_pois}). For ATO,
 the IPTW- and DR-type estimators performed similarly.

\begin{figure}[htbp]
\centering
\includegraphics[width=0.9\linewidth]{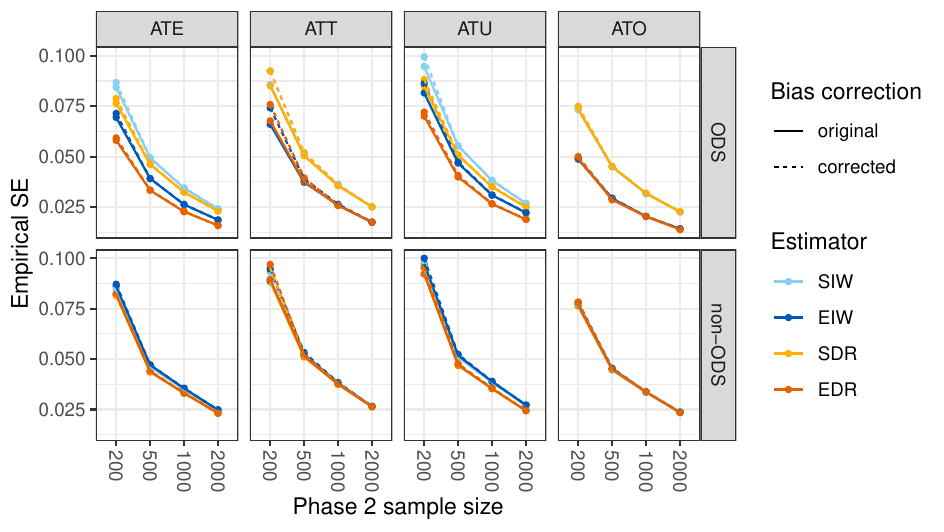}
\caption{Empirical standard error of each estimator by sample size under Poisson sampling with $V_{\text{obs}}=V_1$ and $n=10m$.}
\label{fig:asymp_empSE_pois}
\end{figure}

\begin{figure}[htbp]
\centering
\includegraphics[width=0.9\linewidth]{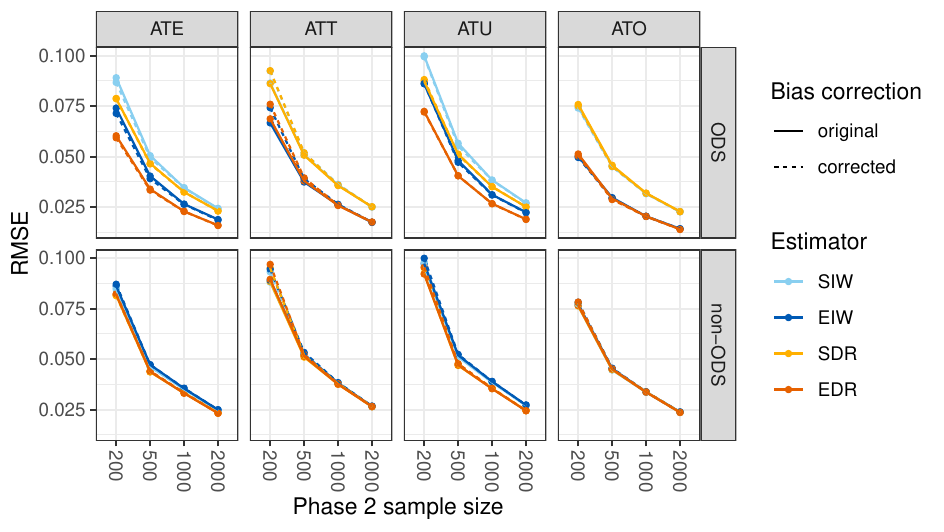}
\caption{Root mean squared error (RMSE) of each estimator by sample size under Poisson sampling with $V_{\text{obs}}=V_1$ and $n=10m$.}
\label{fig:asymp_RMSE_pois}
\end{figure}

\clearpage
\subsection{\texorpdfstring{Estimation Performance (DGP~1)}{Estimation Performance (DGP 1)}}
 Tables~\ref{tab:tab:ATT_n10000_m1000_poisson_ODS}--\ref{tab:tab:ATO_n10000_m1000_poisson_ODS}
report the remaining estimands under ODS, and
Tables~\ref{tab:tab:ATE_n10000_m1000_poisson_non-ODS} and
\ref{tab:tab:ATE_n10000_m1000_srswor_ODS} report the ATE under non-ODS and
stratified SRSWOR.

  Among the uncorrected estimators, bias was largest for ATU and smallest
for ATO, and the jackknife reduced bias. Enrichment made little difference
 under non-ODS, while the uncorrected estimators agreed exactly under stratified
SRSWOR as implied by Proposition~\ref{prop:srswor}.

\begin{table}[tbp]
\centering
\caption{\label{tab:tab:ATT_n10000_m1000_poisson_ODS}Performance of each estimator when targeting the ATT under $n = 10{,}000$, $m = 1{,}000$, outcome-dependent Poisson sampling. For each cell of $V_{\text{obs}}$ and estimator, the three numbers shown (from top to bottom) represent the relative RMSE, relative empirical SE, and bias ($\times 10^2$).}
\centering
\begin{tabular}[t]{ccccccccc}
\toprule
\multicolumn{1}{c}{ } & \multicolumn{4}{c}{Original} & \multicolumn{4}{c}{Bias-corrected} \\
\cmidrule(l{3pt}r{3pt}){2-5} \cmidrule(l{3pt}r{3pt}){6-9}
$V_{\text{obs}}$ & SIW & EIW & SDR & EDR & SIW & EIW & SDR & EDR\\
\midrule
$V_1$ & 5.06 & 3.67 & 5.02 & 3.63 & 5.11 & 3.71 & 5.06 & 3.65\\
 & 5.05 & 3.67 & 5.01 & 3.62 & 5.11 & 3.71 & 5.06 & 3.65\\
 & -0.25 & -0.18 & -0.26 & -0.21 & 0.01 & 0.02 & 0.06 & 0.06\vspace{2mm}\\
$V_2$ & 5.10 & 3.84 & 5.15 & 3.87 & 5.14 & 3.89 & 5.20 & 3.92\\
 & 5.09 & 3.83 & 5.14 & 3.86 & 5.14 & 3.89 & 5.20 & 3.92\\
 & -0.25 & -0.18 & -0.30 & -0.23 & 0.12 & 0.14 & 0.12 & 0.15\vspace{2mm}\\
$V_3$ & 5.12 & 3.72 & 5.09 & 3.67 & 5.12 & 3.74 & 5.08 & 3.66\\
 & 5.09 & 3.71 & 5.04 & 3.64 & 5.12 & 3.74 & 5.08 & 3.66\\
 & -0.41 & -0.25 & -0.48 & -0.33 & -0.09 & 0.03 & -0.09 & 0.02\vspace{2mm}\\
$V_4$ & 4.99 & 3.70 & 4.96 & 3.69 & 4.99 & 3.74 & 4.95 & 3.71\\
 & 4.94 & 3.68 & 4.91 & 3.66 & 4.98 & 3.74 & 4.94 & 3.71\\
 & -0.50 & -0.33 & -0.54 & -0.38 & -0.23 & -0.11 & -0.22 & -0.10\vspace{2mm}\\
$V_5$ & 5.25 & 3.95 & 5.28 & 4.00 & 5.29 & 4.00 & 5.32 & 4.05\\
 & 5.22 & 3.92 & 5.24 & 3.96 & 5.29 & 4.00 & 5.32 & 4.05\\
 & -0.42 & -0.35 & -0.47 & -0.40 & -0.08 & -0.05 & -0.06 & -0.03\vspace{2mm}\\
$V_6$ & 5.06 & 3.91 & 5.07 & 3.92 & 5.07 & 3.91 & 5.09 & 3.93\\
 & 5.04 & 3.88 & 5.05 & 3.88 & 5.08 & 3.92 & 5.11 & 3.94\\
 & -0.38 & -0.40 & -0.39 & -0.42 & -0.03 & -0.10 & 0.01 & -0.05\vspace{2mm}\\
$V_7$ & 5.13 & 3.98 & 5.10 & 3.90 & 5.17 & 4.03 & 5.14 & 3.93\\
 & 5.11 & 3.96 & 5.07 & 3.87 & 5.17 & 4.03 & 5.14 & 3.93\\
 & -0.39 & -0.34 & -0.41 & -0.34 & -0.08 & -0.07 & -0.04 & -0.02\vspace{2mm}\\
$V_8$ & 5.16 & 4.19 & 5.11 & 4.12 & 5.21 & 4.24 & 5.17 & 4.18\\
 & 5.15 & 4.17 & 5.09 & 4.10 & 5.21 & 4.24 & 5.17 & 4.18\\
 & -0.26 & -0.23 & -0.31 & -0.27 & 0.09 & 0.08 & 0.10 & 0.09\vspace{2mm}\\
\bottomrule
\end{tabular}
\end{table}

\begin{table}[tbp]
\centering
\caption{\label{tab:tab:ATC_n10000_m1000_poisson_ODS}Performance of each estimator when targeting the ATU under $n = 10{,}000$, $m = 1{,}000$, outcome-dependent Poisson sampling. For each cell of $V_{\text{obs}}$ and estimator, the three numbers shown (from top to bottom) represent the relative RMSE, relative empirical SE, and bias ($\times 10^2$).}
\centering
\begin{tabular}[t]{ccccccccc}
\toprule
\multicolumn{1}{c}{ } & \multicolumn{4}{c}{Original} & \multicolumn{4}{c}{Bias-corrected} \\
\cmidrule(l{3pt}r{3pt}){2-5} \cmidrule(l{3pt}r{3pt}){6-9}
$V_{\text{obs}}$ & SIW & EIW & SDR & EDR & SIW & EIW & SDR & EDR\\
\midrule
$V_1$ & 5.67 & 4.60 & 5.18 & 3.95 & 5.65 & 4.58 & 5.19 & 3.94\\
 & 5.63 & 4.56 & 5.16 & 3.92 & 5.65 & 4.58 & 5.19 & 3.94\\
 & -0.44 & -0.42 & -0.31 & -0.33 & 0.13 & 0.09 & 0.07 & -0.00\vspace{2mm}\\
$V_2$ & 6.07 & 5.10 & 5.51 & 4.41 & 6.02 & 5.05 & 5.49 & 4.40\\
 & 6.01 & 5.04 & 5.47 & 4.38 & 6.02 & 5.05 & 5.49 & 4.40\\
 & -0.60 & -0.54 & -0.43 & -0.37 & 0.02 & 0.02 & 0.06 & 0.06\vspace{2mm}\\
$V_3$ & 5.64 & 4.66 & 5.17 & 3.95 & 5.55 & 4.61 & 5.11 & 3.92\\
 & 5.53 & 4.59 & 5.09 & 3.89 & 5.55 & 4.61 & 5.11 & 3.92\\
 & -0.75 & -0.59 & -0.63 & -0.46 & -0.15 & -0.05 & -0.15 & -0.04\vspace{2mm}\\
$V_4$ & 5.57 & 4.55 & 5.18 & 4.01 & 5.53 & 4.51 & 5.17 & 4.00\\
 & 5.51 & 4.50 & 5.13 & 3.98 & 5.53 & 4.52 & 5.17 & 4.01\\
 & -0.58 & -0.45 & -0.48 & -0.36 & 0.01 & 0.08 & -0.07 & -0.00\vspace{2mm}\\
$V_5$ & 5.96 & 5.04 & 5.41 & 4.35 & 5.86 & 4.94 & 5.37 & 4.31\\
 & 5.83 & 4.92 & 5.31 & 4.26 & 5.85 & 4.94 & 5.36 & 4.31\\
 & -0.84 & -0.74 & -0.70 & -0.60 & -0.20 & -0.16 & -0.18 & -0.14\vspace{2mm}\\
$V_6$ & 6.01 & 4.93 & 5.56 & 4.40 & 6.01 & 4.91 & 5.58 & 4.42\\
 & 5.99 & 4.91 & 5.53 & 4.37 & 6.01 & 4.91 & 5.59 & 4.42\\
 & -0.39 & -0.34 & -0.38 & -0.32 & 0.23 & 0.22 & 0.12 & 0.11\vspace{2mm}\\
$V_7$ & 5.59 & 4.67 & 5.23 & 4.06 & 5.52 & 4.59 & 5.23 & 4.05\\
 & 5.51 & 4.58 & 5.17 & 4.00 & 5.53 & 4.59 & 5.23 & 4.05\\
 & -0.66 & -0.62 & -0.54 & -0.47 & -0.05 & -0.07 & -0.06 & -0.05\vspace{2mm}\\
$V_8$ & 5.92 & 5.03 & 5.59 & 4.63 & 5.87 & 4.99 & 5.61 & 4.67\\
 & 5.86 & 4.97 & 5.55 & 4.60 & 5.87 & 4.99 & 5.60 & 4.67\\
 & -0.57 & -0.51 & -0.39 & -0.31 & 0.08 & 0.07 & 0.13 & 0.13\vspace{2mm}\\
\bottomrule
\end{tabular}
\end{table}

\begin{table}[tbp]
\centering
\caption{\label{tab:tab:ATO_n10000_m1000_poisson_ODS}Performance of each estimator when targeting the ATO under $n = 10{,}000$, $m = 1{,}000$, outcome-dependent Poisson sampling. For each cell of $V_{\text{obs}}$ and estimator, the three numbers shown (from top to bottom) represent the relative RMSE, relative empirical SE, and bias ($\times 10^2$).}
\centering
\begin{tabular}[t]{ccccccccc}
\toprule
\multicolumn{1}{c}{ } & \multicolumn{4}{c}{Original} & \multicolumn{4}{c}{Bias-corrected} \\
\cmidrule(l{3pt}r{3pt}){2-5} \cmidrule(l{3pt}r{3pt}){6-9}
$V_{\text{obs}}$ & SIW & EIW & SDR & EDR & SIW & EIW & SDR & EDR\\
\midrule
$V_1$ & 4.66 & 3.00 & 4.69 & 3.00 & 4.65 & 3.00 & 4.68 & 2.99\\
 & 4.64 & 2.98 & 4.67 & 2.98 & 4.65 & 3.00 & 4.68 & 2.99\\
 & -0.27 & -0.21 & -0.30 & -0.25 & 0.06 & 0.06 & 0.07 & 0.05\vspace{2mm}\\
$V_2$ & 4.80 & 3.15 & 4.81 & 3.17 & 4.78 & 3.15 & 4.79 & 3.16\\
 & 4.78 & 3.14 & 4.79 & 3.15 & 4.78 & 3.15 & 4.79 & 3.15\\
 & -0.28 & -0.21 & -0.31 & -0.22 & 0.06 & 0.07 & 0.06 & 0.07\vspace{2mm}\\
$V_3$ & 4.69 & 2.98 & 4.70 & 2.95 & 4.65 & 2.96 & 4.66 & 2.92\\
 & 4.64 & 2.94 & 4.65 & 2.91 & 4.65 & 2.96 & 4.66 & 2.92\\
 & -0.48 & -0.31 & -0.49 & -0.31 & -0.16 & -0.04 & -0.14 & -0.04\vspace{2mm}\\
$V_4$ & 4.60 & 3.04 & 4.62 & 3.01 & 4.54 & 3.01 & 4.55 & 2.97\\
 & 4.55 & 3.01 & 4.56 & 2.97 & 4.54 & 3.02 & 4.55 & 2.97\\
 & -0.48 & -0.31 & -0.51 & -0.33 & -0.15 & -0.04 & -0.13 & -0.03\vspace{2mm}\\
$V_5$ & 4.82 & 3.21 & 4.85 & 3.21 & 4.79 & 3.18 & 4.81 & 3.18\\
 & 4.78 & 3.17 & 4.80 & 3.17 & 4.79 & 3.18 & 4.80 & 3.18\\
 & -0.42 & -0.33 & -0.44 & -0.33 & -0.09 & -0.06 & -0.07 & -0.04\vspace{2mm}\\
$V_6$ & 4.76 & 3.24 & 4.78 & 3.21 & 4.77 & 3.24 & 4.78 & 3.21\\
 & 4.76 & 3.23 & 4.78 & 3.20 & 4.78 & 3.25 & 4.79 & 3.22\\
 & -0.27 & -0.26 & -0.28 & -0.26 & 0.08 & 0.03 & 0.09 & 0.05\vspace{2mm}\\
$V_7$ & 4.71 & 3.19 & 4.73 & 3.13 & 4.68 & 3.17 & 4.69 & 3.10\\
 & 4.68 & 3.15 & 4.69 & 3.09 & 4.68 & 3.17 & 4.69 & 3.10\\
 & -0.39 & -0.34 & -0.41 & -0.32 & -0.06 & -0.06 & -0.05 & -0.04\vspace{2mm}\\
$V_8$ & 4.63 & 3.25 & 4.65 & 3.23 & 4.62 & 3.24 & 4.64 & 3.22\\
 & 4.62 & 3.24 & 4.63 & 3.22 & 4.62 & 3.24 & 4.64 & 3.21\\
 & -0.22 & -0.17 & -0.25 & -0.17 & 0.12 & 0.12 & 0.12 & 0.12\vspace{2mm}\\
\bottomrule
\end{tabular}
\end{table}

\begin{table}[tbp]
\centering
\caption{\label{tab:tab:ATE_n10000_m1000_poisson_non-ODS}Performance of each estimator when targeting the ATE under $n = 10{,}000$, $m = 1{,}000$, non-outcome-dependent Poisson sampling. For each cell of $V_{\text{obs}}$ and estimator, the three numbers shown (from top to bottom) represent the relative RMSE, relative empirical SE, and bias ($\times 10^2$).}
\centering
\begin{tabular}[t]{ccccccccc}
\toprule
\multicolumn{1}{c}{ } & \multicolumn{4}{c}{Original} & \multicolumn{4}{c}{Bias-corrected} \\
\cmidrule(l{3pt}r{3pt}){2-5} \cmidrule(l{3pt}r{3pt}){6-9}
$V_{\text{obs}}$ & SIW & EIW & SDR & EDR & SIW & EIW & SDR & EDR\\
\midrule
$V_1$ & 5.55 & 5.62 & 5.26 & 5.24 & 5.57 & 5.64 & 5.30 & 5.28\\
 & 5.55 & 5.62 & 5.26 & 5.24 & 5.57 & 5.64 & 5.29 & 5.28\\
 & -0.05 & -0.07 & 0.07 & 0.06 & 0.05 & 0.03 & 0.05 & 0.05\vspace{2mm}\\
$V_2$ & 5.11 & 5.12 & 5.01 & 5.02 & 5.11 & 5.12 & 5.03 & 5.04\\
 & 5.12 & 5.12 & 5.01 & 5.02 & 5.11 & 5.12 & 5.03 & 5.04\\
 & 0.03 & 0.02 & 0.12 & 0.12 & 0.11 & 0.11 & 0.13 & 0.13\vspace{2mm}\\
$V_3$ & 5.07 & 5.16 & 4.74 & 4.73 & 5.08 & 5.18 & 4.76 & 4.75\\
 & 5.07 & 5.16 & 4.74 & 4.73 & 5.08 & 5.18 & 4.76 & 4.75\\
 & -0.06 & -0.08 & 0.01 & 0.01 & 0.01 & -0.01 & 0.00 & 0.00\vspace{2mm}\\
$V_4$ & 5.39 & 5.41 & 5.07 & 5.07 & 5.38 & 5.40 & 5.07 & 5.08\\
 & 5.36 & 5.38 & 5.05 & 5.06 & 5.37 & 5.38 & 5.06 & 5.06\\
 & -0.35 & -0.36 & -0.25 & -0.26 & -0.27 & -0.27 & -0.25 & -0.26\vspace{2mm}\\
$V_5$ & 5.39 & 5.41 & 5.17 & 5.16 & 5.41 & 5.43 & 5.19 & 5.18\\
 & 5.39 & 5.41 & 5.17 & 5.17 & 5.41 & 5.43 & 5.19 & 5.18\\
 & -0.02 & -0.03 & -0.01 & -0.00 & 0.03 & 0.03 & -0.01 & -0.01\vspace{2mm}\\
$V_6$ & 5.59 & 5.62 & 5.29 & 5.30 & 5.61 & 5.63 & 5.31 & 5.32\\
 & 5.59 & 5.62 & 5.29 & 5.30 & 5.61 & 5.63 & 5.31 & 5.32\\
 & 0.03 & 0.03 & 0.14 & 0.15 & 0.10 & 0.10 & 0.15 & 0.16\vspace{2mm}\\
$V_7$ & 5.41 & 5.44 & 5.16 & 5.16 & 5.41 & 5.43 & 5.17 & 5.17\\
 & 5.40 & 5.43 & 5.16 & 5.16 & 5.41 & 5.43 & 5.17 & 5.16\\
 & -0.15 & -0.18 & -0.10 & -0.10 & -0.09 & -0.10 & -0.10 & -0.10\vspace{2mm}\\
$V_8$ & 5.98 & 5.99 & 5.65 & 5.66 & 5.98 & 5.98 & 5.66 & 5.67\\
 & 5.97 & 5.98 & 5.65 & 5.66 & 5.97 & 5.98 & 5.66 & 5.67\\
 & -0.26 & -0.27 & -0.20 & -0.20 & -0.20 & -0.21 & -0.20 & -0.20\vspace{2mm}\\
\bottomrule
\end{tabular}
\end{table}

\begin{table}[tbp]
\centering
\caption{\label{tab:tab:ATE_n10000_m1000_srswor_ODS}Performance of each estimator when targeting the ATE under $n = 10{,}000$, $m = 1{,}000$, outcome-dependent Stratified SRSWOR. For each cell of $V_{\text{obs}}$ and estimator, the three numbers shown (from top to bottom) represent the relative RMSE, relative empirical SE, and bias ($\times 10^2$).}
\centering
\begin{tabular}[t]{ccccccccc}
\toprule
\multicolumn{1}{c}{ } & \multicolumn{4}{c}{Original} & \multicolumn{4}{c}{Bias-corrected} \\
\cmidrule(l{3pt}r{3pt}){2-5} \cmidrule(l{3pt}r{3pt}){6-9}
$V_{\text{obs}}$ & SIW & EIW & SDR & EDR & SIW & EIW & SDR & EDR\\
\midrule
$V_1$ & 3.89 & 3.89 & 3.46 & 3.46 & 3.81 & 3.80 & 3.45 & 3.43\\
 & 3.81 & 3.81 & 3.41 & 3.41 & 3.82 & 3.80 & 3.45 & 3.43\\
 & -0.51 & -0.51 & -0.41 & -0.41 & -0.01 & -0.05 & -0.05 & -0.10\vspace{2mm}\\
$V_2$ & 4.19 & 4.19 & 3.77 & 3.77 & 4.14 & 4.12 & 3.79 & 3.77\\
 & 4.13 & 4.13 & 3.74 & 3.74 & 4.13 & 4.12 & 3.78 & 3.76\\
 & -0.43 & -0.43 & -0.31 & -0.31 & 0.16 & 0.11 & 0.15 & 0.09\vspace{2mm}\\
$V_3$ & 4.18 & 4.18 & 3.74 & 3.74 & 4.10 & 4.10 & 3.72 & 3.70\\
 & 4.11 & 4.11 & 3.68 & 3.68 & 4.10 & 4.10 & 3.72 & 3.70\\
 & -0.51 & -0.51 & -0.44 & -0.44 & 0.06 & 0.01 & 0.01 & -0.05\vspace{2mm}\\
$V_4$ & 3.78 & 3.78 & 3.41 & 3.41 & 3.73 & 3.72 & 3.41 & 3.40\\
 & 3.72 & 3.72 & 3.37 & 3.37 & 3.73 & 3.72 & 3.41 & 3.40\\
 & -0.43 & -0.43 & -0.35 & -0.35 & 0.09 & 0.04 & 0.03 & -0.03\vspace{2mm}\\
$V_5$ & 4.20 & 4.20 & 3.82 & 3.82 & 4.19 & 4.18 & 3.85 & 3.83\\
 & 4.17 & 4.17 & 3.80 & 3.80 & 4.18 & 4.17 & 3.84 & 3.83\\
 & -0.35 & -0.35 & -0.30 & -0.30 & 0.24 & 0.19 & 0.19 & 0.12\vspace{2mm}\\
$V_6$ & 4.23 & 4.23 & 3.82 & 3.82 & 4.13 & 4.13 & 3.79 & 3.79\\
 & 4.13 & 4.13 & 3.75 & 3.75 & 4.13 & 4.13 & 3.80 & 3.79\\
 & -0.60 & -0.60 & -0.47 & -0.47 & -0.02 & -0.07 & 0.01 & -0.06\vspace{2mm}\\
$V_7$ & 4.12 & 4.12 & 3.75 & 3.75 & 4.05 & 4.04 & 3.73 & 3.72\\
 & 4.05 & 4.05 & 3.70 & 3.70 & 4.05 & 4.04 & 3.73 & 3.72\\
 & -0.49 & -0.49 & -0.42 & -0.42 & 0.08 & 0.03 & 0.03 & -0.03\vspace{2mm}\\
$V_8$ & 4.38 & 4.38 & 4.00 & 4.00 & 4.25 & 4.24 & 3.94 & 3.93\\
 & 4.25 & 4.25 & 3.89 & 3.89 & 4.25 & 4.24 & 3.94 & 3.93\\
 & -0.69 & -0.69 & -0.59 & -0.59 & -0.08 & -0.13 & -0.08 & -0.14\vspace{2mm}\\
\bottomrule
\end{tabular}
\end{table}

\clearpage
\subsection{Efficiency gains from enrichment (DGP~1)}
\subsubsection{Under Poisson sampling}
Across the additional sample sizes in Figures~\ref{fig:gain_cor_200_pois}--\ref{fig:gain_cor_2000_pois}, enrichment gains remained negligible under non-ODS and increased with sample size under ODS. The jackknife correction had little visible effect on empirical \%Gain (Figure~\ref{fig:gain_org_1000_pois}).

\begin{figure}[htbp]
\centering
\includegraphics[width=0.95\linewidth]{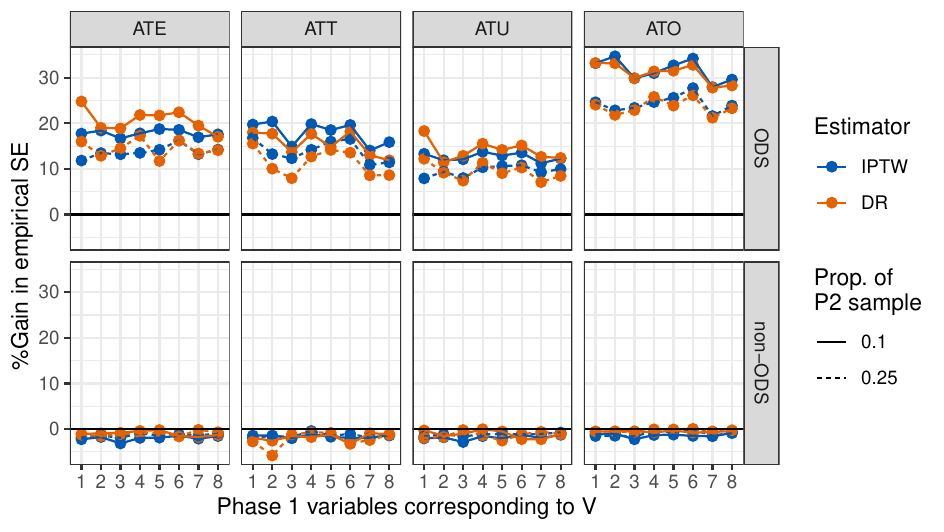}
\caption{Efficiency gains from enrichment by the choice of $V_{\text{obs}}$ under Poisson sampling. The estimator is bias-corrected, and the phase-2 sample size is 200.}
\label{fig:gain_cor_200_pois}
\end{figure}

\begin{figure}[htbp]
\centering
\includegraphics[width=0.95\linewidth]{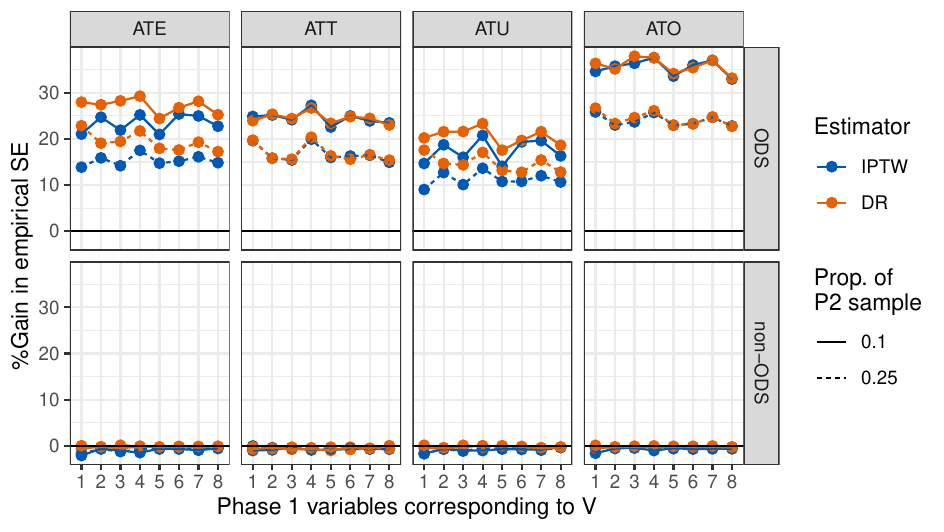}
\caption{Efficiency gains from enrichment by the choice of $V_{\text{obs}}$ under Poisson sampling. The estimator is bias-corrected, and the phase-2 sample size is 500.}
\label{fig:gain_cor_500_pois}
\end{figure}

\begin{figure}[htbp]
\centering
\includegraphics[width=0.95\linewidth]{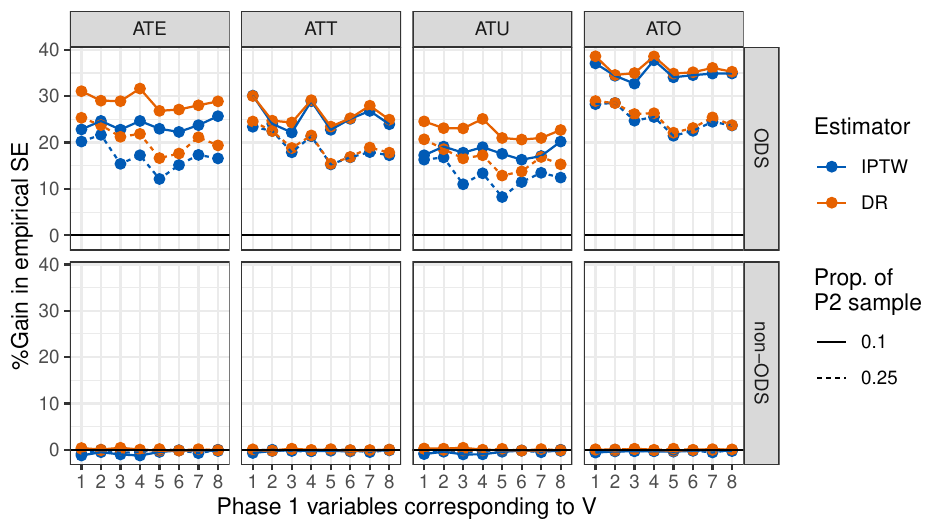}
\caption{Efficiency gains from enrichment by the choice of $V_{\text{obs}}$ under Poisson sampling. The estimator is bias-corrected, and the phase-2 sample size is 2{,}000.}
\label{fig:gain_cor_2000_pois}
\end{figure}

\begin{figure}[htbp]
\centering
\includegraphics[width=0.95\linewidth]{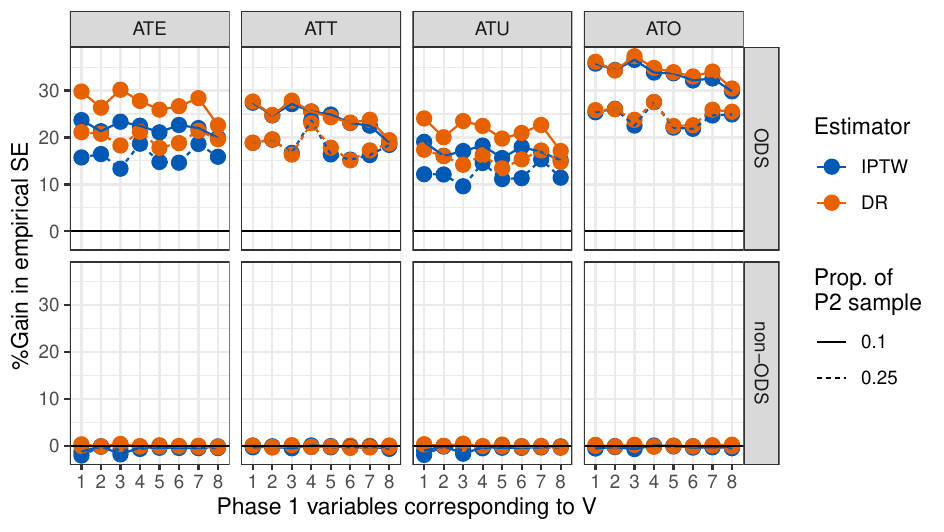}
\caption{Efficiency gains from enrichment by the choice of $V_{\text{obs}}$ under Poisson sampling. The estimator is not bias-corrected, and the phase-2 sample size is 1{,}000.}
\label{fig:gain_org_1000_pois}
\end{figure}

\clearpage
\subsubsection{Under stratified SRSWOR}
  For $m=1{,}000$, the uncorrected estimators had exactly zero gain, as implied by
Proposition~\ref{prop:srswor}. The small positive ODS gains after correction
arose because the unstratified delete-group partition disrupts the fixed
within-stratum sampling fractions in the jackknife replicates
(Figures~\ref{fig:gain_cor_1000_srs} and \ref{fig:gain_org_1000_srs}).

\begin{figure}[htbp]
\centering
\includegraphics[width=0.95\linewidth]{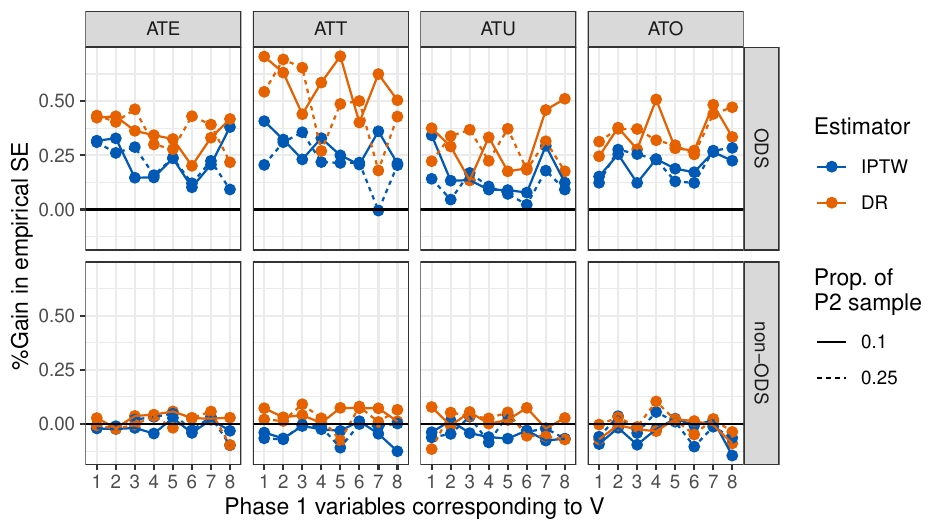}
\caption{Efficiency gains from enrichment by the choice of $V_{\text{obs}}$ under stratified SRSWOR. The estimator is bias-corrected, and the phase-2 sample size is 1{,}000.}
\label{fig:gain_cor_1000_srs}
\end{figure}

\begin{figure}[htbp]
\centering
\includegraphics[width=0.95\linewidth]{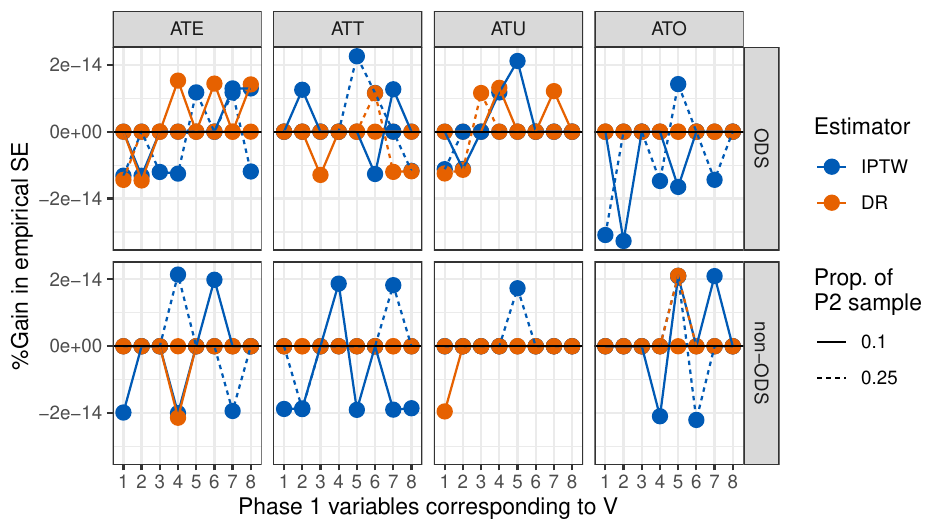}
\caption{Efficiency gains from enrichment by the choice of $V_{\text{obs}}$ under stratified SRSWOR. The estimator is not bias-corrected, and the phase-2 sample size is 1{,}000.}
\label{fig:gain_org_1000_srs}
\end{figure}

\clearpage
\subsection{\texorpdfstring{Full Results for the ML-Based EDR Validation Study (DGP~2)}{Full Results for the ML-Based EDR Validation Study (DGP 2)}}\label{sec:addsim_dgp2}

Tables~\ref{tab:dgp2_lin_full} and \ref{tab:dgp2_nonlin_full} present the
complete Monte Carlo results for the ML-based EDR estimators under DGP~2,
covering both the linear and nonlinear scenarios.
For each combination of phase-1 sample size $n$ and estimation configuration, we report the signed Monte Carlo bias (Bias), empirical standard
error (EmpSE), and empirical coverage of the 95\% Wald confidence interval
(Cov) for all four estimands (ATE, ATT, ATU, ATO).

\begin{table}[tbp]
  \centering
  \caption{Full simulation results for the ML-based EDR estimators under DGP~2,
    \textbf{linear} scenario ($W_j = Z_j$).
    All true parameter values are zero.
    Bias: signed Monte Carlo bias;
    EmpSE: empirical standard error;
    Cov: empirical coverage (\%) of the 95\% Wald confidence interval
    (1{,}000 replications per configuration).}
  \label{tab:dgp2_lin_full}
  \small
  \setlength{\tabcolsep}{5pt}
  \begin{tabular}{lr rrr c rrr}\toprule
    & & \multicolumn{3}{c}{ML-EDR (w/o CF)}
    & \phantom{x}
    & \multicolumn{3}{c}{DML-EDR ($K=5$)} \\
    \cmidrule{3-5}\cmidrule{7-9}
    Estimand & $n$ & Bias & EmpSE & Cov(\%)
                  && Bias & EmpSE & Cov(\%) \\\midrule
    \multirow[t]{4}{*}{ATE}
      & 1{,}000 & $ 0.003$ & $0.045$ & $84.9$ && $ 0.001$ & $0.055$ & $94.0$ \\
      & 2{,}000 & $ 0.001$ & $0.030$ & $91.5$ && $-0.001$ & $0.033$ & $96.0$ \\
      & 4{,}000 & $ 0.002$ & $0.022$ & $93.3$ && $ 0.001$ & $0.022$ & $96.1$ \\
      & 8{,}000 & $ 0.001$ & $0.015$ & $94.2$ && $ 0.000$ & $0.016$ & $95.2$ \\\addlinespace
    \multirow[t]{4}{*}{ATT}
      & 1{,}000 & $ 0.004$ & $0.053$ & $84.8$ && $ 0.001$ & $0.070$ & $93.5$ \\
      & 2{,}000 & $ 0.001$ & $0.036$ & $91.0$ && $-0.002$ & $0.042$ & $95.2$ \\
      & 4{,}000 & $ 0.003$ & $0.026$ & $92.7$ && $ 0.002$ & $0.027$ & $95.5$ \\
      & 8{,}000 & $ 0.001$ & $0.018$ & $93.4$ && $ 0.000$ & $0.019$ & $94.5$ \\\addlinespace
    \multirow[t]{4}{*}{ATU}
      & 1{,}000 & $ 0.003$ & $0.054$ & $84.0$ && $ 0.001$ & $0.070$ & $93.6$ \\
      & 2{,}000 & $ 0.001$ & $0.035$ & $91.3$ && $ 0.000$ & $0.040$ & $95.1$ \\
      & 4{,}000 & $ 0.000$ & $0.025$ & $91.5$ && $-0.001$ & $0.027$ & $94.4$ \\
      & 8{,}000 & $ 0.000$ & $0.018$ & $94.3$ && $ 0.000$ & $0.018$ & $95.5$ \\\addlinespace
    \multirow[t]{4}{*}{ATO}
      & 1{,}000 & $-0.007$ & $0.046$ & $92.8$ && $ 0.002$ & $0.048$ & $93.9$ \\
      & 2{,}000 & $-0.004$ & $0.029$ & $95.4$ && $ 0.000$ & $0.030$ & $94.7$ \\
      & 4{,}000 & $-0.001$ & $0.021$ & $95.7$ && $ 0.001$ & $0.021$ & $95.5$ \\
      & 8{,}000 & $-0.001$ & $0.015$ & $95.3$ && $ 0.000$ & $0.015$ & $95.4$ \\
    \bottomrule
  \end{tabular}
\end{table}

\begin{table}[tbp]
  \centering
  \caption{Full simulation results for the ML-based EDR estimators under DGP~2,
    \textbf{nonlinear} scenario (transformed covariates $W$).
    All true parameter values are zero.
    Bias: signed Monte Carlo bias;
    EmpSE: empirical standard error;
    Cov: empirical coverage (\%) of the 95\% Wald confidence interval
    (1{,}000 replications per configuration).}
  \label{tab:dgp2_nonlin_full}
  \small
  \setlength{\tabcolsep}{5pt}
  \begin{tabular}{lr rrr c rrr}\toprule
    & & \multicolumn{3}{c}{ML-EDR (w/o CF)}
    & \phantom{x}
    & \multicolumn{3}{c}{DML-EDR ($K=5$)} \\
    \cmidrule{3-5}\cmidrule{7-9}
    Estimand & $n$ & Bias & EmpSE & Cov(\%)
                  && Bias & EmpSE & Cov(\%) \\\midrule
    \multirow[t]{4}{*}{ATE}
      & 1{,}000 & $ 0.007$ & $0.044$ & $77.8$ && $ 0.005$ & $0.071$ & $94.4$ \\
      & 2{,}000 & $ 0.002$ & $0.031$ & $88.8$ && $ 0.000$ & $0.038$ & $96.0$ \\
      & 4{,}000 & $ 0.003$ & $0.022$ & $92.0$ && $ 0.001$ & $0.024$ & $95.6$ \\
      & 8{,}000 & $ 0.001$ & $0.015$ & $94.1$ && $ 0.000$ & $0.016$ & $96.4$ \\\addlinespace
    \multirow[t]{4}{*}{ATT}
      & 1{,}000 & $ 0.014$ & $0.053$ & $77.9$ && $ 0.000$ & $0.095$ & $93.9$ \\
      & 2{,}000 & $ 0.004$ & $0.037$ & $87.6$ && $-0.002$ & $0.050$ & $94.8$ \\
      & 4{,}000 & $ 0.005$ & $0.027$ & $90.3$ && $ 0.003$ & $0.030$ & $94.6$ \\
      & 8{,}000 & $ 0.001$ & $0.018$ & $93.3$ && $ 0.000$ & $0.019$ & $94.6$ \\\addlinespace
    \multirow[t]{4}{*}{ATU}
      & 1{,}000 & $ 0.000$ & $0.055$ & $76.6$ && $ 0.010$ & $0.093$ & $93.5$ \\
      & 2{,}000 & $ 0.000$ & $0.037$ & $87.6$ && $ 0.001$ & $0.048$ & $95.4$ \\
      & 4{,}000 & $ 0.001$ & $0.026$ & $90.2$ && $ 0.000$ & $0.029$ & $95.6$ \\
      & 8{,}000 & $ 0.000$ & $0.018$ & $93.7$ && $ 0.000$ & $0.019$ & $96.8$ \\\addlinespace
    \multirow[t]{4}{*}{ATO}
      & 1{,}000 & $-0.012$ & $0.050$ & $91.4$ && $ 0.008$ & $0.053$ & $91.7$ \\
      & 2{,}000 & $-0.006$ & $0.030$ & $95.1$ && $ 0.001$ & $0.032$ & $93.6$ \\
      & 4{,}000 & $-0.002$ & $0.021$ & $95.1$ && $ 0.001$ & $0.022$ & $94.8$ \\
      & 8{,}000 & $-0.001$ & $0.015$ & $94.8$ && $ 0.000$ & $0.015$ & $95.2$ \\
    \bottomrule
  \end{tabular}
\end{table}

\clearpage
\section{\texorpdfstring{Additional Information for the Real Data Analysis}{Additional Information for the Real Data Analysis}}\label{sec:add_real}

\begingroup
\footnotesize
\setlength{\tabcolsep}{2pt}
\setlength{\LTcapwidth}{\linewidth}

\begin{longtable}{@{}
  p{3.5cm}
  p{2.5cm}
  >{\raggedleft\arraybackslash}p{2.2cm}
  >{\raggedleft\arraybackslash}p{2.2cm}
  >{\raggedleft\arraybackslash}p{2.2cm}
  >{\raggedleft\arraybackslash}p{0.9cm}
@{}}

\caption{Recorded patient characteristics by RHC status
  (Connors et al., 1996; $n = 5{,}735$).
  Continuous variables: mean (SD);
  categorical and binary variables: $n$ (\%).
  SMD: standardized mean difference (RHC $-$ No~RHC).
  Day-1 variables whose measurement window overlaps
  RHC are shown descriptively but were excluded from the nuisance models;
  variables with more than 50\% missingness were omitted. For descriptive
  tabulation, remaining missing values were replaced by the median
  (continuous) or mode (categorical).}
\label{tab:table1}\\

\toprule
& & \multicolumn{1}{c}{No RHC} & \multicolumn{1}{c}{RHC}
  & \multicolumn{1}{c}{Overall} & \\
Variable & Level / Stat
  & \multicolumn{1}{c}{($n{=}3{,}551$)}
  & \multicolumn{1}{c}{($n{=}2{,}184$)}
  & \multicolumn{1}{c}{($n{=}5{,}735$)}
  & \multicolumn{1}{c}{SMD} \\
\midrule
\endfirsthead

\multicolumn{6}{l}{\textit{Table~\ref{tab:table1} continued}}\\[2pt]
\toprule
& & \multicolumn{1}{c}{No RHC} & \multicolumn{1}{c}{RHC}
  & \multicolumn{1}{c}{Overall} & \\
Variable & Level / Stat
  & \multicolumn{1}{c}{($n{=}3{,}551$)}
  & \multicolumn{1}{c}{($n{=}2{,}184$)}
  & \multicolumn{1}{c}{($n{=}5{,}735$)}
  & \multicolumn{1}{c}{SMD} \\
\midrule
\endhead

\midrule
\multicolumn{6}{r}{\textit{continued on next page}}\\
\endfoot

\bottomrule
\endlastfoot

\multicolumn{6}{l}{\textit{Outcome}} \\[1pt]
30-day mortality & $n$ (\%)
  & 1088 (30.6\%) & 830 (38.0\%) & 1918 (33.4\%) & 0.156 \\[5pt]

\multicolumn{6}{l}{\textit{Phase-1 variables ($V$)}} \\[1pt]
Age group, $n$ (\%)
  & \quad $<55$ & 1162 (32.7\%) & 706 (32.3\%) & 1868 (32.6\%) & $-0.008$ \\
  & \quad 55--70 & 1080 (30.4\%) & 782 (35.8\%) & 1862 (32.5\%) & 0.115 \\
  & \quad $\geq\!70$ & 1309 (36.9\%) & 696 (31.9\%) & 2005 (35.0\%) & $-0.105$ \\[3pt]
Sex, $n$ (\%)
  & \quad Female & 1637 (46.1\%) & 906 (41.5\%) & 2543 (44.3\%) & $-0.093$ \\
  & \quad Male & 1914 (53.9\%) & 1278 (58.5\%) & 3192 (55.7\%) & 0.093 \\[3pt]
Disease cluster, $n$ (\%)
  & \quad ARF/Sepsis/\allowbreak MOSF
    & 2349 (66.2\%) & 1767 (80.9\%) & 4116 (71.8\%) & 0.339 \\
  & \quad Cardiac/Resp & 646 (18.2\%) & 267 (12.2\%) & 913 (15.9\%) & $-0.167$ \\
  & \quad Other & 556 (15.7\%) & 150 (6.9\%) & 706 (12.3\%) & $-0.281$ \\[5pt]

\multicolumn{6}{l}{\textit{Demographics}} \\[1pt]
Age (years) & Mean (SD) & 61.8 (17.3) & 60.8 (15.6) & 61.4 (16.7) & $-0.061$ \\[3pt]
Race, $n$ (\%)
  & \quad White & 2753 (77.5\%) & 1707 (78.2\%) & 4460 (77.8\%) & 0.015 \\
  & \quad Black & 585 (16.5\%) & 335 (15.3\%) & 920 (16.0\%) & $-0.031$ \\
  & \quad Other & 213 (6.0\%) & 142 (6.5\%) & 355 (6.2\%) & 0.021 \\[3pt]
Education (years) & Mean (SD) & 11.6 (3.1) & 11.9 (3.2) & 11.7 (3.2) & 0.091 \\[3pt]
Income, $n$ (\%)
  & \quad $<$\$11k & 2081 (58.6\%) & 1145 (52.4\%) & 3226 (56.3\%) & $-0.125$ \\
  & \quad \$11--25k & 713 (20.1\%) & 452 (20.7\%) & 1165 (20.3\%) & 0.015 \\
  & \quad \$25--50k & 500 (14.1\%) & 393 (18.0\%) & 893 (15.6\%) & 0.107 \\
  & \quad $>$\$50k & 257 (7.2\%) & 194 (8.9\%) & 451 (7.9\%) & 0.060 \\[3pt]
Insurance, $n$ (\%)
  & \quad Private & 967 (27.2\%) & 731 (33.5\%) & 1698 (29.6\%) & 0.136 \\
  & \quad Private+Medicare & 746 (21.0\%) & 490 (22.4\%) & 1236 (21.6\%) & 0.035 \\
  & \quad Medicare & 947 (26.7\%) & 511 (23.4\%) & 1458 (25.4\%) & $-0.076$ \\
  & \quad Medicare+Medicaid & 251 (7.1\%) & 123 (5.6\%) & 374 (6.5\%) & $-0.059$ \\
  & \quad Medicaid & 454 (12.8\%) & 193 (8.8\%) & 647 (11.3\%) & $-0.127$ \\
  & \quad No insurance & 186 (5.2\%) & 136 (6.2\%) & 322 (5.6\%) & 0.043 \\[5pt]

\multicolumn{6}{l}{\textit{Disease characteristics}} \\[1pt]
Primary disease cat., $n$ (\%)
  & \quad MOSF/Sepsis & 527 (14.8\%) & 700 (32.1\%) & 1227 (21.4\%) & 0.415 \\
  & \quad ARF & 1581 (44.5\%) & 909 (41.6\%) & 2490 (43.4\%) & $-0.059$ \\
  & \quad CHF & 247 (7.0\%) & 209 (9.6\%) & 456 (8.0\%) & 0.095 \\
  & \quad COPD & 399 (11.2\%) & 58 (2.7\%) & 457 (8.0\%) & $-0.342$ \\
  & \quad Coma & 341 (9.6\%) & 95 (4.3\%) & 436 (7.6\%) & $-0.207$ \\
  & \quad Cirrhosis & 175 (4.9\%) & 49 (2.2\%) & 224 (3.9\%) & $-0.145$ \\
  & \quad MOSF/Malig. & 241 (6.8\%) & 158 (7.2\%) & 399 (7.0\%) & 0.018 \\
  & \quad Lung cancer & 34 (1.0\%) & 5 (0.2\%) & 39 (0.7\%) & $-0.095$ \\
  & \quad Colon cancer & 6 (0.2\%) & 1 (0.0\%) & 7 (0.1\%) & $-0.038$ \\[3pt]
Cancer status, $n$ (\%)
  & \quad None & 2652 (74.7\%) & 1727 (79.1\%) & 4379 (76.4\%) & 0.104 \\
  & \quad Non-metastatic & 638 (18.0\%) & 334 (15.3\%) & 972 (16.9\%) & $-0.072$ \\
  & \quad Metastatic & 261 (7.4\%) & 123 (5.6\%) & 384 (6.7\%) & $-0.070$ \\[3pt]
DNR order on day 1 & $n$ (\%)
  & 499 (14.1\%) & 155 (7.1\%) & 654 (11.4\%) & $-0.228$ \\[5pt]

\multicolumn{6}{l}{\textit{
  Day-1 severity scores (descriptive only)}} \\[1pt]
SUPPORT 2-mo.\ survival & Mean (SD) & 0.61 (0.19) & 0.57 (0.20) & 0.59 (0.20) & $-0.199$ \\
APACHE~III score & Mean (SD) & 50.9 (18.8) & 60.7 (20.3) & 54.7 (20.0) & 0.506 \\
Glasgow Coma Score & Mean (SD) & 22.2 (31.4) & 19.0 (28.3) & 21.0 (30.3) & $-0.109$ \\[5pt]

\multicolumn{6}{l}{\textit{
  Day-1 vital signs and laboratory values (descriptive only)}} \\[1pt]
Mean arterial BP (mmHg) & Mean (SD) & 84.9 (38.9) & 68.2 (34.2) & 78.5 (38.1) & $-0.448$ \\
Heart rate (bpm) & Mean (SD) & 112.9 (40.9) & 118.9 (41.5) & 115.2 (41.2) & 0.147 \\
Respiratory rate (/min) & Mean (SD) & 29.0 (14.0) & 26.7 (14.2) & 28.1 (14.1) & $-0.166$ \\
Temperature ($^{\circ}$C) & Mean (SD) & 37.6 (1.7) & 37.6 (1.8) & 37.6 (1.8) & $-0.021$ \\
PaO$_2$/FiO$_2$ & Mean (SD) & 240.6 (116.7) & 192.4 (105.5) & 222.3 (115.0) & $-0.428$ \\
WBC ($\times10^{3}/\mu$L) & Mean (SD) & 15.3 (11.4) & 16.3 (12.6) & 15.7 (11.9) & 0.085 \\
Albumin (g/dL) & Mean (SD) & 3.2 (0.7) & 3.0 (0.9) & 3.1 (0.8) & $-0.239$ \\
Hematocrit (\%) & Mean (SD) & 32.7 (8.8) & 30.5 (7.4) & 31.9 (8.4) & $-0.264$ \\
Bilirubin (mg/dL) & Mean (SD) & 2.0 (4.4) & 2.7 (5.3) & 2.3 (4.8) & 0.148 \\
Creatinine (mg/dL) & Mean (SD) & 1.9 (2.0) & 2.5 (2.1) & 2.1 (2.1) & 0.270 \\
Sodium (mEq/L) & Mean (SD) & 137.0 (7.7) & 136.3 (7.6) & 136.8 (7.7) & $-0.092$ \\
Potassium (mEq/L) & Mean (SD) & 4.1 (1.0) & 4.1 (1.0) & 4.1 (1.0) & $-0.027$ \\
PaCO$_2$ (mmHg) & Mean (SD) & 40.0 (14.2) & 36.8 (11.0) & 38.8 (13.2) & $-0.241$ \\
pH & Mean (SD) & 7.39 (0.11) & 7.38 (0.11) & 7.39 (0.11) & $-0.121$ \\
Weight (kg) & Mean (SD) & 65.0 (29.5) & 72.4 (27.7) & 67.8 (29.1) & 0.254 \\[5pt]

\multicolumn{6}{l}{\textit{Comorbid illness history}} \\[1pt]
Cardiac history & $n$ (\%) & 567 (16.0\%) & 446 (20.4\%) & 1013 (17.7\%) & 0.116 \\
Congestive heart failure & $n$ (\%) & 596 (16.8\%) & 425 (19.5\%) & 1021 (17.8\%) & 0.070 \\
Dementia & $n$ (\%) & 413 (11.6\%) & 151 (6.9\%) & 564 (9.8\%) & $-0.163$ \\
Psychiatric & $n$ (\%) & 286 (8.1\%) & 100 (4.6\%) & 386 (6.7\%) & $-0.143$ \\
Chronic pulmonary dis. & $n$ (\%) & 774 (21.8\%) & 315 (14.4\%) & 1089 (19.0\%) & $-0.192$ \\
Renal disease & $n$ (\%) & 149 (4.2\%) & 106 (4.9\%) & 255 (4.4\%) & 0.032 \\
Liver disease & $n$ (\%) & 265 (7.5\%) & 136 (6.2\%) & 401 (7.0\%) & $-0.049$ \\
GI bleeding & $n$ (\%) & 131 (3.7\%) & 54 (2.5\%) & 185 (3.2\%) & $-0.070$ \\
Malignancy & $n$ (\%) & 872 (24.6\%) & 444 (20.3\%) & 1316 (22.9\%) & $-0.101$ \\
Immunosuppression & $n$ (\%) & 907 (25.5\%) & 636 (29.1\%) & 1543 (26.9\%) & 0.080 \\
Transfer from hospital & $n$ (\%) & 335 (9.4\%) & 327 (15.0\%) & 662 (11.5\%) & 0.170 \\
Acute MI history & $n$ (\%) & 105 (3.0\%) & 95 (4.3\%) & 200 (3.5\%) & 0.074 \\[5pt]

\multicolumn{6}{l}{\textit{Admission diagnosis subcategories}} \\[1pt]
Respiratory & $n$ (\%) & 1481 (41.7\%) & 632 (28.9\%) & 2113 (36.8\%) & $-0.270$ \\
Cardiovascular & $n$ (\%) & 1007 (28.4\%) & 924 (42.3\%) & 1931 (33.7\%) & 0.295 \\
Neurological & $n$ (\%) & 575 (16.2\%) & 118 (5.4\%) & 693 (12.1\%) & $-0.353$ \\
Gastrointestinal & $n$ (\%) & 522 (14.7\%) & 420 (19.2\%) & 942 (16.4\%) & 0.121 \\
Renal & $n$ (\%) & 147 (4.1\%) & 148 (6.8\%) & 295 (5.1\%) & 0.116 \\
Metabolic & $n$ (\%) & 172 (4.8\%) & 93 (4.3\%) & 265 (4.6\%) & $-0.028$ \\
Haematological & $n$ (\%) & 239 (6.7\%) & 115 (5.3\%) & 354 (6.2\%) & $-0.062$ \\
Sepsis & $n$ (\%) & 515 (14.5\%) & 516 (23.6\%) & 1031 (18.0\%) & 0.234 \\
Trauma & $n$ (\%) & 18 (0.5\%) & 34 (1.6\%) & 52 (0.9\%) & 0.104 \\
Orthopaedic & $n$ (\%) & 3 (0.1\%) & 4 (0.2\%) & 7 (0.1\%) & 0.027 \\

\end{longtable}
\endgroup

\end{appendices}

\end{document}